\PassOptionsToPackage{hyphens}{url}
\ifdefined\pdfminorversion
\fi
\documentclass[11pt,letterpaper]{article}

\usepackage[T1]{fontenc}
\usepackage[utf8]{inputenc}
\usepackage{lmodern}
\usepackage[margin=1in]{geometry}
\usepackage{microtype}
\UseMicrotypeSet[protrusion]{basicmath}
\usepackage{amsmath,amssymb}
\usepackage{amsthm}
\usepackage{mathtools}
\usepackage{enumitem}
\usepackage{booktabs}
\usepackage{graphicx}
\usepackage{float}
\usepackage{caption}
\usepackage{xcolor}
\usepackage{authblk}

\usepackage{setspace}
\newcommand{\keywords}[1]{\noindent\textit{Keywords:} #1}

\makeatletter
\renewcommand{\paragraph}{\@startsection{paragraph}{4}{\z@}{3.25ex \@plus 1ex \@minus .2ex}{.5em}{\normalfont\normalsize\bfseries}}
\makeatother

\usepackage[authoryear,round]{natbib}
\usepackage{xurl}
\usepackage[unicode]{hyperref}
\usepackage{bookmark}
\newtheorem{theorem}{Theorem}
\newtheorem{proposition}[theorem]{Proposition}
\newtheorem{lemma}{Lemma}
\newtheorem{corollary}[theorem]{Corollary}
\newtheorem{assumption}{Assumption}
\newtheorem{remark}{Remark}[section]

\newcommand{\coinsincludegraphics}[2][]{\includegraphics[#1]{#2}}
\hypersetup{
    pdftitle={COINS: Any-Stage-Valid and Utility-Oriented Sequential Conformal Prediction},
    pdfauthor={Wangcheng Li; Nan Qiao; Xu Guo; Wenguang Sun},
    pdfkeywords={conformal prediction, prediction sequence, any-stage validity, process-level utility, sequential feature acquisition, multiple-score aggregation},
    colorlinks=true,
    linkcolor=blue,
    citecolor=blue,
    urlcolor=magenta,
    pdfborder={0 0 0}
}

\newcommand{\pr}{\mathbb{P}}
\newcommand{\ca}{{\rm \scriptscriptstyle cal}}
\newcommand{\tr}{{\rm \scriptscriptstyle tr}}
\newcommand{\coin}{{\rm \scriptscriptstyle CI}}
\newcommand{\wcoin}{{\rm \scriptscriptstyle WCI}}
\newcommand{\rlcoin}{{\rm \scriptscriptstyle RLCI}}
\newcommand{\ad}{{\rm \scriptscriptstyle ad}}
\newcommand{\bon}{\mathrm{\scriptscriptstyle Bo}}
\newcommand{\ora}{{\rm \scriptscriptstyle ora}}

\usepackage{xspace}

\newcommand{\coinm}{\texttt{COINS}\xspace}
\newcommand{\vocoinm}{\texttt{Vopt-COINS}\xspace}
\newcommand{\rlcoinm}{\texttt{RL-COINS}\xspace}
\newcommand{\E}{\mathbb{E}}
\newcommand{\rmd}{{\rm d}}
\newcommand{\cA}{\mathcal{A}}
\newcommand{\cB}{\mathcal{B}}
\newcommand{\cC}{\mathcal{C}}
\newcommand{\cD}{\mathcal{D}}
\newcommand{\cE}{\mathcal{E}}
\newcommand{\cF}{\mathcal{F}}
\newcommand{\cI}{\mathcal{I}}
\newcommand{\cN}{\mathcal{N}}
\newcommand{\cO}{\mathcal{O}}
\newcommand{\cR}{\mathcal{R}}
\newcommand{\cS}{\mathcal{S}}
\newcommand{\cT}{\mathcal{T}}
\newcommand{\cU}{\mathcal{U}}
\newcommand{\cX}{\mathcal{X}}
\newcommand{\cY}{\mathcal{Y}}
\newcommand{\bc}{\mathbf{c}}
\newcommand{\bq}{\mathbf{q}}
\newcommand{\bs}{\mathbf{s}}
\newcommand{\bS}{\mathbf{S}}
\newcommand{\bu}{\mathbf{u}}
\newcommand{\balp}{\boldsymbol{\alpha}}
\newcommand{\btau}{\boldsymbol{\tau}}
\newcommand{\bvrho}{\boldsymbol{\varrho}}
\newcommand{\bbC}{\mathbb{C}}
\newcommand{\bbI}{\mathbb{I}}
\newcommand{\bbR}{\mathbb{R}}
\newcommand{\hA}{\widehat{A}}
\newcommand{\hcC}{\widehat{\cC}}
\newcommand{\hcU}{\widehat{\cU}}
\newcommand{\hq}{\widehat{q}}
\newcommand{\hbq}{\widehat{\bq}}
\newcommand{\hQ}{\widehat{Q}}
\newcommand{\htau}{\widehat{\tau}}
\newcommand{\hbvrho}{\widehat{\bvrho}}
\newcommand{\hbtau}{\widehat{\btau}}
\newcommand{\tcC}{\widetilde{\cC}}
\newcommand{\tX}{\widetilde{X}}
\newcommand{\ocD}{\overline{\cD}}
\newcommand{\argmin}[1]{\underset{#1}{\mathrm{arg\, min}} \,}
\newcommand{\ole}{\mathord{\le}}
\newcommand{\Lbag}{\mathopen{[\mkern-3mu[}}
\newcommand{\Rbag}{\mathclose{]\mkern-3mu]}}

\newcounter{algorithm}

\newcommand{\bestcell}[1]{\textcolor{red}{\underline{#1}}}

\begin{document}
\title{COINS: Any-Stage-Valid and Utility-Oriented Sequential Conformal Prediction}

\author[1]{Wangcheng Li}
\author[2]{Nan Qiao}
\author[1]{Xu Guo}
\author[3]{Wenguang Sun}
\affil[1]{School of Statistics, Beijing Normal University}
\affil[2]{School of Statistics, Renmin University of China}
\affil[3]{Center for Data Science, Zhejiang University}

\date{}

\setstretch{1}
\maketitle

\begin{abstract}
In clinical diagnosis, sequential screening, and cascaded prediction systems, information is often acquired in stages, with intermediate predictions guiding whether to act on the available evidence or defer action to gather further information.
We study sequential predictive inference in these settings, aiming to construct a sequence of prediction sets that supports timely, resource-efficient decisions while remaining reliable throughout the acquisition process.
Our objective is to maximize process-level utility subject to any-stage validity, which controls the probability of miscoverage anywhere in the prediction sequence.
We develop \coinm, a conformal framework for constructing such sequences, and establish a universality theorem identifying the common exclusion structure underlying all nested, calibration-symmetric, any-stage valid prediction sequences.
Building on \coinm, we develop \vocoinm\ for data-driven optimization of prediction sequences under task-specific utilities based on prediction-set size, stopping time and acquisition costs.
Branchwise and localized extensions further support individualized acquisition and calibration.
Numerical studies demonstrate the effectiveness and practical benefits of the proposed framework for reliable, resource-efficient sequential prediction.
\end{abstract}


\setstretch{1}

\section{Introduction}

Across clinical, scientific, and computational workflows, predictive inference is often intertwined with sequential information acquisition and resource allocation: an initial assessment is based on readily available evidence, and further resources are deployed only when that assessment is not yet sufficiently informative for action \citep{olsson2022estimating,von2025evaluation}.
For example, clinical diagnosis may progress from medical history and routine measurements to laboratory tests and, if substantial ambiguity remains, to an invasive procedure.
In sequential screening and scientific measurement, a low-cost screen may be followed by a targeted confirmatory assay, with equivocal results prompting a more precise but resource-intensive follow-up measurement.
In cascaded computational prediction, a lightweight model handles routine cases, a more capable model evaluates uncertain cases, and human review is reserved for the remaining difficult cases.

Across these settings, sequential acquisition can resolve easier cases sooner and concentrate costly resources on unresolved cases, thereby reducing delay and cost.
We study \emph{sequential predictive inference}, in which predictive uncertainty is updated as information accumulates and an intermediate report is issued at each stage, informing whether to stop and, if not, what information to acquire next.
The resulting inferential object is a \emph{prediction sequence}: an ordered collection of predictive reports for the same test unit that may be inspected and acted upon throughout the information-acquisition process.
We focus on the conformal setting \citep{vovk2005algorithmic,lei2018distribution}, where the stagewise reports are prediction sets.

\subsection{Problem formulation}

Let $X \in \cX \subseteq \bbR^p$ denote a $p$-dimensional
feature vector and let $Y \in \cY$ denote the associated outcome,
which may be a class label or a numerical response.
Consider a new feature--outcome pair $(X_{n+1},Y_{n+1})$,
where $Y_{n+1}$ is unobserved and we seek to predict $Y_{n+1}$ as information about $X_{n+1}$
is progressively revealed over $T$ stages.
Let $\cI_t(X_{n+1})$ denote the information about $X_{n+1}$
available by stage $t \in [T]$, with
$
\cI_1(X_{n+1}) \subseteq \cdots \subseteq \cI_T(X_{n+1}).
$
At stage $t$, the procedure reports a prediction set
$\hcC_t(X_{n+1})$ for $Y_{n+1}$.
We require $\hcC_t(X_{n+1})$ to depend on $X_{n+1}$
only through $\cI_t(X_{n+1})$.
Let $\hcC_{1:T}:=(\hcC_1,\ldots,\hcC_T)$ denote the resulting prediction sequence. We formulate the design of prediction sequences in terms of two
complementary goals: validity across stages and utility for downstream
decision-making.

To ensure validity, we must account for the fact that intermediate reports may be inspected and acted upon throughout the acquisition process. Hence, we
seek to control the \emph{any-stage miscoverage risk} (AMR), defined as
the probability that at least one reported prediction set fails
to contain the true outcome:
\begin{equation*}
    \operatorname{AMR}(\hcC_{1:T})
    :=
    \pr\!\left\{
        \exists\,t\in[T]:
        Y_{n+1}\notin\hcC_t(X_{n+1})
    \right\}.
\end{equation*}
We call a prediction sequence $\hcC_{1:T}$ \emph{any-stage valid}
at level $\alpha\in(0,1)$ if
$\operatorname{AMR}(\hcC_{1:T})\leq\alpha$.
In particular, the coverage guarantee holds at any stopping time
$\tau\in[T]$, including the fixed terminal stage $T$
\citep{darling1967confidence,jennison1989interim}.

To assess the usefulness of a prediction sequence, let
$u(\hcC_{1:T})$ denote a task-specific \emph{process-level utility}.
This utility may reflect (i) the informativeness of the prediction
sets, measured by set size or interval length; (ii) the stage at
which they become sufficiently informative to support a decision;
and (iii) the cumulative costs of measurements, computation, or
expert review.
These components may be considered individually or combined to
reflect the priorities and tradeoffs of the downstream task,
as illustrated in Section~\ref{ssec:utility-design}.
Our statistical objective is to maximize expected process-level
utility subject to any-stage validity.

\subsection{Conformal setup and examples}

Let $\cD^\ca=\{Z_i=(X_i,Y_i)\}_{i=1}^n$ denote a labeled calibration sample, and write $Z_{n+1}=(X_{n+1},Y_{n+1})$ for the test observation. We assume that the calibration and test observations $Z_1,\ldots,Z_n,Z_{n+1}$ are exchangeable.
To construct the prediction sequence, we use stage-specific nonconformity score functions $s_1,\ldots,s_T$, where larger values indicate less compatibility.
Each $s_t(x,y)$ depends on $x$ only through $\cI_t(x)$.
These score functions may be specified a priori or fitted using a training sample $\cD^\tr$ independent of the calibration and test observations.
Thus, $Z_1,\ldots,Z_{n+1}$ remain exchangeable conditional on $s_1,\ldots,s_T$.

We illustrate this setup through \emph{sequential feature acquisition},
where features are acquired in successive blocks, and the current
prediction set guides whether to acquire the next block.
In our application to dermatological diagnosis \citep{guvenir1998learning}, clinical features form the first block, followed by histopathological measurements obtained from a skin biopsy.
Accordingly, $Y$ denotes the patient's disease label, $\cI_1(x)$ comprises the clinical features, and $\cI_2 (x)$ comprises both the clinical and histopathological features.
The calibration sample consists of patients with fully observed features and known diagnoses, and only the information in $\cI_t(X_{n+1})$ is observed for the new patient at stage~$t$.

For each stage $t$, a classifier is fitted on $\cD^\tr$ using the feature blocks available by that stage.
The nonconformity score is one minus the estimated probability assigned to candidate diagnosis $y$ for a patient with features $x$.
These stage-specific score functions are held fixed during calibration.
At stage $t$, each calibration patient is scored at the known
diagnosis, yielding $s_t(X_i,Y_i)$, while the new patient is
scored at each candidate diagnosis $y$, yielding $s_t(X_{n+1},y)$.
Both evaluations use the same stage-specific model and only
the first $t$ feature blocks.
These scores provide the conformal comparisons used to construct
$\hcC_t(X_{n+1})$.

In this application, the objective is to output an informative diagnosis while reducing the need for costly biopsies.
For this task, one may specify the required level of informativeness through a set-size threshold $\lambda_0$ and seek to reach this threshold as early as possible.
We can capture this objective through an \emph{early-resolution utility}.
For example, define
$\tau=\min\{t\in[T]:|\hcC_t(X_{n+1})|\leq\lambda_0\}$ and
$u_{\rm res}(\hcC_{1:T})=(T+1-\tau)/T$, with $\tau=T+1$ if the threshold is never reached.
When $\tau\leq T$, the procedure stops acquiring features and reports $\hcC_\tau(X_{n+1})$; otherwise, it completes all $T$
stages and reports $\hcC_T(X_{n+1})$.
Optimizing this utility subject to any-stage validity,
which preserves the coverage guarantee at stopping time $\tau\wedge T$,
promotes earlier attainment of the required informativeness with fewer measurements.

When all test features and score functions are available in advance
and only the final report matters, utility can instead be defined
solely as the informativeness of the terminal prediction set: $u_{\rm term}(\hcC_{1:T})=-|\hcC_T(X_{n+1})|$.
This gives the special case of \emph{ordered score aggregation},
discussed in Section~\ref{ssec:utility-design} and Appendix~\ref{ssec:numerical-aggregation}.
Different predictive models or score constructions may have
complementary strengths (formalized in Theorem~\ref{the:oracle-recovery}), motivating the use of multiple scores to construct prediction sets
rather than reliance on a single score.
Optimizing this utility subject to any-stage validity, which preserves the coverage guarantee at the terminal stage $T$,
promotes more informative terminal prediction sets through score aggregation.

\subsection{Methodological challenges and contributions}

Our framework takes the entire prediction sequence, rather than a fixed
prediction set, as the inferential object and evaluates it through
any-stage validity and process-level utility. This formulation allows
intermediate reports to guide further information acquisition, stopping,
and downstream action, but raises two coupled methodological challenges.

The first challenge is to control the AMR while making efficient
use of a common error budget across stages.
Calibrating each stage separately at miscoverage level $\alpha$
does not generally provide such control.
Bonferroni correction restores validity by allocating $\alpha$
across stages, but can be conservative because it does not
account for overlap among the stagewise miscoverage events.
The key difficulty is to use earlier calibration results to
reduce this conservativeness while preserving the symmetry required
for conformal validity.

The second challenge is to design prediction sequences to optimize
a task-specific process-level utility.
Stagewise calibration affects the sizes of intermediate and final
prediction sets, thereby influencing stopping and cumulative
acquisition costs.
Allocating more error budget to early stages may
enable earlier stopping, but leaves less budget for later
calibration.

The two challenges are coupled: changing the calibration strategy
to improve utility must not compromise any-stage validity.
This calls for a sequential conformal framework that supports learning calibration strategies jointly across stages to optimize a broad range of process-level utilities.

Our first contribution is a structural theory of any-stage validity
based on a generic sequential exclusion framework.
We identify monotone exclusion, a fixed total exclusion count,
and survivor-symmetric adaptation as sufficient conditions for
finite-sample AMR control
(Section~\ref{ssec:principles-seq-test}).
A universality theorem further shows that every nested,
calibration-symmetric prediction sequence that controls AMR
under every exchangeable distribution admits a representation
within this framework
(Section~\ref{ssec:universality-pp}).
Together, these results provide a general foundation for
constructing valid prediction sequences while leaving flexibility
in how they are refined across stages to serve task-specific
utility.

Our second contribution is \coinm
(\textbf{CO}unt \textbf{IN}vestment across the \textbf{S}equence),
a budget-efficient sequential calibration procedure built on
the general exclusion framework
(Section~\ref{ssec:coinprocedure}).
It allocates a fixed exclusion-count budget across stages,
spending each unit only on a previously unexcluded observation
and thereby avoiding redundant use of the budget.
Under exchangeability, \coinm achieves finite-sample AMR control
and produces prediction sets no larger than their
Bonferroni counterparts at every stage
(Section~\ref{ssec:dom-bonferroni}).

Our third contribution is \vocoinm, a data-driven method for
utility-oriented sequential design within the \coinm family.
It learns the stagewise allocation for a specified process-level
utility while preserving finite-sample any-stage validity (Section~\ref{ssec:vopt-coinm}).
Rather than relying on a prespecified allocation, the method
tailors the prediction sequence to task-specific objectives,
including early resolution, terminal informativeness, and
tradeoffs between predictive precision and acquisition costs
(Section~\ref{ssec:utility-design}).
We establish a nonasymptotic shifted oracle inequality comparing
the learned allocation with the population-optimal allocation
within the \coinm family (Section~\ref{ssec:population-utility-loss}).

Finally, we demonstrate the broader applicability of the framework
through extensions and additional prediction tasks:
Branchwise and localized constructions support covariate-dependent
acquisition and test-specific calibration (Section~\ref{sec:generality}),
and a weighted version of \coinm accommodates covariate shift (Appendix~\ref{asec:weighted-process}).
Moreover, experiments on synthetic and real data provide empirical support for our theoretical guarantees and demonstrate gains from utility-oriented allocation
(Section~\ref{sec:numerical-results} and Appendix~\ref{asec:algorithmic-numerical-details}).

\subsection{Related work}

Our work connects to three lines of research, discussed briefly
below. Appendix~\ref{asec:literaturereview} provides further
discussion and detailed comparisons.

\textbf{Anytime-valid inference.}
Confidence sequences and $e$-processes provide guarantees that
remain valid under continuous monitoring and data-dependent
stopping
\citep{darling1967confidence,jennison1989interim,howard2021time,
ramdas2023game,waudby2024time}.
Classical confidence-sequence applications typically focus on parameter inference as new observations accumulate.
We study a related but distinct prediction problem: providing
simultaneous coverage across stages for a single outcome as
the test unit's features are progressively revealed, using the
same calibration sample throughout.
In this setting, our symmetry-preserving exclusion framework also induces an $e$-process (Section~\ref{ssec:eprocess}).

\textbf{Selection and aggregation for conformal prediction.}
Existing methods select predictive models, aggregate nonconformity
scores, or merge prediction sets to improve the informativeness of
a single prediction set
\citep{yang2025selection,liang2024conformal,qin2024data,
patel2025conformal,alami2026symmetric,wang2026localized}.
For example, \texttt{COLA}  \citep{xu2025aggregating} allocates miscoverage levels across
score-specific prediction sets to minimize the size of their
intersection, and \texttt{SeqSB} \citep{schrab2026aggregation}
is a sequential single-batch aggregation method for
ordered test statistics. Our emphasis is on task-specific design of the entire prediction
sequence, accounting for when intermediate reports become
sufficiently informative and what acquisition costs are incurred.

\textbf{Adaptive testing and conformal selection.}
Interactive methods for multiple testing and conformal selection
use filtrations to restrict the information available for adaptive
ranking and selection while controlling the false discovery rate
\citep{lei2018adapt,ren2023knockoffs,gui2025acs}.
Our framework uses related information restrictions but targets
simultaneous coverage of a single outcome across stages.
The work also connects to Holm's procedure
\citep{holm1979simple} and the sequential rejection principle
\citep{goeman2010sequential}, which use earlier rejections to
improve subsequent decisions.
In contrast, the adaptation in \coinm occurs within the calibration reference
for each candidate outcome: previously excluded augmented
observations no longer contribute to subsequent ranks, and
symmetry converts a fixed total exclusion count into an
AMR guarantee.

\section{A Conformal Framework for Sequential Prediction}

For any finite indexed dataset $\cD=(Z_1,\ldots,Z_m)$, let
$\Lbag\cD\Rbag:=\Lbag Z_1,\ldots,Z_m\Rbag$ denote its associated
\emph{multiset}; that is, the unordered collection of observations in
$\cD$, with repeated values retained. Our theoretical analysis in this section
relies on the following standard conditions for conformal inference.

\begin{assumption}[Conditional exchangeability and no ties]
    \label{cond:exchangeability}
    Let $Z_i=(X_i,Y_i)$ for $i\in[n+1]$.
    Conditional on the stage-specific score functions,
    the augmented sample $(Z_1,\ldots,Z_{n+1})$ is exchangeable; that is,
    for every permutation $\pi$ of $[n+1]$,
    \[
        (Z_1,\ldots,Z_{n+1})
        \overset{d}{=}
        (Z_{\pi(1)},\ldots,Z_{\pi(n+1)})
        \qquad\text{conditionally on } (s_1, \cdots, s_T).
    \]
    Moreover, for every $t\in[T]$, the score values
    $s_t(Z_1),\ldots,s_t(Z_{n+1})$ are almost surely distinct.
\end{assumption}

We develop the theory under the no-ties assumption.
The case with ties can be handled by a randomized tie-breaking
argument.

\subsection{Sequential exclusion and any-stage validity}
\label{ssec:principles-seq-test}

We first construct prediction sequences through test inversion, then establish structural conditions on the sequential exclusion process that ensure finite-sample any-stage validity under
exchangeability.

For each candidate outcome $y\in\cY$, define the augmented sample
$\ocD(y)=\{Z_1^y,\ldots,Z_n^y,Z_{n+1}^y\}$,
where $Z_i^y=Z_i$ for $i\in[n]$ and
$Z_{n+1}^y=(X_{n+1},y)$.
For each stage $t\in[T]$ and augmented observation
$z\in\ocD(y)$, let $\delta_t^y(z)\in\{0,1\}$ denote its
cumulative exclusion indicator:
\[
    \delta_t^y(z)
    :=
    \mathbb{I}\bigl\{
        z \text{ has been excluded from } \ocD(y)
        \text{ by stage } t
    \bigr\}.
\]
By convention, $\delta_0^y(z)\equiv0$.
We refer to
$\boldsymbol{\delta}
\coloneqq(\delta_t^y)_{t\in[T],\,y\in\cY}$
as a \emph{sequential exclusion process}.
Candidate outcome $y$ is rejected by stage $t$ if and only if $Z_{n+1}^y$ is excluded from $\ocD(y)$; i.e.,
$\delta_t^y(Z_{n+1}^y)=1$.
By test inversion, the stage-$t$ prediction set retains all
candidate outcomes not rejected by that stage:
\begin{equation}\label{eq:pred-set-from-seq-test}
    \hcC_t^{\boldsymbol{\delta}}(X_{n+1})
    =
    \left\{
        y\in\cY:
        \delta_t^y\bigl(Z_{n+1}^y\bigr)=0
    \right\},
    \qquad t\in[T].
\end{equation}
We suppress the superscript $\boldsymbol{\delta}$ when the
exclusion process is clear from the context.

At the true outcome $y=Y_{n+1}$, the any-stage miscoverage event is therefore:
\[
    \left\{
        \exists\,t\in[T]:
        Y_{n+1}\notin\hcC_t(X_{n+1})
    \right\}
    =
    \left\{
        \exists\,t\in[T]:
        \delta_t^{Y_{n+1}}(Z_{n+1})=1
    \right\}.
\]
Thus, constructing an any-stage-valid prediction sequence reduces to
designing a sequential exclusion process that controls the probability
that the true test observation is excluded at any stage.
To this end, we introduce three structural conditions, denoted by
\textnormal{(S1)}--\textnormal{(S3)}, that are sufficient for finite-sample AMR control.
Here, ``S'' stands for ``structural.''

\begin{enumerate}[label=(S\arabic*), ref=(S\arabic*)]
    \item\label{p1:nondecreasing}
    (Monotone exclusion)
    The exclusion decisions are \emph{non-decreasing} across stages; that is,
    $\delta_{t-1}^y(z)\le\delta_t^y(z)$ for every $t>1$,
    $y\in\cY$, and $z\in\ocD(y)$.

    \item\label{p2:reje-count}
    (Fixed total exclusion count)
    For each $y \in \cY$,
    the final exclusion indicators satisfy
    \begin{equation*}
        \sum_{i = 1}^{n+1}
        \delta_T^y(Z_i^y)
        =
        \lfloor (n+1) \alpha \rfloor
        =:
        c_{\rm tot}.
    \end{equation*}

    \item\label{p3:per-inv}
    (Survivor-symmetric adaptation)
    For each $y\in\cY$, initialize
    $\cF_0^y=\sigma\bigl(\Lbag\ocD(y)\Rbag,s_1\bigr)$ and
    $\cR_0^y=\varnothing$. For each $t\in[T]$, define the cumulative
    exclusion set by
    $
        \cR_t^y
        \coloneqq
        \{i\in[n+1]:\delta_t^y(Z_i^y)=1\},
    $
    and recursively update the filtration by
    \[
        \cF_t^y = \cF_{t-1}^y \vee
        \sigma\Big(
            \cR_t^y,
            \Lbag
                Z_i^y:
                i \in \cR_t^y
            \Rbag,
            s_{t + 1}
        \Big),
        \qquad t\in[T],
    \]
    where $\vee$ denotes the smallest $\sigma$-field containing both
    operands, and $s_{t+1}$ is omitted when
    $t=T$. For every $t\in[T]$ and $y\in\cY$,
    $\delta_t^y(z)$ is $\cF_{t-1}^y$-measurable for each
    $z\in\ocD(y)$, and the same value-based exclusion rule is used for all
    $y\in\cY$.
\end{enumerate}

Condition~\ref{p1:nondecreasing} makes exclusions irreversible
and hence produces a nested prediction sequence.
Any miscoverage at an intermediate stage therefore persists
through stage $T$, so
\begin{equation}\label{eq:equiv-two-events}
    \{\exists\,t\in[T]:Y_{n+1}\notin\hcC_t(X_{n+1})\}
    =
    \{Y_{n+1}\notin\hcC_T(X_{n+1})\}.
\end{equation}
This identity simplifies the analysis of any-stage validity.

Condition~\ref{p2:reje-count} fixes the total number of
exclusions by stage $T$ at $c_{\rm tot}=\lfloor(n+1)\alpha\rfloor$,
without prescribing a stagewise allocation.

Condition~\ref{p3:per-inv} allows later exclusions to adapt to
earlier decisions while preserving symmetry among the surviving
observations.
The filtration records the excluded index set and the multiset
of excluded observations separately.
Together with the augmented multiset, these determine the
surviving index set and the multiset of surviving values,
without revealing which observation value corresponds to which
surviving index. Standard conformal calibration relies on treating the augmented
observations symmetrically
\citep{vovk2005algorithmic,angelopoulos2024theoretical}.
In the sequential setting, this symmetry must be maintained among
the surviving observations conditional on the accumulated
exclusion history.

The following theorem establishes a general construction principle:
under exchangeability, any sequential exclusion process satisfying
Conditions~\ref{p1:nondecreasing}--\ref{p3:per-inv} yields a prediction
sequence with finite-sample any-stage validity.
\begin{theorem}[Finite-sample any-stage validity]\label{the:basicresult}
    Under Assumption~\ref{cond:exchangeability}, suppose a sequential
    exclusion process $\boldsymbol{\delta}$ satisfies
    Conditions~\ref{p1:nondecreasing}--\ref{p3:per-inv} almost surely. Then
    the prediction sequence $\hcC_{1:T}^{\boldsymbol{\delta}}$ induced
    by~\eqref{eq:pred-set-from-seq-test} satisfies
\vskip -1cm
    \[
    \operatorname{AMR}(\hcC_{1:T}^{\boldsymbol{\delta}})
    =\pr \big\{
        \exists \, t \in [T],
        Y_{n+1} \notin \hcC_t^{\boldsymbol{\delta}}(X_{n+1})
    \big\} = \frac{\lfloor (n+1) \alpha \rfloor}{n+1}\leq\alpha.
    \]
\end{theorem}

\subsection{Universality of sequential exclusion}
\label{ssec:universality-pp}

We adopt sequential exclusion under
Conditions~\ref{p1:nondecreasing}--\ref{p3:per-inv} as a basic design
strategy. A natural question is whether this strategy is restrictive:
could other approaches achieve the same validity guarantee while
producing prediction sequences that cannot be obtained through
sequential exclusion?

The following theorem establishes a universality property of sequential exclusion: every nested, calibration-symmetric prediction sequence with the stated AMR guarantee admits a sequential exclusion representation that reproduces its prediction set at every stage.
We include this universality result here for theoretical completeness. Readers primarily interested in methodology may proceed directly to Section~\ref{ssec:coinprocedure}.

\begin{theorem}[Sequential universality principle]
    \label{the:universality-principles}
    Let $(\cC_t)_{t=1}^T$, where
    $\cC_t=\cC_t(x;\cD)\subseteq\cY$, be a sequence of measurable
    prediction rules satisfying the following conditions:
    \begin{enumerate}[label=(\roman*)]
        \item \emph{Calibration symmetry.}
        For all $t\in[T]$ and $x\in\cX$,
        $\cC_t(x;\cD)=\cC_t(x;\cD')
            \text{ if }
            \Lbag\cD\Rbag=\Lbag\cD'\Rbag.
        $

        \item \emph{Nestedness.}
        For every $t=2,\ldots,T$, $x\in\cX$ and $\cD$,
        $
            \cC_t(x;\cD)\subseteq\cC_{t-1}(x;\cD).
        $

        \item \emph{Exact AMR control.}
        For any exchangeable observations $(Z_1,\ldots,Z_{n+1})$
        that are pairwise distinct almost surely,
    \vskip -1.2cm
        \[
            \pr\left\{
                \exists\,t\in[T]:
                Y_{n+1}\notin
                \cC_t(X_{n+1};\cD^\ca)
            \right\}
            =
            \frac{\lfloor(n+1)\alpha\rfloor}{n+1}.
        \]
    \end{enumerate}
    \vskip -0.5cm

    Then there exists a sequential exclusion process
    $\boldsymbol{\delta}$ whose induced prediction sequence satisfies
    \vskip -1.2cm
    \[
        \hcC_t^{\boldsymbol{\delta}}(X_{n+1})
        =\cC_t(X_{n+1};\cD^\ca),
        \qquad t\in[T].
    \]
 \end{theorem}

This theorem provides a converse to Theorem~\ref{the:basicresult}
and establishes the completeness of the sequential exclusion
framework for the class of nested, calibration-symmetric
prediction sequences. It generalizes classical full conformal universality
(\citealp[Section~2.4]{vovk2005algorithmic};
\citealp[Section~9.3]{angelopoulos2024theoretical}) from a single
prediction set under a marginal coverage constraint to a nested
prediction sequence under an AMR constraint.

This generalization is critical in sequential settings, where intermediate
prediction sets guide stopping and further data acquisition. For
example, if a first-stage set already meets the informativeness
requirement, reporting it immediately rather than at a later stage
can avoid additional acquisition costs and delay, while yielding the
same terminal set and AMR. The theorem thus establishes universality for the full prediction
sequence, including the intermediate reports needed to evaluate
stopping behavior and resource use.

The structural framework leaves the stagewise exclusion rules open
to design. Section~\ref{ssec:coinprocedure} develops \coinm, a family
of sequential conformal procedures indexed by prespecified
stagewise exclusion counts and satisfying
Conditions~\ref{p1:nondecreasing}--\ref{p3:per-inv}.
Section~\ref{sec:opt-via-erm} then studies how to allocate the total
exclusion count across stages within this family to optimize
process-level utility.

\begin{remark}[Relaxed AMR requirement]
    If the exact-AMR requirement is replaced by
    $\operatorname{AMR}\leq\alpha$, the same
    representation holds with Condition~\ref{p2:reje-count} replaced by
    $
        \sum_{i=1}^{n+1}\delta_T^y(Z_i^y)
        \leq\lfloor(n+1)\alpha\rfloor,
    $
    while Conditions~\ref{p1:nondecreasing} and~\ref{p3:per-inv} remain
    unchanged.
    This conservative version no longer requires observations to be pairwise distinct.
\end{remark}

\begin{remark}[Nested refinement under AMR]
Theorem~\ref{the:universality-principles} focuses on nested
prediction sequences. Any prediction sequence $\cC_{1:T}$,
possibly non-nested, can be refined to the nested sequence
$\cC_t^\cap=\bigcap_{r\leq t}\cC_r$.
This refinement satisfies
\[
    \operatorname{AMR}(\cC_{1:T}^\cap)
    =\operatorname{AMR}(\cC_{1:T}),
    \qquad
    \cC_t^\cap\subseteq\cC_t,
\]
and preserves stagewise information restrictions and calibration
symmetry. Thus, under these conditions, restricting
attention to nested sequences entails no loss in utility,
provided that utility cannot decrease when each report is
replaced by a subset.
\end{remark}

\subsection{The \coinm procedure}
\label{ssec:coinprocedure}

We begin with a prespecified allocation of exclusion counts
$(c_1,\ldots,c_T)\in\mathbb{Z}_{\geq0}^T$ satisfying
$\sum_{t=1}^T c_t=c_{\rm tot}$, and construct the \coinm procedure
to satisfy Conditions~\ref{p1:nondecreasing}--\ref{p3:per-inv}.

For each candidate outcome $y\in\cY$, let
$\delta_t^{y,\coin}(Z_i^y)$ denote the cumulative exclusion indicator
of augmented observation $Z_i^y$ by stage $t$, initialized as
$\delta_0^{y,\coin}(Z_i^y)=0$ for all $i\in[n+1]$.
At stage $t$, \coinm excludes the $c_t$ observations with the largest
stage-specific nonconformity scores among those that have survived
all preceding stages. Specifically, define the index set of observations newly excluded
at stage $t$ by
\begin{equation}\label{eq:coin-selection-set}
    \cS_t^{y,\coin}
    \coloneqq
    \bigg\{
        i\in[n+1]:
        \delta_{t-1}^{y,\coin}(Z_i^y)=0,\,
        \sum_{j=1}^{n+1}
        \bbI\Big\{
            \delta_{t-1}^{y,\coin}(Z_j^y)=0,\,
            s_t(Z_j^y)\geq s_t(Z_i^y)
        \Big\}
        \leq c_t
    \bigg\}.
\end{equation}
Update the cumulative exclusion indicators as
\begin{equation}\label{eq:coinupdatestrategy}
    \delta_t^{y,\coin}(Z_i^y)
    \coloneqq
    \max\left\{
        \delta_{t-1}^{y,\coin}(Z_i^y),
        \bbI\{i\in\cS_t^{y,\coin}\}
    \right\},
    \qquad i\in[n+1].
\end{equation}
By test inversion, the resulting \coinm prediction sets are
\begin{equation}\label{eq:coin-prediction-set-seq}
    \hcC_t^\coin(X_{n+1})
    =
    \big\{
        y\in\cY:
        \delta_t^{y,\coin}(Z_{n+1}^y)=0
    \big\},
    \qquad t\in[T].
\end{equation}
Figure~\ref{fig:demo-coin} in Appendix~\ref{ssec:implementcoin}
illustrates the construction.

Next, we verify the three structural conditions needed for the
AMR guarantee. First, the cumulative update
in~\eqref{eq:coinupdatestrategy} makes exclusions irreversible,
establishing Condition~\ref{p1:nondecreasing}.
Second, at the true candidate $y=Y_{n+1}$,
Assumption~\ref{cond:exchangeability} ensures that the scores are
almost surely distinct. By induction, the number of survivors before
stage $t$ is $n+1-\sum_{r<t}c_r\geq c_t$, so
$|\cS_t^{Y_{n+1},\coin}|=c_t$.
Since each stage excludes only survivors, the stagewise exclusion
sets are disjoint. Therefore, Condition~\ref{p2:reje-count} holds at the true outcome by noting that
\[
    \sum_{i=1}^{n+1}\delta_T^{Y_{n+1},\coin}(Z_i)
    =\sum_{t=1}^T|\cS_t^{Y_{n+1},\coin}|
    =\sum_{t=1}^T c_t
    =c_{\rm tot}.
\]
Third, $\cF_{t-1}^y$ contains $s_t$, the augmented multiset, and the
multiset of previously excluded observations, and hence determines
the multiset of survivors. For a fixed observation value $z$, its
rank among the survivors and its previous exclusion status are
therefore $\cF_{t-1}^y$-measurable. The update applies the same
value-based rule to all survivors and all candidate outcomes,
establishing Condition~\ref{p3:per-inv}.
The next corollary follows directly from Theorem~\ref{the:basicresult}.

\begin{corollary}[Finite-sample any-stage validity of \coinm]
    \label{the:errorcontrolcoin}
    Under Assumption~\ref{cond:exchangeability}, for any prespecified
    allocation $(c_1,\ldots,c_T)\in\mathbb{Z}_{\geq0}^T$ satisfying
    $\sum_{t=1}^T c_t=\lfloor(n+1)\alpha\rfloor$, the prediction
    sequence produced by \coinm satisfies
    \[
        \pr\Big\{
            \exists\,t\in[T]:
            Y_{n+1}\notin\hcC_t^\coin(X_{n+1})
        \Big\}
        =\frac{\lfloor(n+1)\alpha\rfloor}{n+1}.
    \]
\end{corollary}

\paragraph*{Efficient implementation via reusable thresholds.}
A direct implementation of~\eqref{eq:coin-prediction-set-seq}
repeats the augmented exclusion process for each candidate outcome $y \in \cY$.
We provide an equivalent threshold representation that
removes this candidate-by-candidate computation.
Specifically, without score ties, whenever the test observation survives stage $t$,
the $c_t$ newly excluded observations are precisely the highest-scoring
calibration observations remaining before that stage.
Thus, for fixed calibration data, scores, and allocation, the stagewise
thresholds can be computed once and reused across candidate outcomes
and test points. Prediction then reduces to forming and intersecting
score-threshold sets. Appendix~\ref{ssec:implementcoin} gives
Algorithm~\ref{alg:coinprocedure} and its justification.
We use this threshold implementation for computation and allocation
learning, retaining the notation $\hcC_t^\coin$.

\subsection{Comparison with the Bonferroni method}
\label{ssec:dom-bonferroni}

The Bonferroni method provides a natural baseline for AMR control.
For a prespecified nonnegative allocation
$\balp=(\alpha_1,\ldots,\alpha_T)$ satisfying
$\sum_{t=1}^T\alpha_t=\alpha$, it forms a
level-$(1-\alpha_t)$ conformal prediction set at stage $t$,
using the score function $s_t$ and the full calibration sample.
It then reports the intersection of these sets through stage $t$,
denoted by $\hcC_t^\bon(X_{n+1};\balp)$.
For nested prediction sequences, terminal coverage is equivalent to AMR control.
Therefore, by the union bound, the Bonferroni method guarantees AMR control for $\hcC_{1:T}^\bon$.
However, this bound can be conservative when miscoverage events overlap, potentially leading to a loss in process-level utility.

Specifically, consider a process-level utility $u(\cC_{1:T})$, with larger values indicating better
performance.
We call $u$ \emph{set-wise monotone} if replacing every reported
set by a subset cannot decrease utility:
\begin{equation}\label{eq:monotone-utility}
    \bigl\{\cC_t\subseteq\cC_t'
        \text{ for every }t\in[T]\bigr\}
    \quad\Longrightarrow\quad
    u(\cC_{1:T})\geq u(\cC_{1:T}').
\end{equation}
This condition accommodates both prediction-set informativeness
and early resolution.
The following proposition gives a
comparison with Bonferroni for a class of such utilities.

\begin{proposition}[Dominance over Bonferroni]
    \label{prop:prediction-set-dominance}
    Suppose that there are no score ties, and that both
    constructions use the same calibration sample and score sequence.
    For any prespecified Bonferroni allocation $\balp$, define
    $\bc=(c_1,\ldots,c_T)$ by
    $c_t=\lfloor(n+1)\alpha_t\rfloor$ for $t<T$ and
    $c_T=c_{\rm tot}-\sum_{t<T}c_t$.
    Then
    $
        \hcC_t^\coin(X_{n+1};\bc)
        \subseteq
        \hcC_t^\bon(X_{n+1};\balp)
    $ for all $t \in [T]$, and thus,
    for every set-wise monotone utility $u$,
    \[
        u\bigl(\hcC_{1:T}^\coin(X_{n+1};\bc)\bigr)
        \geq
        u\bigl(\hcC_{1:T}^\bon(X_{n+1};\balp)\bigr).
    \]
\end{proposition}

Two special cases illustrate the implications of this dominance result:
\begin{enumerate}[label=(\roman*)]
    \item \emph{Informativeness.}
    For any fixed stage $t$, take
    $u(\cC_{1:T})=-\nu(\cC_t)$, where $\nu$ denotes cardinality
    for classification or Lebesgue measure for regression.
    The proposition implies that \coinm produces prediction
    sets no larger than those of its matched Bonferroni
    counterpart at every stage, including the terminal stage.

    \item \emph{Early resolution.}
    For a prespecified set-size threshold $\lambda_0$, define
    \[
        \tau(\cC_{1:T})
        =\min\{t\in[T]:\nu(\cC_t)\leq\lambda_0\},
        \qquad
        u_{\rm res}(\cC_{1:T})
        =\frac{T+1-\tau(\cC_{1:T})}{T},
    \]
    with $\tau(\cC_{1:T})=T+1$ if the threshold is never reached.
    The proposition implies that \coinm reaches this stopping
    criterion no later than its matched Bonferroni counterpart.
\end{enumerate}

\subsection{Connection to \texorpdfstring{$e$}{e}-processes}
\label{ssec:eprocess}

The AMR guarantee ensures coverage at data-dependent stopping times within the prescribed stages. This parallels classical confidence sequences \citep{darling1967confidence}, which can be constructed by inverting sequential tests based on e-processes; see \citet{ramdas2023game} for a review.
Their approaches are commonly used to infer an unknown parameter
as observations accumulate; our setting instead concerns
prediction of a single outcome across stages using a shared
calibration sample.

We now show that sequential exclusion also induces an $e$-process
in our prediction setting.
Symmetry among surviving augmented observations allows us
to construct this process from the cumulative exclusion
count and the survival status of the candidate test observation.
Moreover, at the specified level $\alpha$, inverting this process
recovers the original prediction set at every stage,
providing an interpretation of the AMR guarantee in terms
of anytime-valid evidence.
Specifically, for each candidate $y\in\cY$, define
\begin{equation}
    \label{eq:e-process}
    E_t^y
    \coloneqq
    \begin{cases}
        \displaystyle 1/{\alpha},
        & n+1\in\cR_t^y,
        \\[8pt]
        \displaystyle
        \frac{\alpha(n+1)-|\cR_t^y|}
        {\alpha\bigl(n+1-|\cR_t^y|\bigr)},
        & n+1\notin\cR_t^y,
    \end{cases}
    \qquad t=0,\ldots,T.
\end{equation}

The process starts at $E_0^y=1$ and satisfies the following guarantee.

\begin{proposition}[Survivor-count $e$-process]
    \label{the:eprocess-val}
    Under Assumption~\ref{cond:exchangeability}, suppose that a sequential
    exclusion process $\boldsymbol{\delta}$ satisfies
    Conditions~\ref{p1:nondecreasing}--\ref{p3:per-inv}. Then
    $\{E_t^{Y_{n+1}}\}_{t=0}^T$ is adapted to
    $\{\cF_t^{Y_{n+1}}\}_{t=0}^T$ and, for every stopping time
    $\tau\in\{0,\ldots,T\}$ with respect to this filtration,
    $
        \E\bigl[E_\tau^{Y_{n+1}}\bigr]\leq 1.
    $
    Hence, $\{E_t^{Y_{n+1}}\}_{t=0}^T$ is an $e$-process under the
    exchangeability null.
\end{proposition}

Candidate-wise inversion at the threshold $\alpha^{-1}$ recovers the
original prediction sequence:
\[
    \hcC_t^{\scriptscriptstyle\rm E}(X_{n+1})
    :=\bigl\{y\in\cY:E_t^y<\alpha^{-1}\bigr\}
    =\hcC_t^{\boldsymbol{\delta}}(X_{n+1}),
    \qquad t\in[T].
\]
Consequently, the $e$-process guarantee recovers AMR control
for the original prediction sequence.
The proof of Proposition~\ref{the:eprocess-val} is given in
Appendix~\ref{assec:proof-eprocess-val}.

\section{Utility-Oriented Sequential Design}\label{sec:opt-via-erm}

\coinm defines a family of any-stage-valid prediction sequences
indexed by stagewise allocations of exclusion counts.
We study how to choose an allocation that optimizes
task-specific process-level utility.
Section~\ref{ssec:utility-design} presents utility formulations
for different sequential prediction tasks.
Section~\ref{ssec:vopt-coinm} develops a sample-splitting
procedure for learning the allocation while preserving
finite-sample any-stage validity.
Section~\ref{ssec:population-utility-loss} introduces an
oracle benchmark and bounds the utility gap under vanishing
perturbations of the stagewise budgets.

\subsection{Task-specific utility design}\label{ssec:utility-design}

Sequential prediction involves task-specific tradeoffs between
the informativeness of successive reports, acquisition costs, and delay.
These tradeoffs make allocation design consequential:
an allocation that yields a small terminal set need not
provide useful intermediate reports or support early stopping.
Conversely, allocating more exclusions to early stages can
promote early resolution but leaves fewer exclusions available
for later stages, potentially limiting the gains in terminal
informativeness from additional measurements.
Task-specific utility design aligns the refinement of prediction
sequences with their intended downstream use.
This section presents three utility formulations that encode
different practical objectives.
Throughout, $\nu(\cC)$ measures the size of a prediction set $\cC$:
its cardinality in classification and its Lebesgue measure
in regression.

\paragraph*{1. Early resolution.}
An informative prediction obtained early can mitigate the adverse
consequences of delayed decisions and reduce the cost of further
data acquisition. Potential applications include (a) narrowing a diagnostic
prediction set using routine clinical information before
considering an invasive examination, and (b) obtaining a
sufficiently narrow prediction interval for an object's size
from initial images before acquiring costly high-resolution scans.
These tasks motivate stopping once a prespecified informativeness
requirement is met.
Let $\lambda_0$ denote a prespecified set-size threshold
for supporting a downstream decision.
We define the resolution stage and early-resolution utility by
\begin{equation}\label{eq:utility-early-resolution-ori}
    \tau(\cC_{1:T})
    =\min\{t\in[T]:\nu(\cC_t)\leq\lambda_0\},
    \qquad
    u_{\rm res}(\cC_{1:T})
    =\frac{T+1-\tau(\cC_{1:T})}{T},
\end{equation}
with $\tau(\cC_{1:T})=T+1$ if the requirement is never met.
Maximizing this utility favors allocations that achieve
the required informativeness earlier.
We also consider a cost-aware variant of this utility that accounts
for unequal stagewise acquisition costs; see Section~\ref{ssec:sfa-realdata}.

\paragraph*{2. Terminal informativeness.}
In offline prediction tasks, such as analyzing previously
collected images or experimental data, measurements are already
available and intermediate reports may have little practical value.
The priority is then the informativeness of the final prediction,
motivating the terminal utility
$u_{\rm term}(\cC_{1:T})=-\nu(\cC_T)$.
An important application of this utility is
\emph{ordered score aggregation}, where
all test features are available
from the outset. The goal is to combine the
complementary strengths of different score functions into
an informative terminal prediction set.
In this setting, the scores are applied in a prespecified order, with each stage
using its allocated exclusion count to refine the set retained
from preceding stages. By nestedness,
$\operatorname{AMR}(\hcC_{1:T}) = \pr\{Y_{n+1}\notin\hcC_T(X_{n+1})\}$.
Thus, maximizing expected terminal utility subject to the AMR
constraint is equivalent to minimizing the expected size of
the aggregated set subject to its coverage constraint.
Appendix~\ref{ssec:numerical-aggregation} provides the theoretical
characterization, an illustrative example, and numerical results.

\paragraph*{3. Balancing informativeness, delay, and acquisition costs.}
In staged diagnostic assessments, early reports can guide initial
decisions, while further measurements may improve precision
but incur additional costs and delay.
To capture these competing demands, one may consider the combined utility
\[
    u(\cC_{1:T})
    =
    \sum_{t=1}^T
    \Big[
        a_t\,\bbI\{\nu(\cC_t)\leq\lambda_0\}
        -
        b_t\,\nu(\cC_t)
    \Big],
\]
where $a_t,b_t\geq0$ are prespecified weights.
Intuitively, the first term rewards meeting the informativeness requirement,
while the second penalizes larger prediction sets; their weights
specify the relative importance of competing objectives across stages.
This formulation includes both preceding utilities.
For nested prediction sets, setting $a_t=1/T$ and $b_t=0$
for $t\in[T]$ recovers $u_{\rm res}$.
Setting $a_t=0$ for $t\in[T]$, $b_T=1$, and $b_t=0$
for $t<T$ recovers $u_{\rm term}$.
Moreover, setting $a_t$ to the cost of the next acquisition
($t<T$) and $a_T=0$ makes the reward term the cost saved by stopping
at resolution.

\subsection{\vocoinm: utility-oriented allocation learning}\label{ssec:vopt-coinm}

The utility formulations above specify which prediction sequences
are desirable, but the allocation that maximizes expected utility
depends on the unknown data distribution. We therefore seek to
learn the allocation from data while retaining finite-sample
any-stage validity. The challenge is that the guarantee for a
prespecified allocation does not automatically extend to an
allocation selected using the same data for final calibration.

We address this challenge with \vocoinm, which combines \coinm
with allocation learning by \emph{empirical risk minimization}
(ERM) and uses sample splitting to preserve validity.
Here, \texttt{V} denotes validity and \texttt{Opt} denotes utility optimization.
To this end, we fix a task-specific utility $u$ and randomly partition
$\cD^\ca$ into a strategy-learning sample $\cD^\ca_1$ and a
final-calibration sample $\cD^\ca_2$, with sample sizes
$n_1$ and $n_2$ satisfying $n_1+n_2=n$.
For simplicity, we write
$\cD^\ca_1=\{(X_i,Y_i)\}_{i=1}^{n_1}$ and
$\cD^\ca_2=\{(X_i,Y_i)\}_{i=n_1+1}^{n}$.
\vocoinm then uses $\cD^\ca_1$ to learn a stagewise proportion vector
and $\cD^\ca_2$ to calibrate the prediction sequence.

Specifically, we parameterize an allocation by stagewise proportions:
$
    \bvrho=(\varrho_1,\ldots,\varrho_T)\in\Theta_{\bvrho}
$
with
$ \Theta_{\bvrho}
    \coloneqq
    \{
        \mathbf v\in\bbR_{\scriptscriptstyle\geq 0}^T:
        \sum_{t=1}^T v_t=1
    \}$.
Here, $\varrho_t$ is the target share of the total exclusion budget
assigned to stage $t$, before rounding to integer counts.
Let $\beta_t(\bvrho)\coloneqq\sum_{r=1}^t\varrho_r$
be the cumulative proportion through stage $t\in[T]$,
and set $\beta_0(\bvrho)\coloneqq0$.
For a calibration sample of size $m$, $\bvrho$ defines an exclusion-count allocation
$
    \bc(\bvrho;m)
    =
    \big(c_1(\bvrho;m),\ldots,c_T(\bvrho;m)\big)
$
by
\begin{equation*}
    c_t(\bvrho;m)
    \coloneqq
    \left\lfloor(m+1)\alpha\beta_t(\bvrho)\right\rfloor
    -
    \left\lfloor(m+1)\alpha\beta_{t-1}(\bvrho)\right\rfloor,
    \qquad t\in[T].
\end{equation*}
Summing across stages gives
$ \sum_{t=1}^T c_t(\bvrho;m) = \lfloor(m+1)\alpha\rfloor$.

Let $\hcC_t^\coin(x;\cD,\bc)$ denote the stage-$t$ prediction set
produced by \coinm using calibration data $\cD$ and allocation $\bc$.
Then, the \emph{ERM allocation} learned from $\cD^\ca_1$ is defined by
\begin{equation}\label{eq:learn-allocation-d1}
    \hbvrho^{\rm s} \in
    \argmin{\bvrho \in \Theta_{\bvrho}}
    \frac{1}{n_1}
    \sum_{i=1}^{n_1}
    -u\left(
        \hcC_{1:T}^\coin(X_i;\cD^\ca_1,\bc(\bvrho;n_1))
    \right).
\end{equation}
Finally, we recompute the counts for sample size $n_2$ and report
$\{\hcC_t^\coin(X_{n+1};\cD^\ca_2, \bc(\hbvrho^{\rm s};n_2))\}_{t=1}^T$.
Conditional on $\cD^\ca_1$ and the independently trained score
functions, the ERM allocation is fixed, and the final-calibration
observations and test observation remain exchangeable.
The next corollary follows from Corollary~\ref{the:errorcontrolcoin}, applied with calibration sample size $n_2$.

\begin{corollary}[Finite-sample validity]
    \label{cor:validity-voptcoins}
    Under Assumption~\ref{cond:exchangeability}, \vocoinm satisfies
    \begin{equation*}
        \pr\Bigl\{
            \exists\,t\in[T]:
            Y_{n+1}\notin
            \hcC_t^\coin
            \bigl(X_{n+1};\cD^\ca_2,\bc(\hbvrho^{\rm s};n_2)\bigr)
        \Bigr\}
        =
        \frac{\lfloor (n_2 + 1) \alpha \rfloor}{n_2 + 1}.
    \end{equation*}
\end{corollary}

\begin{remark}
    Sample splitting preserves validity but reduces the data available
    for allocation learning and final calibration and introduces
    variability across splits. Related work on conformal selection
    and aggregation addresses data reuse through full conformalization
    or explicit adjustment for data-dependent selection
    \citep{yang2025selection,xu2025aggregating,liang2024conformal}.
    Developing data-reuse strategies with theoretical guarantees for sequential allocation learning
    is left for future work.
\end{remark}

\subsection{Oracle benchmark and utility guarantee}
\label{ssec:population-utility-loss}

To assess the utility loss from learning an allocation from a finite
sample, we compare the ERM allocation with an oracle that knows the data
distribution and chooses the allocation with the largest expected utility.
The oracle uses the same score functions and
sequential exclusion rule as \coinm, but calibrates its thresholds using
the true distribution.

To formalize this comparison, we first state the sampling assumption and define population utility for threshold-based prediction sequences.
Throughout this subsection, we impose the following
assumption that is standard in ERM analysis \citep{vapnik2013nature}.

\begin{assumption}
    \label{cond:independent}
    $\{(X_i,Y_i)\}_{i=1}^{n+1}$ are i.i.d. and are independent of the score functions
    $\{s_t\}_{t=1}^T$.
\end{assumption}

For the given score functions, define the \emph{population utility}
of a threshold vector $\btau\in\overline{\bbR}^T$ by
\[
    \cU(\btau)
    \coloneqq
    \E_{X\sim P_X}
    \left[
        u\left(\cC_{1:T}(X;\btau)\right)
    \right],
\]
where, for $t \in [T]$,
$\cC_t(x;\btau)\coloneqq
\bigcap_{r=1}^t\{y\in\cY:s_r(x,y)\leq\tau_r\}$ is the prediction set induced by the threshold vector at stage~$t$.
Notably, the ERM criterion in~\eqref{eq:learn-allocation-d1}
maximizes the corresponding empirical utility.

We now define the population analogue of the empirical calibration
procedure used by \coinm.
Let $S_t=s_t(X,Y)$ for $t\in[T]$ and $(X,Y)\sim P_{XY}$.
For a nonnegative stagewise budget vector
$\balp=(\alpha_1,\ldots,\alpha_T)$ with
$\sum_{t=1}^T\alpha_t<1$, we define the oracle thresholds as the following
\emph{population nested quantiles}:
\[
    q_t(\balp)
    \coloneqq
    \begin{cases}
        +\infty,
        & \alpha_t=0,
        \\
        \inf\Big\{
            v\in\overline{\bbR}:
            \pr_{P_{XY}}\big(
                S_t\leq v,
                S_r\leq q_r(\balp),\ \forall\, r<t
            \big)
            \geq
            1-\sum_{r=1}^t\alpha_r
        \Big\},
        & \alpha_t>0.
    \end{cases}
\]
Write $\bq(\balp)=(q_1(\balp),\ldots,q_T(\balp))$.
Under the no-ties assumption, this definition recursively removes probability mass
$\alpha_t$, measured under $P_{XY}$, from the observations that satisfy
all preceding thresholds. Thus, the probability of satisfying the first
$t$ thresholds is $1-\sum_{r=1}^t\alpha_r$.
When $\alpha_t=0$, setting $q_t=+\infty$ skips exclusion at that stage.
Appendix~\ref{asec:nested-quantile-results} studies the finite-sample behavior of empirical nested quantiles.

Then, for each $\bvrho\in\Theta_{\bvrho}$ and $t \in [T]$,  calibration under the true distribution gives
the \emph{population prediction set}
$
    \cC_t^\coin(x;\alpha\bvrho)
    \coloneqq
    \cC_t\bigl(x;\bq(\alpha\bvrho)\bigr)
$.
Here, $\Theta_{\bvrho}$ is the simplex of nonnegative stagewise allocation proportions defined in Section~\ref{ssec:vopt-coinm}.
Accordingly, we define the \emph{oracle allocation} $\bvrho^*$ as
an allocation that maximizes the population utility of the resulting
prediction sequence:
\begin{equation}
    \label{eq:population-risk-investment}
    \bvrho^*
    \in
    \argmin{\bvrho\in\Theta_{\bvrho}}
    -\cU\bigl(\bq(\alpha\bvrho)\bigr).
\end{equation}

The following result compares the learned and oracle allocations after
small changes to their stagewise budgets that account for threshold
estimation. Write $a_{n_1}\coloneqq5\sqrt{T\log(n_1)/n_1}$ and
$\iota_T\coloneqq(1,\ldots,T)$, and let $(\cdot)_+$ denote the
coordinatewise positive part. The threshold comparison underlying the
result is given in Proposition~\ref{prop:conservativequantile}.

\begin{theorem}[Shifted oracle inequality for \vocoinm]
    \label{the:shifted-oracle-inequality}
    Suppose that Assumption~\ref{cond:independent} holds and that the score
    values are almost surely free of ties. Assume that
    $n_1>15$, $1<T<n_1\log(n_1)/9$, and
    $a_{n_1}T(T+1)/2<1-\alpha$. If $u$ is bounded and set-wise monotone and
    Problem~\eqref{eq:population-risk-investment} admits a solution
    $\bvrho^*$, then, with probability at least $1-n_1^{-1}$,
    \[
        \cU\Bigl(
            \bq\bigl((\alpha\bvrho^*-a_{n_1}\iota_T)_+\bigr)
        \Bigr)
        -
        \cU\Bigl(
            \bq\bigl(\alpha\hbvrho^{\rm s}+a_{n_1}\iota_T\bigr)
        \Bigr)
        \leq
        2\Delta_{n_1},
    \]
    where
    $
        \Delta_{n_1}
        \coloneqq
        \sup_{\btau\in\overline{\bbR}^T}
        \left|
            \cU(\btau)
            -
            n_1^{-1}
            \sum_{i=1}^{n_1}
            u\left(\cC_{1:T}(X_i;\btau)\right)
        \right|
    $
    is the uniform deviation of empirical utility from population utility.
\end{theorem}

The theorem separates threshold estimation from utility estimation.
The budget perturbations account for the former: the oracle (learned) allocation
is evaluated with slightly reduced (increased) budgets. The remaining discrepancy is bounded
by $2\Delta_{n_1}$. For fixed $T$, the perturbations vanish as
$n_1\to\infty$, and this upper bound converges to zero in probability
whenever $\Delta_{n_1}=o_p(1)$. If the effect of the perturbations on
population utility also vanishes uniformly over allocations, the learned
allocation approaches the oracle utility.

Appendix~\ref{asec:uniform-convergence-delta} establishes the
uniform convergence in two representative settings. For classification
with $\cY=[M]$ and $0\leq u\leq U_{\rm max}$,
$\Delta_{n_1}=O_p\{U_{\rm max}\sqrt{T\log(n_1M)/n_1}\}$, covering the
early-resolution utility $u_{\rm res}$. For
$u_{\rm term}$ with $\cY=[0,M]$ or $\cY=[M]$,
$\Delta_{n_1}=O_p\{M\sqrt{T\log(n_1)/n_1}\}$, covering terminal
informativeness in ordered score aggregation.

\section{Extensions: Adaptive Acquisition and Localization}\label{sec:generality}

The sequential exclusion framework accommodates variation across
observations in both acquisition paths and calibration weights.
We illustrate the flexibility of this framework through
two extensions: covariate-dependent acquisition
(Section~\ref{ssec:branchwise-coinm}) and randomized localization
with adaptive allocation learning
(Section~\ref{ssec:locally-adaptive}).

\subsection{Branchwise \coinm for covariate-dependent acquisition}\label{ssec:branchwise-coinm}

In sequential inference, previously observed information may determine which measurement is acquired next. For example, initial clinical features may
guide the choice of a subsequent laboratory or imaging examination
\citep{von2025evaluation}. We briefly illustrate how \coinm accommodates such
acquisition paths in a two-stage construction. The same design extends
recursively to multiple stages, provided acquisition decisions use only
observed information and the symmetry and exclusion-count conditions
remain satisfied.

Consider $X=(X^1,X^2)\in\cX^1\times\cX^2$, with $X^1$ observed
at stage~1. At stage~2, one of $J$ measurement options is selected:
option $j$ reveals $g_j(X^2)$, such as a specified subset of features.
A prespecified selection rule $\kappa:\cX^1\to[J]$ determines
which option to use based on the first-stage features.
An observation with $X^1=x^1$ therefore follows branch $\kappa(x^1)$
and reveals $g_{\kappa(x^1)}(X^2)$.
Let $s_{2,j}$ be the score function for branch $j$, using only
the first-stage features, the selected measurement, and the
candidate outcome. For $z=(x^1,x^2,y)$, the resulting
second-stage score is
$
    s_2^\ad(z)
    \coloneqq
    s_{2,\kappa(x^1)}
    \bigl(x^1,g_{\kappa(x^1)}(x^2),y\bigr).
$
We take the selection rule as given; learning acquisition rules that
balance informativeness and acquisition costs is left for future work.

Fix prespecified counts $c_1,c_2\in\mathbb{Z}_{\geq0}$ satisfying
$c_1+c_2=c_{\rm tot}$, and apply ordinary \coinm with budget $c_1$
at stage~1. At stage~2, using $s_2^\ad$ in the exclusion
rule~\eqref{eq:coin-selection-set} yields a valid procedure that
ranks all surviving observations together. Under this pooled
ranking, the number of exclusions in each branch depends on
how its scores compare with those in other branches.
Branchwise \coinm makes the allocation across branches explicit,
without requiring a common score scale. It first distributes
$c_2$ across branches according to their sizes in the augmented
sample, then ranks surviving observations within each branch.
Specifically, the branch budgets are proportional to branch
sizes up to integer rounding:
\begin{equation*}
    c_{2,j}
    \coloneqq
    \bigg\lfloor
        \frac{c_2\sum_{i=1}^{n+1}
        \bbI\{\kappa(X_i^1)\leq j\}}{n+1}
    \bigg\rfloor
    -
    \bigg\lfloor
        \frac{c_2\sum_{i=1}^{n+1}
        \bbI\{\kappa(X_i^1)\leq j-1\}}{n+1}
    \bigg\rfloor,
    \qquad j\in[J].
\end{equation*}

This allocation satisfies $\sum_{j=1}^J c_{2,j}=c_2$.
Within branch $j$, exclude up to $c_{2,j}$ surviving observations
with the largest second-stage scores, with index set
\begin{equation*}
    \begin{aligned}
        \cS_{2,j}^{y,\mathrm{br}}
        \coloneqq
        \bigg\{
            i\in[n+1]:
            &\ \kappa(X_i^1)=j,\;
            \delta_1^y(Z_i^y)=0,
            \\
            \sum_{i'=1}^{n+1}
            \bbI\Big\{
                &\ \kappa(X_{i'}^1)=j,\;
                \delta_1^y(Z_{i'}^y)=0,\;
                s_2^\ad(Z_{i'}^y)\geq s_2^\ad(Z_i^y)
            \Big\}
            \leq c_{2,j}
        \bigg\}.
    \end{aligned}
\end{equation*}
Set
$\cS_2^{y,\mathrm{br}}
\coloneqq\bigcup_{j=1}^J\cS_{2,j}^{y,\mathrm{br}}$.
Applying the update and inversion rules
in~\eqref{eq:coinupdatestrategy} and~\eqref{eq:coin-prediction-set-seq}
gives the prediction sequence
$\{\hcC_t^{\mathrm{br}}(X_{n+1})\}_{t=1}^2$.

As branch budgets depend only on branch sizes and stage-2
exclusions depend only on the within-branch ranks of surviving
observations, the prediction sequence is invariant to
strictly increasing transformations of the second-stage scores
in each branch. The following proposition establishes
finite-sample any-stage validity.

\begin{proposition}[Validity of branchwise \coinm]
    \label{prop:adaptive-feature-acquisition-validity}
    Under Assumption~\ref{cond:exchangeability}, suppose that the selection
    rule $\kappa$, measurement functions $\{g_j\}_{j=1}^J$, and
    branch-specific score functions are prespecified.
    Then branchwise \coinm satisfies
    $
        \pr\big\{
            \exists\,t\in[2]:
            Y_{n+1}\notin\hcC_t^{\mathrm{br}}(X_{n+1})
        \big\}
        \leq\alpha.
    $
\end{proposition}

\subsection{Locally adaptive \coinm}\label{ssec:locally-adaptive}

The allocation-learning procedure developed above optimizes
average process-level utility over the population. However,
the informativeness of successive measurements can vary across
covariate regions, so different test points may benefit from
different stagewise allocations. This motivates adapting both
calibration and allocation learning to observations similar to
the test point. Building on localized conformal calibration
\citep{guan2023localized}, we combine randomized localization
\citep{hore2025conformal} with weighted sequential exclusion (Appendix~\ref{assec:principles-coinm-weighted})
to support this adaptation while retaining finite-sample
AMR control.

Let $H:\cX\times\cX\to\bbR_{\scriptscriptstyle\geq0}$ be a prespecified
nonnegative kernel. It acts as a \emph{localizer} by assigning larger values
to nearby covariate pairs, with $H(x,\cdot)$ a density concentrated near
$x$; normalized Gaussian and box kernels are common choices on
$\cX\subseteq\bbR^p$.
To incorporate local weights while preserving validity, we adopt the
randomized localization of \citet{hore2025conformal}: draw an auxiliary point
$\tX_{n+1}$ from $H(X_{n+1},\cdot)$  and set
$w(x)=H(x,\tX_{n+1})$. This emphasizes observations near a randomly chosen
neighbor of the test point. Compared with the calibrated localization of
\citet{guan2023localized}, randomization avoids a data-dependent
quantile-level adjustment and yields validity conditional on the auxiliary
point. We call the resulting procedure
\textbf{r}andomly \textbf{l}ocalized \coinm (\rlcoinm).

Conditional on $\tX_{n+1}$ and the unordered augmented sample, the
probability that an observation occupies the test position is proportional
to its weight. We therefore control AMR by bounding the total weight of excluded
observations and distribute this weighted exclusion budget
across stages according to $\bvrho\in\Theta_{\bvrho}$:
\begin{equation*}
    c_{\rm tot}^\rlcoin
    \coloneqq \alpha\sum_{i=1}^{n+1}H(X_i,\tX_{n+1}),
    \qquad
    c_t^\rlcoin
    \coloneqq c_{\rm tot}^\rlcoin\varrho_t,
    \quad t\in[T].
\end{equation*}
At stage $t$, weighted \coinm excludes surviving observations in descending
score order, subject to a newly excluded weight of at most
$c_t^\rlcoin$; see Appendix~\ref{assec:principles-coinm-weighted} for detailed implementation.
This construction yields
\begin{equation*}
    \hcC_t^\rlcoin(X_{n+1},\tX_{n+1})
    \coloneqq
    \big\{
        y\in\cY:
        \delta_t^{y,\rlcoin}(Z_{n+1}^y)=0
    \big\},
    \qquad t\in[T].
\end{equation*}

\begin{corollary}[Finite-sample validity of \rlcoinm]
    \label{cor:validity-rlcoinm}
    Suppose that $\{(X_i,Y_i)\}_{i=1}^{n+1}$ are i.i.d.\ from
    $P_X\times P_{Y\mid X}$, and the score functions and allocation
    $\bvrho\in\Theta_{\bvrho}$ are fixed or learned from data independent
    of the augmented sample. Let $H$ be prespecified, nonnegative, and
    measurable, with
    $\int_{\cX}H(x,x')\,\rmd x'=1$ for every $x\in\cX$.
    Then \rlcoinm satisfies, almost surely,
    \[
        \pr\Bigl\{
            \exists\,t\in[T]:
            Y_{n+1}\notin
            \hcC_t^\rlcoin(X_{n+1},\tX_{n+1})
            \mathrel{\Big|}\tX_{n+1}
        \Bigr\}
        \leq\alpha.
    \]
    Consequently,
    $
        \pr\big\{
            \exists\,t\in[T]:
            Y_{n+1}\notin
            \hcC_t^\rlcoin(X_{n+1},\tX_{n+1})
        \big\}
        \leq\alpha.
    $
\end{corollary}

Appendix~\ref{assec:robust-covariate-shift} gives AMR bounds under
covariate shift, and Appendix~\ref{assec:localized-allocation-learning}
develops local allocation learning through weighted empirical utility
maximization with sample splitting.

\section{Numerical Studies}\label{sec:numerical-results}

We use numerical studies to address three questions:
whether \coinm improves efficiency relative to the Bonferroni method,
whether allocation learning improves the task-specific utility,
and how optimizing process-level utility differs from optimizing terminal informativeness.
Section~\ref{ssec:sfa-simulation} examines \coinm's performance
and allocation learning in a synthetic sequential feature acquisition
problem with varying information arrival times.
Section~\ref{ssec:sfa-realdata} considers a cost-aware application to dermatological diagnosis and
compares cost-aware early-resolution and terminal-set-size objectives.
Appendix~\ref{assec:utility-efficiency-diagnostic} provides a
controlled simulation that further examines the third question
by varying information timing and acquisition costs.
We also summarize ordered score
aggregation under a terminal-utility objective
(Section~\ref{ssec:specialcase}), with full results in
Appendix~\ref{ssec:numerical-aggregation}.

\subsection{Sequential feature acquisition: a simulation study}
\label{ssec:sfa-simulation}

Write each feature vector as
$X_i=(X_i^1,\cdots,X_i^T)\in\cX^1\times\cdots\times\cX^T$.
The calibration data consist of fully observed feature-outcome pairs,
$\{(X_i^1,\cdots,X_i^T,Y_i)\}_{i=1}^n$, whereas only
$(X_{n+1}^1,\cdots,X_{n+1}^t)$ is available for the test point at stage~$t$.
Accordingly, for each $t\in[T]$, we define a score function
$s_t:\cX^1\times\cdots\times\cX^t\times\cY
\to\bbR_{\scriptscriptstyle\ge0}$ based only on the first $t$ feature blocks
and apply \coinm with these scores and a prespecified allocation.
Our goal is to reach a prespecified set-size threshold $\lambda_0$ as
early as possible while maintaining AMR control. This instantiates the
early-resolution objective in~\eqref{eq:utility-early-resolution-ori}:
\begin{equation*}
    u(\cC_{1:T})
    :=
    \frac{T+1-\min\{t\in[T]:|\cC_t|\le\lambda_0\}}{T},
\end{equation*}
with the convention $\min\varnothing=T+1$.

We compare \vocoinm with three baselines: (1) a naive method that constructs
a $(1-\alpha)$-level conformal set at every stage; (2) \texttt{Equal-Bonf},
which divides $\alpha$ equally across the three stages; and
(3) \texttt{Equal-COINS}, which divides the exclusion-count budget equally.
These comparisons distinguish two sources of improvement.
Comparing \texttt{Equal-Bonf} with \texttt{Equal-COINS} isolates the gain from replacing
stagewise Bonferroni accounting by count investment under the same
equal allocation.
Comparing \texttt{Equal-COINS} with \vocoinm isolates the
additional gain from tailoring the allocation to the specified early-resolution utility.

We consider a three-class classification problem with three stages.
The observations are generated from the following model:
\begin{equation*}
    \pr(Y=k\mid X)
    =
    \frac{
        \exp\Big\{
            2\gamma_k^\top
            \big[
                (1-\eta)X^1
                +\sqrt{2\eta(1-\eta)}X^2
                +\eta X^3
            \big]
        \Big\}
    }{
        \sum_{k'=1}^{3}
        \exp\Big\{
            2\gamma_{k'}^\top
            \big[
                (1-\eta)X^1
                +\sqrt{2\eta(1-\eta)}X^2
                +\eta X^3
            \big]
        \Big\}
    },
    \qquad k\in[3],
\end{equation*}
where $X=(X^1,X^2,X^3)$ with
$X^1,X^2,X^3\stackrel{\mathrm{i.i.d.}}{\sim}N_2(0,I_2)$,
and the coefficient vectors are
$\gamma_1=(1,0)^\top$,
$\gamma_2=(-1/2,\sqrt{3}/2)^\top$, and
$\gamma_3=(-1/2,-\sqrt{3}/2)^\top$.
At stage $t$, the first $t$ feature blocks $(X^1,\ldots,X^t)$ are available for the test point.
Here, $\eta\in[0,1]$ controls when the discriminative information becomes available:
$\eta=0$ places all discriminative information in the first stage,
whereas $\eta=1$ delays it until the third stage.

Throughout this simulation, we set the target AMR level to $\alpha=0.1$ and take $\lambda_0=1$.
The calibration and test samples each contain 1,000 observations.
The nonconformity scores are
$s_t(x^1,\ldots,x^t,y)=1-\widehat{\pr}(Y=y\mid
X^1=x^1,\ldots,X^t=x^t)$ for $t\in[3]$,
where the conditional probabilities are estimated from an independent training set of size $1,000$.
The reported results are averaged over 500 independent repetitions.

\begin{figure}[!t]
    \setstretch{1.2}
    \centering
    \includegraphics[width=0.9\linewidth]{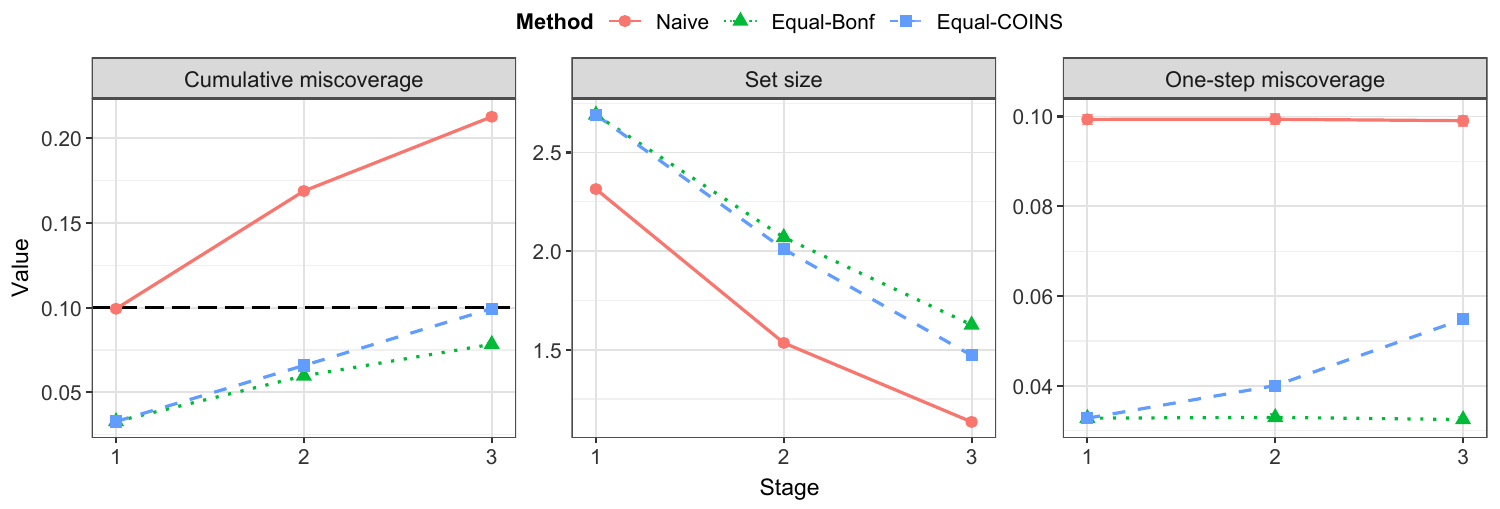}
    \caption{\small
    Comparison of the three methods at $\eta=0.5$.
    Results are averaged over 500 repetitions.
    }
    \label{fig:demo-sequentially-acquire}
\end{figure}

We first fix $\eta=0.5$ and compare the naive, \texttt{Equal-Bonf}, and \texttt{Equal-COINS} methods across the three stages.
As shown in Figure~\ref{fig:demo-sequentially-acquire}, the naive method produces the smallest prediction sets at every stage, but its empirical AMR increases from $0.099$ at stage~1 to $0.213$ at stage~3, exceeding the target level of $0.1$.
\texttt{Equal-Bonf} maintains empirical AMR below the target level but is conservative.
In contrast, \texttt{Equal-COINS} maintains empirical AMR near the nominal
level and produces smaller prediction sets than \texttt{Equal-Bonf} across
the sequence, corroborating the prediction-set dominance result in
Section~\ref{ssec:dom-bonferroni}.

\begin{figure}[!t]
    \setstretch{1.2}
    \centering
    \includegraphics[width=0.9\linewidth]
    {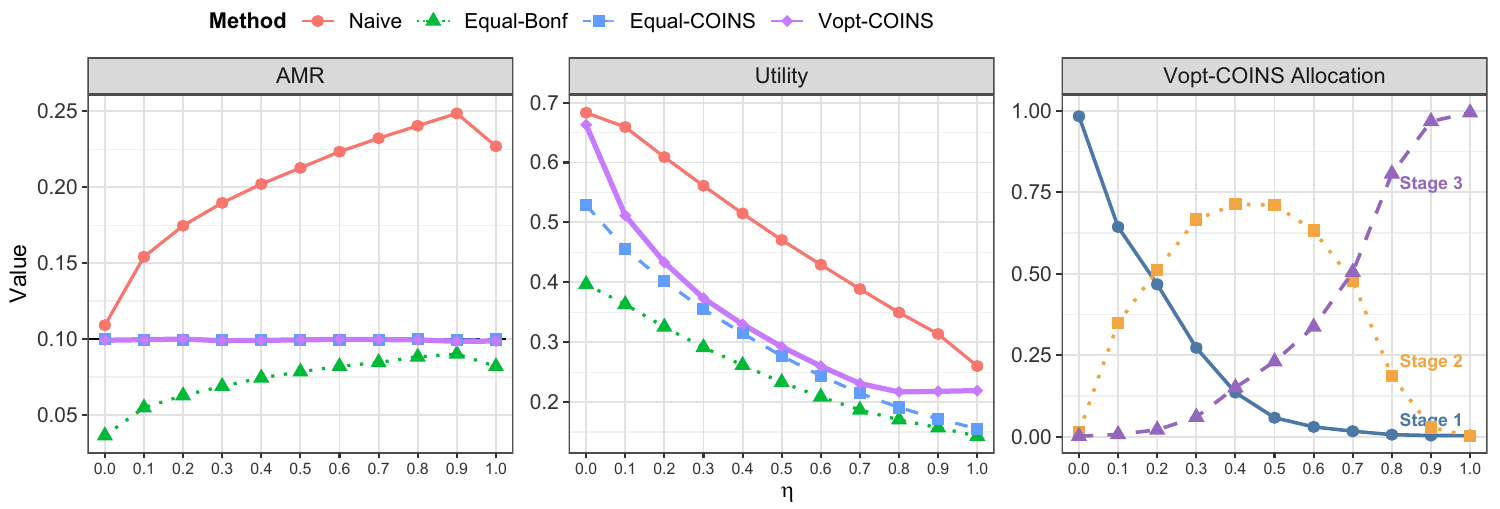}
    \caption{
        \small
        Comparison of the four methods across $\eta$.
        The first two panels report empirical AMR and utility, respectively.
        The third panel reports the proportions of the exclusion-count budget
        allocated by \vocoinm to the three stages.
        All metrics are averaged over 500 repetitions.
    }
    \label{fig:sfa-simulation-eta}
\end{figure}

We next include \vocoinm and vary $\eta$ to shift the discriminative information from earlier to later stages.
Figure~\ref{fig:sfa-simulation-eta} summarizes the results.
Several patterns emerge.
First, the naive method violates AMR control for every value of $\eta$, whereas all three corrected methods control AMR.
\texttt{Equal-Bonf} remains conservative, while the empirical AMRs of \texttt{Equal-COINS} and \vocoinm stay close to the target level.
Second, among the methods that control AMR, \vocoinm always achieves the highest utility, with its gains being most pronounced when $\eta$ is close to $0$ or $1$.
The learned allocations in the third panel explain this advantage: as $\eta$
increases, \vocoinm shifts the exclusion-count budget from earlier to later
stages, adapting to when the discriminative information becomes available.

\subsection{Cost-aware sequential differential diagnosis}
\label{ssec:sfa-realdata}

The differential diagnosis of erythemato-squamous diseases is clinically
important but challenging because the candidate diseases share the clinical
features of erythema and scaling \citep{guvenir1998learning}.
Diagnosis typically begins with a clinical assessment,
which is sufficient in some cases; otherwise, histopathological
features obtained from a skin biopsy may be needed for a definitive diagnosis.
This staged workflow raises a practical question:
can we reduce the need for costly biopsy by identifying patients for whom clinical assessment alone is sufficient to support a reliable diagnosis?

This clinical goal can be formulated as a sequential feature acquisition problem,
in which we use a singleton prediction set as a statistical proxy for reaching a diagnostic conclusion,
impose AMR control as a reliability requirement for the prediction sequence,
and use \emph{cost-aware early-resolution utility} to evaluate the procedure's cost efficiency.

We illustrate this framework using the Dermatology dataset collected by \citet{guvenir1998learning}, which contains 366 patients, 34 features, and six disease classes.
We write the observations as
$\{(X_i^1,X_i^2,Y_i)\}_{i=1}^{366}$, where $X_i^1$ comprises the 12 clinical features,
$X_i^2$ comprises the 22 histopathological features obtained from skin biopsy,
and $Y_i$ denotes the disease label.
These two feature blocks define a two-stage acquisition process:
clinical information is available at stage~1, whereas histopathological
information is acquired only if the procedure continues to stage~2.
Both feature blocks are available for training and calibration, but information about each test patient is revealed according to the staged workflow.

To account for the costs of information acquisition and further
assessment, we seek to maximize the expected
\emph{cost-aware early-resolution utility}, defined by
\begin{equation*}
    u(\cC_{1:2})
    =
    1-0.2 \bigg\{
        \sum_{t=1}^{\min(\tau,2)} d_t
        +
        \xi \times \bbI\{\tau=3\}
    \bigg\},
    \qquad
    \tau
    =
    \min\big\{t\in\{1,2\}:|\cC_t|\leq 1\big\}
    \footnote{
        \renewcommand{\baselinestretch}{1}\selectfont
        For analytical convenience, we treat an empty set as an admissible decision.
        In practice, one can replace an empty set with the singleton containing the classifier's most probable label at that stage.
        This replacement preserves the stopping time and cost-aware utility without increasing AMR.
    }
    ,
\end{equation*}
with the convention $\min\varnothing=3$.
Here, $d_1$ is the cost of the initial clinical assessment,
incurred for every patient,
and $d_2$ is the additional cost of biopsy.
The penalty $\xi$ represents the additional burden of further assessment,
such as specialist review or additional diagnostic testing,
when neither stage satisfies the stopping criterion.
Compared with the one in~\eqref{eq:utility-early-resolution-ori},
this utility accounts for unequal stagewise acquisition costs and a separate penalty for cases that remain unresolved after both stages.
It therefore more closely reflects the practical costs of the diagnostic process.
In our experiments, we fix $d_1=1$ and $\xi=5$ and vary $d_2$ over $\{1,2,3,4,5\}$
to examine how changes in the relative cost of biopsy affect the learned stagewise allocations and the comparative performance of the procedures.

We compare five methods:
the naive method, \texttt{Equal-Bonf}, \texttt{Equal-COINS},
and two variants of \vocoinm that maximize cost-aware early-resolution utility (\texttt{Vopt-COINS-Cost}) and negative terminal set size (\texttt{Vopt-COINS-Size}), respectively.
We repeat the experiment 500 times, each time randomly partitioning the data into training, calibration, and test sets in an approximate $3{:}3{:}1$ ratio.
For both variants of \vocoinm, we further divide the calibration set equally into a strategy-learning sample and a final-calibration sample.
Within each repetition, we fit an $\ell_2$-regularized multinomial logistic regression at each stage using the training data and the features available up to that stage.
The nonconformity score for each candidate diagnosis is one minus its estimated conditional probability.
We set $\alpha=0.1$ throughout.
Below, Table~\ref{tab:dermatology-cost-125} reports results for $d_2 = 2$, while Figure~\ref{fig:dermatology-cost-allocation} covers a broader range of $d_2$.

\begin{table}[!htbp]
    \setstretch{1}
    \centering
    \caption{
        \small
        Comparison of five methods with $(d_1,d_2,\xi)=(1,2,5)$.
        Results are averaged over 500 repetitions; Monte Carlo standard errors are reported in parentheses.
    }
    \label{tab:dermatology-cost-125}
    \begin{tabular}{@{}lcccc@{}}
        \toprule
        Method
        & AMR
        & \shortstack{Cost-aware\\Utility}
        & \shortstack{Terminal\\Set Size}
        & \shortstack{Stage-1\\Budget Share} \\
        \midrule
        Naive
        & 0.157 (0.002)
        & 0.751 (0.001)
        & 0.850 (0.002)
        & - \\
        \texttt{Equal-Bonf}
        & 0.084 (0.002)
        & 0.669 (0.002)
        & 0.935 (0.002)
        & - \\
        \texttt{Equal-COINS}
        & 0.100 (0.002)
        & 0.670 (0.002)
        & 0.913 (0.002)
        & - \\
        \texttt{Vopt-COINS-Cost}
        & 0.103 (0.002)
        & 0.703 (0.004)
        & 0.946 (0.004)
        & 0.841 (0.005) \\
        \texttt{Vopt-COINS-Size}
        & 0.100 (0.002)
        & 0.430 (0.004)
        & 0.912 (0.003)
        & 0.050 (0.007) \\
        \bottomrule
    \end{tabular}
\end{table}

As shown in Table~\ref{tab:dermatology-cost-125}, the naive method attains the highest cost-aware utility but has inflated empirical AMR.
\texttt{Equal-Bonf} is conservative, whereas the three \coinm-based procedures have empirical AMRs close to the nominal level.
Among these three procedures, \texttt{Vopt-COINS-Cost} achieves the highest cost-aware utility.
It allocates most of its budget to stage~1, favoring stopping based on clinical features alone and reducing biopsy use.

Table~\ref{tab:dermatology-cost-125} also shows the distinction between early-resolution utility and terminal set size:
compared to \texttt{Vopt-COINS-Cost},
\texttt{Vopt-COINS-Size} produces smaller prediction sets at the terminal stage but achieves lower cost-aware utility.
It allocates most of its budget to stage~2, where additional histopathological information helps narrow the candidate diagnoses.
Its emphasis on refinement after biopsy leads to fewer stops based on clinical features alone, so smaller terminal sets come with greater biopsy use and lower cost-aware utility.
Appendix~\ref{assec:utility-efficiency-diagnostic} provides a controlled demonstration of this distinction using synthetic data.

Figure~\ref{fig:dermatology-cost-allocation} shows how cost-aware utility and stagewise budget allocation vary with biopsy cost.
As biopsy becomes more costly, cost-aware utility decreases for all methods, but \texttt{Vopt-COINS-Cost} shows a smaller decline than the other non-naive procedures.
As shown in the right panel, \texttt{Vopt-COINS-Cost} responds to higher biopsy costs by shifting more budget to stage~1, promoting early stopping with clinical features alone and reducing biopsy use.
In contrast, the allocation of \texttt{Vopt-COINS-Size} remains concentrated at stage~2 as its objective excludes acquisition costs.
Thus, as biopsy becomes more expensive, cost-aware optimization increasingly favors early resolution, and its utility advantage over \texttt{Vopt-COINS-Size} becomes more pronounced.

\begin{figure}[!t]
    \setstretch{1.2}
    \centering
    \includegraphics[width=0.85\linewidth]
    {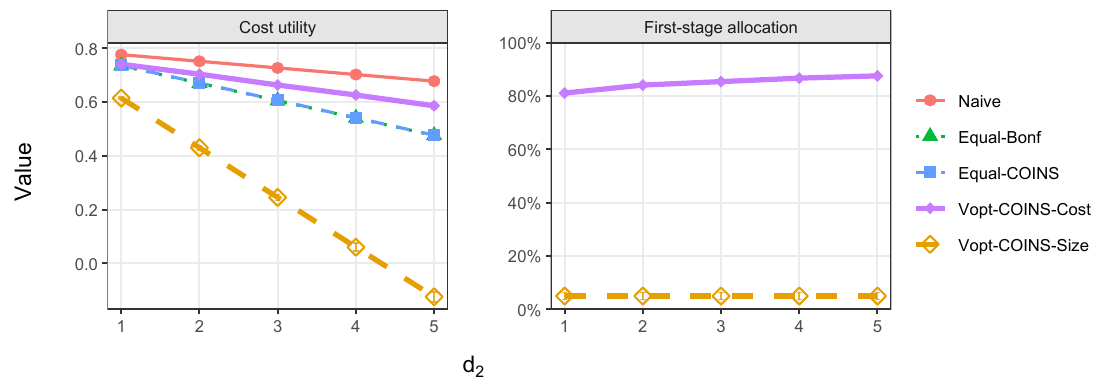}
    \caption{
        \small
        Comparison of five methods on the Dermatology dataset as $d_2$ varies.
        The left panel reports cost-aware early-resolution utility.
        The right panel reports the mean proportion of the exclusion-count budget allocated to stage~1 by \texttt{Vopt-COINS-Cost}
        and \texttt{Vopt-COINS-Size}.
        Results are averages over 500 repetitions.
    }
    \label{fig:dermatology-cost-allocation}
\end{figure}

\subsection{Ordered score aggregation: a terminal-utility special case}\label{ssec:specialcase}

When all test features and score functions are available in advance
and utility depends only on the terminal prediction set, the framework
specializes to ordered score aggregation.
The theoretical characterization and illustrative example in
Appendix~\ref{ssec:numerical-aggregation} explain how count investment
can improve on individual scores by combining their exclusions across
different parts of the oracle exclusion region.
In the accompanying study on eight UCI regression datasets, \vocoinm
attains the shortest average prediction interval on five and is within
two Monte Carlo standard errors of the shortest on a sixth, with
empirical coverage close to the nominal level.
Full theoretical and numerical details are provided in
Appendix~\ref{ssec:numerical-aggregation}.

\section*{Acknowledgments}

The authors used large language model (primarily GPT-6 and GPT-5.6) to assist with language editing, manuscript organization, and code development.
The authors take full responsibility for the accuracy, originality, and integrity of the manuscript.
All real-world datasets analyzed in this study are publicly available from the UCI Machine Learning Repository.

\begingroup
\small
\linespread{1}\selectfont
\bibliography{reference}
\endgroup

\clearpage

\appendix

\renewcommand*{\theHsection}{supp.\arabic{section}}
\renewcommand*{\theHsubsection}{\theHsection.\arabic{subsection}}
\renewcommand*{\theHequation}{supp.\arabic{equation}}
\renewcommand*{\theHfigure}{supp.\arabic{figure}}
\renewcommand*{\theHtable}{supp.\arabic{table}}
\renewcommand*{\theHtheorem}{supp.\arabic{theorem}}
\renewcommand*{\theHlemma}{supp.\arabic{lemma}}
\renewcommand*{\theHassumption}{supp.\arabic{assumption}}
\renewcommand*{\theHremark}{supp.\arabic{remark}}

\renewcommand{\thesection}{\Alph{section}}
\renewcommand{\theequation}{S.\arabic{equation}}
\renewcommand{\thefigure}{S.\arabic{figure}}
\renewcommand{\thetable}{S.\arabic{table}}
\renewcommand{\thealgorithm}{S.\arabic{algorithm}}
\renewcommand{\thetheorem}{S.\arabic{theorem}}
\renewcommand{\thelemma}{S.\arabic{lemma}}
\renewcommand{\theassumption}{S.\arabic{assumption}}
\renewcommand{\theremark}{S.\arabic{remark}}

\setcounter{section}{0}
\setcounter{equation}{0}
\setcounter{table}{0}
\setcounter{figure}{0}

\setcounter{theorem}{0}
\setcounter{lemma}{0}
\setcounter{assumption}{0}
\setcounter{remark}{0}

\begin{center}
    {\bfseries\large Supplementary Material for}\\[3pt]
    {\bfseries\Large ``COINS: Any-Stage-Valid and Utility-Oriented Sequential Conformal Prediction''}
\end{center}

The supplement is organized as follows. Section~\ref{asec:literaturereview}
provides additional connections to the literature.
Section~\ref{asec:weighted-process} develops weighted and localized
extensions. Section~\ref{asec:algorithmic-numerical-details} collects
algorithmic details and additional numerical results. Sections~\ref{asec:nested-quantile-results}
and~\ref{asec:uniform-convergence-delta} state supporting theoretical results.
Proofs of main-text and supplementary results are collected separately in
Sections~\ref{asec:proofs-main} and~\ref{asec:proofs-appendix}, respectively.

\section{Literature Review}\label{asec:literaturereview}

\subsection{Anytime-valid inference}

Any-stage validity is closely related to confidence sequences and $e$-processes,
which support continuous monitoring and optional stopping
\citep{darling1967confidence,jennison1989interim,howard2021time,ramdas2023game,waudby2024time}.
Classical approaches typically construct anytime-valid evidence from a
nonnegative martingale, supermartingale, or betting process and then obtain
sequential inference by thresholding or inversion. Proposition~\ref{the:eprocess-val}
provides a different route: survivor-symmetric exclusion in an exchangeable
augmented sample first yields exact finite-sample count identities, which then
generate stagewise $e$-values and an $e$-process. Thus, the evidence process is
a consequence of the conformal exclusion structure rather than the starting
point of the construction.

\subsection{Selection and aggregation for conformal prediction}

Existing approaches improve a single prediction set through model selection,
score aggregation, calibrated evidence combination, or miscoverage allocation
\citep{yang2025selection,liang2024conformal,wang2026localized,patel2025conformal,alami2026symmetric,qin2024data,xu2025aggregating}.
These methods are most directly related to the terminal-utility special case of
our framework, in which all scores are available and intermediate reports have
no operational role. We discuss the allocation-based \texttt{COLA} method
and the related sequential evidence-aggregation construction \texttt{SeqSB}
in more detail below.

\subsubsection{Comparison with \texttt{COLA}}

The closest allocation-based comparison is \texttt{COLA}
\citep{xu2025aggregating}, which allocates miscoverage across score-specific
conformal sets to optimize their intersection.
The distinction is twofold: the optimization objective and the calibration mechanism.
\texttt{COLA} targets the informativeness of a single aggregated set,
whereas \vocoinm\ learns an allocation for a specified process-level utility,
which can depend on intermediate set sizes, stopping time, and acquisition costs.
Optimizing the informativeness of the terminal set is a special case of our framework.

The calibration mechanisms also differ.
Each separately calibrated set in the Bonferroni construction uses the full
calibration sample, so overlapping exclusions can consume the error budget
at several stages. In contrast, \coinm allocates a finite-sample exclusion-count
budget only among surviving augmented observations, and an excluded
observation cannot consume the budget again.
Proposition~\ref{prop:prediction-set-dominance} shows that, for a fixed
Bonferroni miscoverage allocation, its matched \coinm count allocation
produces prediction sets contained in their Bonferroni counterparts
at every stage, and hence dominates for every set-wise monotone utility.
Appendix~\ref{ssec:numerical-aggregation} compares both learned and
transferred allocations empirically.

\subsubsection{Comparison with \texttt{SeqSB}}

The concurrent and independent work by \citet{schrab2026aggregation} develops
Single-Batch (SB) aggregation and its sequential extension, Sequential
Single-Batch (\texttt{SeqSB}), for ordered test statistics under exchangeability
induced by group invariance.
\texttt{SeqSB} uses stagewise spending and permanent elimination of
exchangeable rows, ensuring that previously excluded rows do not consume
subsequent budgets and permitting valid early rejection.
This construction can be extended to sequential conformal prediction and
shares the survivor-based exclusion mechanism of \coinm.
Despite this shared mechanism, our work differs in its objectives,
theoretical scope, and computational approach.

Our objective is to optimize the expected utility of an entire prediction
sequence for a single outcome as information is progressively acquired.
Because intermediate prediction sets guide decisions about stopping and
further information acquisition, this process-level utility can depend on
prediction-set informativeness, stopping time, and acquisition costs.
Whereas \citet{schrab2026aggregation} leave open how to choose or learn the
spending sequence, we develop \vocoinm to learn the stagewise allocation
for a specified utility within the \coinm family.

Beyond the concrete construction, we develop a general structural theoretical
framework for sequential exclusion.
First, Theorem~\ref{the:basicresult} identifies monotonicity, exclusion-count, and filtration conditions
sufficient for finite-sample any-stage validity under exchangeability.
The filtration specifies how exclusion rules may adapt while preserving
conditional symmetry among surviving observations, providing a basis
for our branchwise and localized conformal extensions.
Conversely, Theorem~\ref{the:universality-principles} shows that every
nested, calibration-symmetric prediction sequence satisfying the stated
AMR condition admits a
sequential exclusion representation satisfying the three conditions.
This universality places the exclusion mechanism shared by \coinm
and \texttt{SeqSB} in a broader structural context: within the stated
class, sequential exclusion provides a common structural description
of prediction sequences constructed from different starting points.

Finally, we derive an efficient implementation of \coinm using
calibration-only thresholds.
For a fixed score sequence and allocation, these thresholds are
computed once and reused across candidate outcomes and test observations.
The Two-Batch (TB) conformal construction of \citet{schrab2026aggregation}
also permits calibration reuse, but splits the calibration data into
reference and aggregation samples; \coinm requires no additional
reference split.
Note that the strategy-learning split in \vocoinm\ serves a different purpose:
preserving validity after allocation learning.

\subsection{Adaptive testing and conformal selection}

Interactive methods for multiple testing and conformal selection use
filtrations to restrict the information available for adaptive ranking
and selection while controlling the false discovery rate.
AdaPT \citep{lei2018adapt} and adaptive knockoffs
\citep{ren2023knockoffs} use side information and partially revealed
data to improve discovery.
ACS \citep{gui2025acs} provides an interactive extension of conformal
selection \citep{jin2023selection}, allowing ranking and selection
strategies to change as new information becomes available.

Our framework likewise uses filtrations to restrict the information
available for adaptive exclusion, but the inferential targets,
risks, and utilities differ.
These selection methods control the expected proportion of false
discoveries among selected hypotheses or test observations; we construct
successive prediction sets for a single unobserved outcome and control
the probability that any reported set excludes it.
Our utility is defined on the entire prediction sequence rather
than a selected collection, and can depend on prediction-set
informativeness, stopping time, and acquisition costs.
We therefore study how to refine predictions as information is
acquired to maximize expected process-level utility.

\section{Weighted and Localized Extensions}\label{asec:weighted-process}

\subsection{Structural conditions and \coinm under weighted settings}\label{assec:principles-coinm-weighted}

We first extend the three structural conditions to a weighted setting.
Specifically, let $w:\cX\to\bbR_{\scriptscriptstyle\ge 0}$ be a nonnegative weight function that reflects the covariate shift.
We replace Conditions~\ref{p2:reje-count} and~\ref{p3:per-inv} with the following weighted counterparts.
\begin{enumerate}[label={\small (S2-W)}, ref=(S2-W), series = weightedcriteria]
    \item\label{p2w:reje-count-weight}
    For each $y \in \cY$, the final exclusion indicators satisfy
    \begin{equation*}
        \sum_{i=1}^{n+1}
        w(X_i)
        \delta_T^y\bigl(Z_i^y\bigr)
        \le
        \alpha\sum_{i=1}^{n+1}w(X_i)
        =:
        c_{\rm tot}^w.
    \end{equation*}
\end{enumerate}

\begin{enumerate}[label={\small (S3-W)}, ref=(S3-W), series = weightedcriteria]
    \item\label{p3w:per-inv-weight}
    Condition~\ref{p3:per-inv} holds with $\cF_0^y$ revealing the weight function:
    \[
        \cF_0^y = \sigma \big( \Lbag \ocD(y) \Rbag,
        s_1 , w \big).
    \]
\end{enumerate}

Condition~\ref{p2w:reje-count-weight} replaces the unweighted exclusion-count budget in Condition~\ref{p2:reje-count} with a weighted counterpart.
This weighted version can be viewed as a sequential generalization of the weighted quantile used by \citet{tibshirani2019conformal}.
Notably, the constraint is now stated as an inequality because real-valued weights may prevent the error budget from being exhausted exactly.
In turn, Condition~\ref{p3w:per-inv-weight} retains the survivor-symmetry restriction of Condition~\ref{p3:per-inv}, with the sole addition of the weight function to the initial $\sigma$-algebra.

Under covariate shift, these conditions likewise induce a prediction sequence with any-stage validity, provided that the weight function is proportional to the corresponding covariate likelihood ratio.
We first formalize this requirement as follows.

\begin{assumption}[Covariate shift]
    \label{cond:covariate-shift}
    Conditional on $w$, $\{(X_i,Y_i)\}_{i = 1}^{n+1}$ are independent random pairs satisfying
    $(X_1,Y_1), \cdots, (X_n,Y_n) \overset{i.i.d.}{\sim} P_X \times P_{Y \mid X}$ and
    $(X_{n+1}, Y_{n+1}) \sim Q_X \times P_{Y \mid X}$.
    Assume also that $Q_X \ll P_X$ and $w(x) \propto (\rmd Q_X / \rmd P_X)(x)$.
\end{assumption}

\begin{theorem}[Any-stage validity under covariate shift]\label{the:basicresult-weighted}
    Suppose that $\{s_t\}_{t = 1}^T$ are pre-trained on independent data.
    Under Assumption~\ref{cond:covariate-shift}, suppose a sequential exclusion process $\boldsymbol{\delta}$ satisfies Conditions~\ref{p1:nondecreasing}, \ref{p2w:reje-count-weight}, and~\ref{p3w:per-inv-weight}. Then the induced prediction sequence satisfies
    \begin{equation*}
        \pr \big\{
            \exists \, t \in [T],
            Y_{n+1} \notin
            \{y: \delta_t^y(Z_{n+1}^y) = 0\}
        \big\}
        \le \alpha
        .
    \end{equation*}
\end{theorem}

Now, we introduce the weighted \coinm procedure that satisfies Conditions~\ref{p1:nondecreasing}, \ref{p2w:reje-count-weight}, and~\ref{p3w:per-inv-weight}.
Notably, the total budget $c_{\rm tot}^w$ depends on the observed covariates and the weight function.
Therefore, we parameterize the stagewise budget allocation using a proportion vector $\bvrho$.
Each $\bvrho \in \Theta_{\bvrho}$ induces the stagewise budget vector
$\bc^w (\bvrho; c_{\rm tot}^w) = \big( c_1^w (\bvrho; c_{\rm tot}^w), \cdots, c_T^w (\bvrho; c_{\rm tot}^w) \big)$,
with
$c_t^w (\bvrho; c_{\rm tot}^w) = c_{\rm tot}^w \varrho_t$ for $t \in [T]$.
With the proportions specified in advance, the stagewise budgets can be computed without compromising the validity guarantee.

Then, to achieve~\ref{p2w:reje-count-weight}, the weighted \coinm procedure follows the rule in~\eqref{eq:coinupdatestrategy},
\begin{equation}\label{eq:coin-update-weight}
    \delta_t^{y, \wcoin}(Z_i^y) =
    \begin{cases}
            1,
        &
            \delta_{t-1}^{y, \wcoin}(Z_i^y) = 0
            \text{ and }
            i \in \cS_t^{y, \wcoin},
        \\
            \delta_{t-1}^{y, \wcoin}(Z_i^y),
        &
            \text{otherwise},
    \end{cases}
\end{equation}
while replacing the newly excluded index set in~\eqref{eq:coin-selection-set} with
\begin{equation}\label{eq:coin-selection-set-weight}
    \begin{aligned}
        \cS_t^{y, \wcoin}
        \! = \!
        \bigg\{
            i \! \in \! [n+1] : \,
            &
            \delta_{t - 1}^{y, \wcoin}(Z_i^y) = 0 ,
            \\
            & \sum_{j = 1}^{n+1}
            w(X_j) \,
            \bbI
            \Big\{
                \delta_{t-1}^{y, \wcoin}(Z_j^y) = 0, s_t(Z_j^y) \ge s_t(Z_i^y)
            \Big\} \le c_t^w(\bvrho; c_{\rm tot}^w)
        \bigg\}.
    \end{aligned}
\end{equation}
This construction ensures
$\sum_{i=1}^{n+1} w(X_i) \delta_T^{y, \wcoin}(Z_i^y)
\le
\sum_{t=1}^{T} c_t^w(\bvrho; c_{\rm tot}^w) = c_{\rm tot}^w$,
verifying Condition~\ref{p2w:reje-count-weight}.
Finally, the prediction sequence induced by weighted \coinm is given by
\begin{equation}\label{eq:prediction-set-wcoin}
    \hcC_t^\wcoin(X_{n+1})
    :=
    \big\{
        y\in\cY:
        \delta_t^{y, \wcoin}(Z_{n+1}^y)=0
    \big\},
    \qquad
    t\in[T].
\end{equation}

The following theorem establishes finite-sample any-stage validity for weighted \coinm.

\begin{theorem}[Any-stage validity of weighted \coinm]
    \label{the:validity-covariate-shift}
    Suppose that $\{s_t\}_{t = 1}^T$ are pre-trained on independent data and $\bvrho$ is prespecified.
    Then, under Assumption~\ref{cond:covariate-shift}, the prediction sequence in \eqref{eq:prediction-set-wcoin} satisfies
    $\pr\big\{
        \exists\,t\in[T],
        Y_{n+1}\notin
        \hcC_t^\wcoin(X_{n+1})
    \big\}
    \le \alpha$.
\end{theorem}

\subsection{Robustness of \rlcoinm to covariate shift}\label{assec:robust-covariate-shift}

This subsection establishes an approximate AMR guarantee for \rlcoinm under a broad class of covariate shifts.
Throughout, $H$ satisfies the kernel conditions in
Corollary~\ref{cor:validity-rlcoinm}. We first introduce the required notation.
For any measurable function $g: \cX \to \bbR_{\scriptscriptstyle \ge 0}$ with $0 < \E_{P_X} [g(X)] < \infty$,
we define $P_X \circ g$ as the distribution $P_X$ reweighted by $g$ and introduce a parameter to characterize the variability of $g$:
for any measurable set $ A \subseteq \cX$,
\begin{equation*}
    (P_X \circ g)(A)
    \coloneqq
    \frac{
        \int_A g(x) \rmd P_X(x)
    }{
        \int_\cX g(x) \rmd P_X(x)
    },
    \qquad
    L_{g,r}(A) \coloneqq
    \sup_{x,x' \in A, \|x - x'\| \le r}
    \frac{
        | g(x) - g(x') |
    }{
        r
    }.
\end{equation*}
$L_{g,r}(A)$ can be interpreted as a relaxed Lipschitz constant for the reweighting function $g$ on $A$.
For example, if $g$ is $L_g$-Lipschitz continuous on $\cX$, then $L_{g,r}(A) \le L_g$ for every $A \subseteq \cX$ and $r >0$.
Let $P_{X,\widetilde X}$ denote the joint law obtained by drawing
$X\sim P_X$ and then $\widetilde X\mid X\sim H(X,\cdot)$.
The following proposition extends Theorem~3 of \citet{hore2025conformal} to AMR control.

\begin{proposition}
    \label{prop:robust-covariate-shift}
    Suppose that the score functions $\{s_t\}_{t=1}^T$ and the allocation
    $\bvrho$ are independent of the calibration and test samples.
    Suppose that the calibration pairs $\{(X_i,Y_i)\}_{i=1}^n$ are i.i.d.\ from $P_X\times P_{Y\mid X}$ and are independent of the test pair $(X_{n+1},Y_{n+1})\sim(P_X\circ g)\times P_{Y\mid X}$.
    Then, for every $g$ with $P_X \circ g$ well defined,
    \begin{equation*}
        \begin{aligned}
            \pr\Big\{ & \,
                \exists\,t\in[T]:
                Y_{n+1}\notin\hcC_t^\rlcoin(X_{n+1}, \tX_{n+1})
            \Big\}
            \le
            \alpha +
            \\ & \,
            \frac{
                \inf_{A \subseteq \cX, \varepsilon > 0}
                \big\{
                    \varepsilon L_{g,2\varepsilon} (A)
                    +
                    \|g\|_\infty
                    \pr_{P_{X,\widetilde X}}
                    (\|X-\widetilde X\|>\varepsilon)
                    +
                    \|g\|_\infty \pr_{P_X} (A^c)
                \big\}
            }{
                \E_{P_X} [g(X)]
            },
        \end{aligned}
    \end{equation*}
    where the probability on the left-hand side is taken with respect to $\{(X_i, Y_i)\}_{i = 1}^{n+1}$ and $\tX_{n+1}$.
\end{proposition}

Proposition~\ref{prop:robust-covariate-shift}
characterizes how the reweighting function $g$ affects the excess term in the AMR upper bound under covariate shift.
\citet{hore2025conformal} discuss several specific choices of $g$ in Section~4.3, including Lipschitz functions and indicator functions; see that section for further details.

\subsection{Locally adaptive allocation learning}
\label{assec:localized-allocation-learning}

Under the data and kernel conditions of
Corollary~\ref{cor:validity-rlcoinm}, the stagewise allocation for \rlcoinm
can be learned with the sample-splitting principle used for \vocoinm.
Randomly partition the calibration
sample into a strategy-learning sample
$\cD_1^\ca=\{(X_i,Y_i)\}_{i=1}^{n_1}$ and a final-calibration sample
$\cD_2^\ca=\{(X_i,Y_i)\}_{i=n_1+1}^{n}$. For
$x,\tilde x\in\cX$, let
$\hcC_t^\rlcoin(x,\tilde x;\cD,\bvrho)$ denote the stage-$t$ set produced
from calibration data $\cD$, allocation proportions $\bvrho$, and weights
$H(\cdot,\tilde x)$. Conditional on $\tX_{n+1}$, define the localized ERM
allocation by
\begin{equation*}
    \hbvrho^{{\rm s},\rlcoin}(\tX_{n+1})
    \in
    \argmin{\bvrho\in\Theta_{\bvrho}}
    -\sum_{i=1}^{n_1}H(X_i,\tX_{n+1})
    u\left(
        \hcC_{1:T}^\rlcoin (X_i,\tX_{n+1};\cD_1^\ca,\bvrho)
    \right).
\end{equation*}
The weights make the empirical criterion emphasize calibration observations
near the auxiliary localization point and hence, indirectly, near the test
point. The final prediction sequence is calibrated on $\cD_2^\ca$ using
$\hbvrho^{{\rm s},\rlcoin}(\tX_{n+1})$. Conditional on
$\cD_1^\ca$ and $\tX_{n+1}$, the allocation is fixed and the test-covariate
likelihood ratio relative to the final-calibration law is proportional to
$H(\cdot,\tX_{n+1})$. Theorem~\ref{the:validity-covariate-shift} therefore
gives finite-sample any-stage validity on the final-calibration split.
The resulting procedure complements existing work on locally adaptive multiple-score aggregation \citep{xu2025aggregating,wang2026localized}.

\section{Algorithmic and Numerical Details}
\label{asec:algorithmic-numerical-details}

This appendix first compares terminal informativeness and cost-aware utility
(Section~\ref{assec:utility-efficiency-diagnostic}), then describes the
reusable threshold implementation of \coinm
(Section~\ref{asubsec:coin-algorithm}). It concludes with theoretical and
numerical results on ordered score aggregation and the associated
implementation details (Section~\ref{ssec:numerical-aggregation}).

\subsection{Terminal informativeness and cost-aware utility}
\label{assec:utility-efficiency-diagnostic}

This study illustrates how terminal informativeness and cost-aware
utility lead to different stagewise allocations within the same
any-stage-valid \coinm family. We use the simulation setup of
Section~\ref{ssec:sfa-simulation} and compare two versions of \vocoinm
with different tuning objectives. On the same strategy-learning sample,
\texttt{Vopt-COINS-Size} maximizes empirical terminal utility, whereas
\texttt{Vopt-COINS-Cost} maximizes empirical cost-aware early-resolution utility, defined by
\begin{equation*}
    u_{\rm term}(\hcC_{1:3})=-|\hcC_3|,
    \qquad
    u_{\rm proc}(\hcC_{1:3})=
    1-0.2\bigg\{
        \sum_{t=1}^{\min(\tau,3)}d_t
        +\xi\times\bbI\{\tau=4\}
    \bigg\},
\end{equation*}
where $\tau=\min\{t\in[3]:|\hcC_t|\leq1\}$ and $\tau=4$ if no stage is
early resolution. We set $\xi=5$, giving the three-stage version of the cost-aware utility in Section~\ref{ssec:sfa-realdata}.
We vary $\eta\in\{0,0.1,\ldots,1\}$ and
$(d_1,d_2,d_3)\in\{(1,2,2),(1,2,4),(1,2,8)\}$, holding the other settings
fixed. Increasing $\eta$ shifts discriminative information toward later
stages, while increasing $d_3$ makes the final acquisition more expensive.

To compare the two utilities on a common scale, we use standardized paired
differences. Let $U^a_{bik}$ be the test utility for method
$a\in\{\mathrm{P},\mathrm{T}\}$, repetition $b$, observation $i$, and utility
$k\in\{\mathrm{term},\mathrm{proc}\}$, where $\mathrm{P}$ and $\mathrm{T}$
denote \texttt{Vopt-COINS-Cost} and \texttt{Vopt-COINS-Size}. With $m=1{,}000$
test observations and $B=500$ repetitions, define
\begin{equation*}
    D_{bk}
    =\frac{m^{-1}\sum_{i=1}^{m}(U^{\mathrm{P}}_{bik}-U^{\mathrm{T}}_{bik})}
    {(s^{\mathrm{P}}_{bk}+s^{\mathrm{T}}_{bk})/2},
    \qquad
    \widehat\Delta_k=\frac{1}{B}\sum_{b=1}^{B}D_{bk},
\end{equation*}
where $s^a_{bk}$ is the sample standard deviation of the $m$ test utilities.
Positive differences favor \texttt{Vopt-COINS-Cost}, and negative differences
favor \texttt{Vopt-COINS-Size}.

\begin{figure}[!t]
    \setstretch{1.2}
    \centering
    \coinsincludegraphics[width=0.9\linewidth]{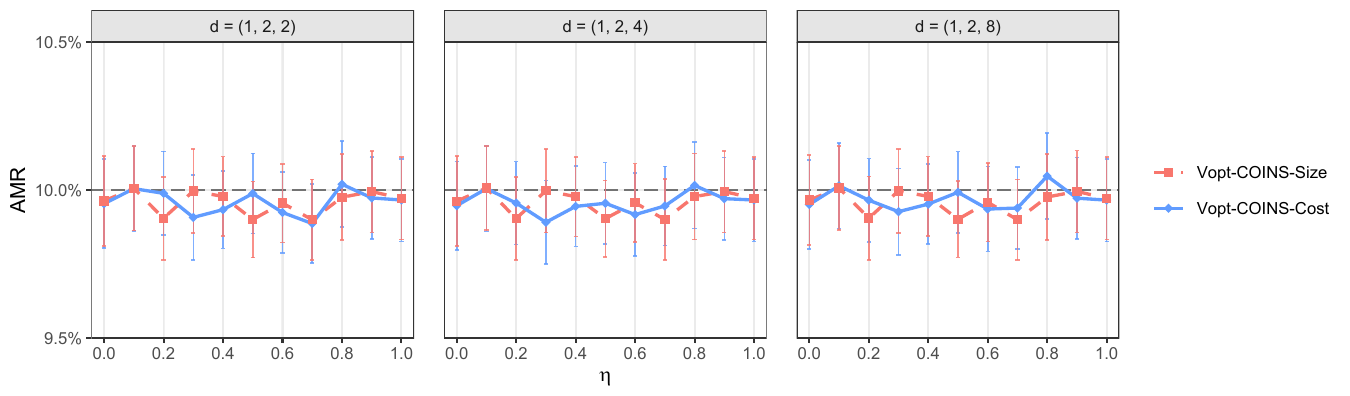}
    \caption{\small
        Empirical AMR for the \texttt{Vopt-COINS-Size} and \texttt{Vopt-COINS-Cost} procedures.
        Curves average 500 repetitions; error bars are pointwise 95\% Monte Carlo $t$ intervals.}
    \label{fig:c1-amr}
\end{figure}

Figure~\ref{fig:c1-amr} shows that empirical AMR control is maintained
under both tuning objectives across the settings considered.
Figure~\ref{fig:c1-utility-performance} highlights the resulting tradeoff:
in most settings, \texttt{Vopt-COINS-Size} produces smaller terminal sets,
whereas \texttt{Vopt-COINS-Cost} achieves higher cost-aware early-resolution utility by reaching
early resolutions at lower acquisition cost. Thus, better terminal
informativeness can come at the expense of cost-aware utility.

\begin{figure}[!t]
    \setstretch{1.2}
    \centering
    \coinsincludegraphics[width=0.9\linewidth]{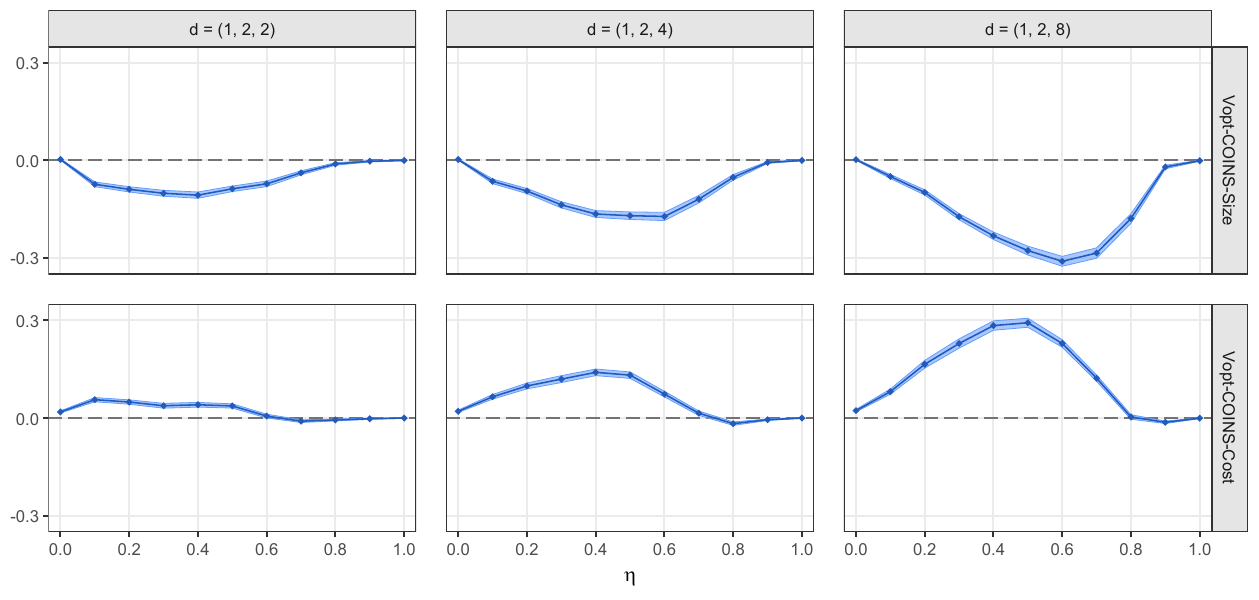}
    \caption{\small
        Standardized paired utility differences for the \texttt{Vopt-COINS-Size} and \texttt{Vopt-COINS-Cost} procedures, averaged over 500 repetitions.
        Positive values favor \texttt{Vopt-COINS-Cost}, and negative values
        favor \texttt{Vopt-COINS-Size}. The dashed lines mark zero.}
    \label{fig:c1-utility-performance}
\end{figure}

Figure~\ref{fig:c1-stagewise-allocation} explains this tradeoff through the
learned allocations. \texttt{Vopt-COINS-Size} places more budget on the final
stage, while \texttt{Vopt-COINS-Cost} shifts budget toward earlier stages as
the final acquisition becomes more expensive. This shift promotes
early resolution before the costly final acquisition. When information
arrives mainly at the final stage, both methods concentrate their budgets
there and their utility differences diminish.

\begin{figure}[!t]
    \setstretch{1.2}
    \centering
    \includegraphics[width=0.9\linewidth]{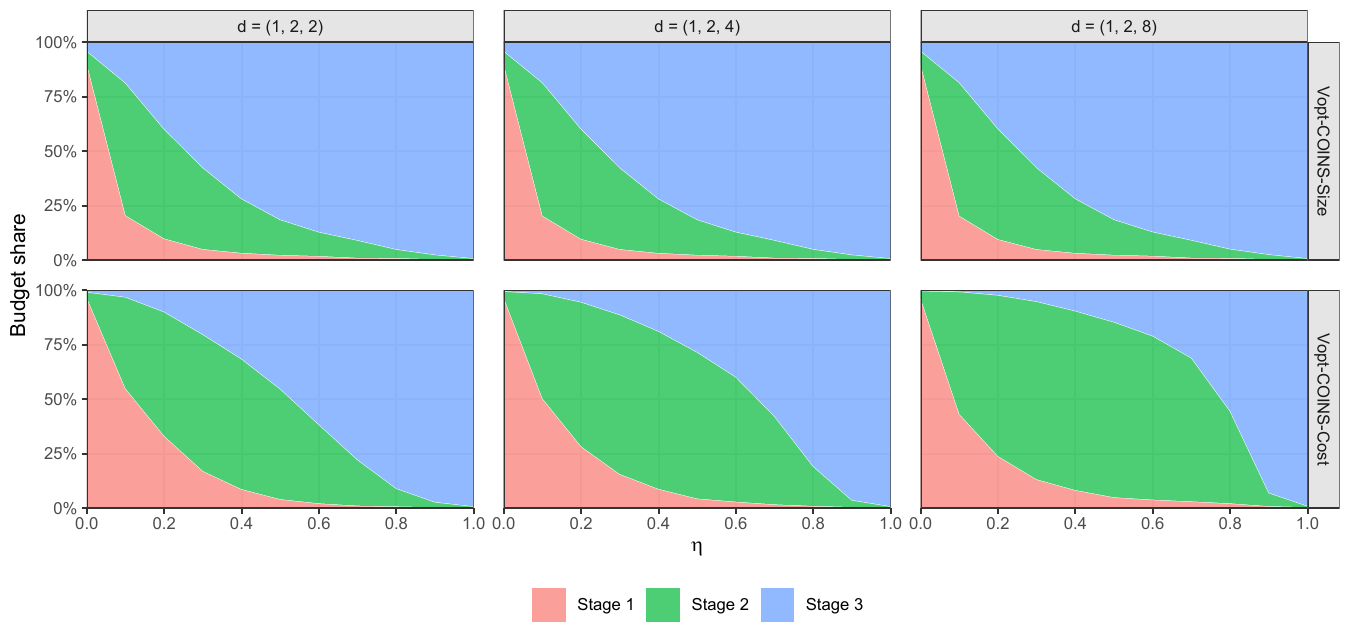}
    \caption{\small
        Stagewise exclusion-count allocations learned under the two
        objectives, averaged over 500 repetitions.
        The stacked areas display the mean proportions $c_t/50$ allocated to stages 1, 2, and 3, from bottom to top.}
    \label{fig:c1-stagewise-allocation}
\end{figure}

\subsection{Reusable thresholds and the \coinm implementation}
\label{asubsec:coin-algorithm}
\label{ssec:implementcoin}

This subsection presents a reusable threshold implementation of \coinm
and establishes its almost-sure agreement with the candidate-wise
construction at the true outcome.
Besides, Figure~\ref{fig:demo-coin} illustrates the exclusion construction in Section~\ref{ssec:coinprocedure}.

\begin{figure}[!t]
    \setstretch{1.2}
    \centering
    \includegraphics[width=1\linewidth]{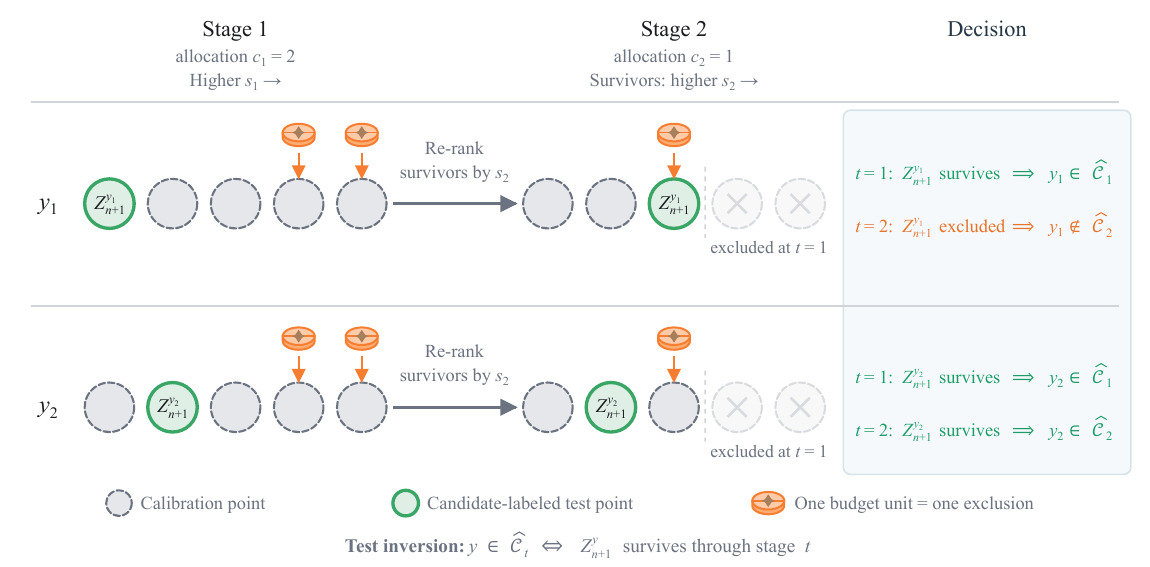}
    \caption{\small Illustration of \coinm for two candidate outcomes and the allocation
    $c_1=2$ and $c_2=1$. Within each row, the augmented observations are ranked
    by $s_1$ at stage~1; only the survivors are reranked by $s_2$ at stage~2.
    Each coin marks one exclusion, crossed circles denote observations excluded
    at stage~1, and the green circle denotes the hypothesized test
    observation. A candidate remains in $\hcC_t^\coin$ if and only if its
    augmented test observation has survived through stage~$t$.}
    \label{fig:demo-coin}
\end{figure}

\paragraph*{Threshold recursion.}
For a finite multiset $A$ of real scores, let $\bbC_c(A)$ denote its
$c$th largest element for $1\leq c\leq|A|$, and set
$\bbC_0(A)=+\infty$, including when $A=\varnothing$.
Fix a prespecified allocation $\bc=(c_1,\ldots,c_T)\in\mathbb{Z}_{\geq0}^T$
with $\sum_{t=1}^T c_t=c_{\rm tot}=\lfloor(n+1)\alpha\rfloor$.
Initialize the active calibration indices as $\cA_0=[n]$ and recursively define
\begin{equation}\label{eq:calibration-only-coin-thresholds}
    \begin{aligned}
        \htau_t(\bc)
        \coloneqq
        \bbC_{c_t}\bigl(\{s_t(Z_i):i\in\cA_{t-1}\}\bigr),
        \quad
        \cA_t
        \coloneqq
        \bigl\{i\in\cA_{t-1}:s_t(Z_i)<\htau_t(\bc)\bigr\},
        \qquad t\in[T].
    \end{aligned}
\end{equation}
Under the no-ties condition, $|\cA_{t-1}|=n-\sum_{r<t}c_r$.
Since $\sum_{t=1}^T c_t=c_{\rm tot}\leq n$, the required order
statistic exists whenever $c_t>0$; when $c_t=0$, no calibration observation is removed.
Note that the sets $\cA_t$ and thresholds
$\htau_t(\bc)$ depend only on the calibration data, scores, and allocation,
so they can be computed before any test point is observed.
The thresholds define
$\tcC_t(x)\coloneqq\{y\in\cY:s_t(x,y)\leq\htau_t(\bc)\}$ and
\begin{equation}\label{eq:calibration-only-coin-sets}
    \hcC_t^{\rm thr}(x)
    \coloneqq\bigcap_{r=1}^t\tcC_r(x),
    \qquad t\in[T].
\end{equation}
The whole algorithm is summarized in Algorithm~\ref{alg:coinprocedure}.

\paragraph*{Justification of the equivalence.}
The key observation is that, in the absence of score ties, a surviving
candidate does not affect which calibration observations are excluded.
In the following proposition, $\hcC_t^\coin$ denotes the candidate-wise
construction in~\eqref{eq:coin-prediction-set-seq}, whereas
$\hcC_t^{\rm thr}$ denotes the threshold construction
in~\eqref{eq:calibration-only-coin-sets}.

\begin{proposition}[Calibration-only threshold representation]
    \label{prop:equiv-coin}
    Under Assumption~\ref{cond:exchangeability}, the two constructions
    with the same allocation $\bc$ satisfy
    \[
        \pr\Bigl\{
            \exists\,t\in[T]:
            Y_{n+1}\in
            \hcC_t^\coin(X_{n+1})
            \mathbin{\triangle}
            \hcC_t^{\rm thr}(X_{n+1})
        \Bigr\}
        =0.
    \]
    Here, $\triangle$ denotes symmetric difference.
\end{proposition}

The proof is given in Appendix~\ref{assec:proof-equivalent-coin}.
The threshold implementation therefore retains the AMR guarantee of
Corollary~\ref{the:errorcontrolcoin}.
Possible differences at candidate values tied with calibration scores do not affect this conclusion.

\paragraph*{Calibration and prediction.}
Algorithm~\ref{alg:coinprocedure} separates the calibration recursion
in~\eqref{eq:calibration-only-coin-thresholds} from prediction
in~\eqref{eq:calibration-only-coin-sets}. The first phase is run once
for each allocation and stores $T$ thresholds. The second phase
reuses them for each test point, evaluating the stage-$t$ score only
when the information required for that stage is available.
Reusing the thresholds can reduce the need for repeated calibration.
As in Section~\ref{ssec:coinprocedure}, the algorithm writes
$\hcC_t^\coin$ for $\hcC_t^{\rm thr}$.

\begin{figure}[H]
    \centering
    \refstepcounter{algorithm}
    \begin{minipage}{0.96\linewidth}
    \raggedright
    \hrule height 0.6pt
    \vspace{5pt}
    \textbf{Algorithm \thealgorithm.}
    \coinm: reusable calibration and stagewise prediction
    \vspace{5pt}
    \hrule
    \label{alg:coinprocedure}
    \vspace{4pt}

    {\renewcommand{\baselinestretch}{1.1}\small
    \textbf{Input:}
    Calibration data $\cD^\ca=\{Z_i=(X_i,Y_i):i\in[n]\}$;
    pre-trained score functions $\{s_t\}_{t=1}^T$;
    target AMR level $\alpha$; prespecified allocation
    $\bc=(c_1,\ldots,c_T)\in\mathbb{Z}_{\geq0}^T$ with
    $\sum_{t=1}^T c_t=\lfloor(n+1)\alpha\rfloor$.

    \smallskip
    \textbf{Calibration:}
    \begin{enumerate}[label=\arabic*:,leftmargin=2.2em,itemsep=1pt,topsep=2pt,
                      parsep=0pt,partopsep=0pt]
        \item Initialize $\cA_0=[n]$.
        \item \textbf{for} $t=1,\ldots,T$ \textbf{do}
        \item \hspace{1em}Compute $S_{t,i}=s_t(Z_i)$ for $i\in\cA_{t-1}$.
        \item \hspace{1em}Set
        $\htau_t(\bc)=\bbC_{c_t}(\{S_{t,i}:i\in\cA_{t-1}\})$,
        with $\bbC_0(\cdot)=+\infty$.
        \item \hspace{1em}Set
        $\cA_t=\{i\in\cA_{t-1}:S_{t,i}<\htau_t(\bc)\}$.
        \item \textbf{end for}; store $(\htau_1(\bc),\ldots,\htau_T(\bc))$.
    \end{enumerate}
    \smallskip
    \textbf{Prediction:} For each new test point $X_{n+1}$:
    \begin{enumerate}[label=\arabic*:,start=7,leftmargin=2.2em,itemsep=1pt,topsep=2pt,
                      parsep=0pt,partopsep=0pt]
        \item Initialize $\hcC_0^\coin(X_{n+1})=\cY$.
        \item \textbf{for} $t=1,\ldots,T$ \textbf{do}
        \item \hspace{1em}When stage-$t$ information is available, form
        $\tcC_t(X_{n+1})=
        \{y\in\cY:s_t(X_{n+1},y)\leq\htau_t(\bc)\}$.
        \item \hspace{1em}Report
        $\hcC_t^\coin(X_{n+1})=
        \hcC_{t-1}^\coin(X_{n+1})\cap\tcC_t(X_{n+1})$.
        \item \textbf{end for}.
    \end{enumerate}
    \smallskip
    \textbf{Output:} Stored calibration thresholds and the prediction
    sequence $\{\hcC_t^\coin(X_{n+1})\}_{t=1}^T$ for each test point.
    \par}
    \vspace{6pt}
    \hrule height 0.6pt
    \end{minipage}
\end{figure}

\subsection{Ordered score aggregation}
\label{ssec:numerical-aggregation}

When all score functions are available in advance and only the terminal set is
evaluated, process-level utility reduces to
$u_{\rm term}(\cC_{1:T})=-|\cC_T|$. In this special case, \vocoinm learns how
to distribute the exclusion-count budget across an ordered collection of
nonconformity scores. In this subsection, we first explain why complementary scores can improve
the terminal prediction set. Section~\ref{sssec:uci-datasets-analyses} then
evaluates the coverage and interval-length performance of \vocoinm on
real-world regression datasets. Section~\ref{assec:detailed-uci-analysis}
provides implementation details for these experiments, while
Section~\ref{assec:demo-aggr-multi-scores} describes the setup and additional
results for the illustrative example in Figure~\ref{fig:demo-aggr-multi-scores}.

\begin{assumption}
    \label{cond:densityfunction}
    For $P_X$-almost every $x\in\cX$, $P_{Y\mid X=x}$ admits a density
    $f_{Y\mid X}(y\mid x)$ with respect to a measure $\nu$. Moreover,
    $-f_{Y\mid X}(Y\mid X)$ has a continuous distribution.
\end{assumption}

Under this assumption, a minimum-$\nu$-measure oracle prediction set is
$\cC^\ora(x)=\{y:-f_{Y\mid X}(y\mid x)\leq\tau^*\}$, where
\[
    \tau^*
    \coloneqq
    \inf\bigl\{
        \tau:\pr_{XY}\{-f_{Y\mid X}(Y\mid X)\leq\tau\}\geq1-\alpha
    \bigr\};
\]
see Section~5.3.1 of \citet{angelopoulos2024theoretical}. Define the oracle
exclusion region
$\cO^*\coloneqq\{(x,y):-f_{Y\mid X}(y\mid x)>\tau^*\}$. For
$t\in[T]$, let $\cO_t(\tau)\coloneqq\{(x,y):s_t(x,y)>\tau\}$ and
\[
    \tau_t^*
    \coloneqq
    \inf\bigl\{
        \tau\in\overline{\bbR}:
        \pr_{XY}\{\cO_t(\tau)\setminus\cO^*\}=0
    \bigr\}.
\]

\begin{theorem}[Oracle recovery by complementary scores]
    \label{the:oracle-recovery}
    Suppose that Assumption~\ref{cond:densityfunction} holds and that the
    score values are almost surely free of ties. If
    \[
        \pr_{XY}\left\{
            \cO^*\mathbin{\triangle}
            \bigcup_{t=1}^T\cO_t(\tau_t^*)
        \right\}=0,
    \]
    then there exists $\bvrho^*\in\Theta_{\bvrho}$ such that
    \[
        \pr_{XY}\left\{
            Y\in
            \cC^\ora(X)\mathbin{\triangle}
            \bigcap_{t=1}^T
            \cC_t^\coin(X;\alpha\bvrho^*)
        \right\}=0.
    \]
\end{theorem}

The premise is an idealized condition involving the unknown oracle region, but
its implication is useful: aggregation can approach the oracle when the
score-specific exclusion regions cover complementary parts of that region.
The example in Figure~\ref{fig:demo-aggr-multi-scores} illustrates this
mechanism. The residual and scaled-residual scores are strongest in different
parts of the covariate space, and their combined exclusions yield a shorter
interval than either score alone.

\begin{figure}[!t]
    \setstretch{1.2}
    \centering
    \includegraphics[width=.9\linewidth]{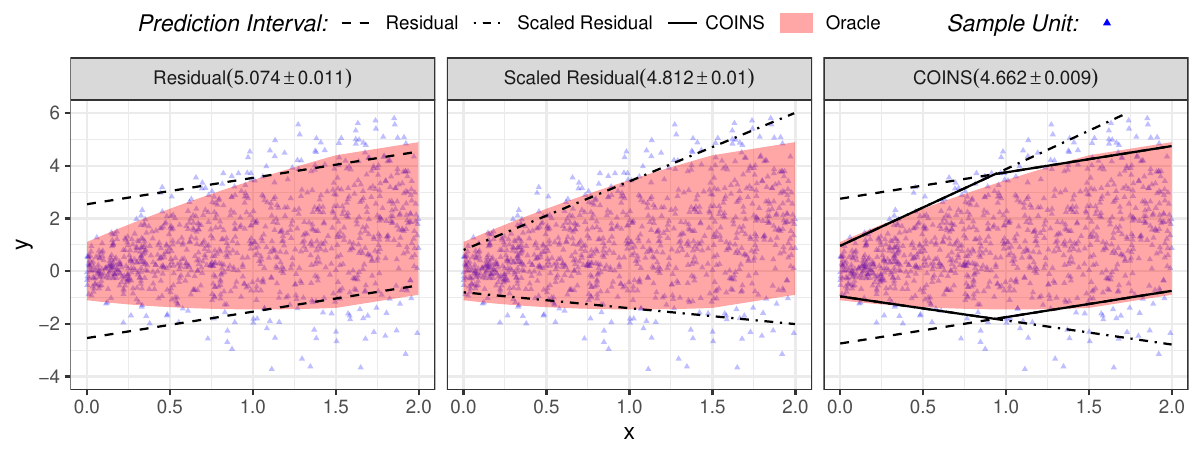}
    \caption{\small Aggregation of complementary nonconformity scores. The two scores
    have different local strengths, whereas \coinm combines their exclusions
    and approaches the oracle interval. Parenthetical values are average
    interval lengths over 1,000 repetitions; empirical coverages are 0.8998,
    0.9004, and 0.9004 at nominal coverage $1-\alpha=0.9$. Additional details
    are given in Appendix~\ref{assec:demo-aggr-multi-scores}.}
    \label{fig:demo-aggr-multi-scores}
\end{figure}

\subsubsection{UCI data analysis}
\label{sssec:uci-datasets-analyses}

We compare \vocoinm with five aggregation or score-selection methods:
\texttt{COLA-s} \citep{xu2025aggregating}, \texttt{TR-COINS},
\texttt{VFCP} \citep{yang2025selection}, \texttt{CSA}
\citep{patel2025conformal}, and \texttt{SACP++}
\citep{alami2026symmetric}. Here, \texttt{TR-COINS} transfers the allocation
learned by \texttt{COLA-s} to the \coinm calibration procedure.
All methods use separate tuning and final calibration samples.
For \vocoinm, the tuning sample serves as the strategy-learning sample.
We use eight public regression datasets from the UCI
Repository \citep{kelly2023uci}; for each of 500 random partitions, 1,000
observations are split in a $3{:}2{:}2{:}3$ ratio into training, tuning,
calibration, and test sets. Five nonconformity scores are constructed from
random-forest, XGBoost, and quantile-regression-forest models; the fifth
score is the conformalized quantile regression score. Complete
implementation details appear in Appendix~\ref{assec:detailed-uci-analysis}.

\begin{table*}[!t]
    \setstretch{1}
    \centering
    \caption{\small
        Empirical coverage (\%) computed over 500 random data partitions.
        Entries are Monte Carlo means with standard errors in parentheses.}
    \label{tab:multiscore-aggr-coverage}
    \begingroup
    \scriptsize
    \setlength{\tabcolsep}{2.2pt}
    \renewcommand{\arraystretch}{0.94}
    \resizebox{\textwidth}{!}{\begin{tabular}{l*{8}{c}}
    \toprule
    Method & UCI-001 & UCI-165 & UCI-265 & UCI-291 & UCI-304 & UCI-332 & UCI-440 & UCI-464 \\
    \midrule
    \vocoinm                  & 90.75 (.11) & 90.08 (.12) & 89.77 (.12) & 90.09 (.13) & 90.71 (.11) & 89.98 (.13) & 90.02 (.12) & 89.99 (.12) \\
    \texttt{COLA-s}           & 91.72 (.12) & 90.44 (.12) & 90.98 (.12) & 90.89 (.12) & 92.28 (.11) & 91.12 (.14) & 91.50 (.13) & 91.47 (.12) \\
    \texttt{TR-COINS}         & 90.75 (.11) & 89.97 (.12) & 89.76 (.12) & 90.21 (.12) & 91.11 (.10) & 90.00 (.12) & 90.18 (.12) & 90.00 (.12) \\
    \texttt{VFCP}             & 91.01 (.12) & 89.98 (.12) & 89.79 (.12) & 90.07 (.12) & 90.69 (.14) & 90.04 (.12) & 89.87 (.13) & 89.96 (.12) \\
    \texttt{CSA}          & 90.07 (.12) & 90.03 (.12) & 89.70 (.12) & 90.15 (.12) & 90.28 (.12) & 90.04 (.12) & 89.96 (.12) & 89.97 (.12) \\
    \texttt{SACP++}       & 90.04 (.12) & 90.08 (.12) & 89.77 (.12) & 90.09 (.13) & 90.14 (.12) & 90.09 (.12) & 90.00 (.12) & 89.88 (.12) \\
    $s_1$                     & 90.13 (.12) & 90.04 (.12) & 89.86 (.12) & 90.32 (.12) & 90.12 (.12) & 90.10 (.13) & 90.00 (.11) & 89.99 (.12) \\
    $s_2$                     & 90.16 (.12) & 90.01 (.12) & 90.01 (.12) & 90.06 (.12) & 90.03 (.11) & 90.04 (.13) & 90.03 (.12) & 89.93 (.12) \\
    $s_3$                     & 89.94 (.12) & 90.04 (.12) & 89.70 (.12) & 90.16 (.12) & 90.14 (.13) & 90.07 (.12) & 89.81 (.13) & 89.88 (.13) \\
    $s_4$                     & 89.90 (.12) & 90.27 (.12) & 89.82 (.12) & 90.21 (.12) & 90.10 (.12) & 89.99 (.13) & 90.12 (.13) & 90.02 (.12) \\
    $s_5$                     & 92.87 (.09) & 90.12 (.12) & 89.79 (.12) & 90.14 (.12) & 94.95 (.07) & 90.08 (.12) & 90.21 (.12) & 90.05 (.12) \\
    \bottomrule
    \end{tabular}}
    \endgroup
\end{table*}

As shown in Table~\ref{tab:multiscore-aggr-coverage}, all methods attain empirical coverage close to the target level of $0.9$.
Table~\ref{tab:multiscore-aggr-length} further shows that \vocoinm attains the shortest
average interval on five datasets and is within two Monte Carlo standard
errors of the shortest on a sixth. It also outperforms every individual score
on most datasets, supporting the value of combining complementary scores.

\begin{table*}[t]
    \setstretch{1}
    \centering
    \caption{\small
        Average prediction-interval length over 500 random data partitions.
        Entries are Monte Carlo means with standard errors in parentheses.
        The minimum and entries no more than two standard errors above the
        minimum are highlighted.}
    \label{tab:multiscore-aggr-length}
    \begingroup
    \scriptsize
    \setlength{\tabcolsep}{2.2pt}
    \renewcommand{\arraystretch}{0.94}
    \resizebox{\textwidth}{!}{\begin{tabular}{l*{8}{c}}
    \toprule
    Method & UCI-001 & UCI-165 & UCI-265 & UCI-291 & UCI-304 & UCI-332 & UCI-440 & UCI-464 \\
    \midrule
    \vocoinm                  & 2.123 (.012) & 1.106 (.005) & \bestcell{2.351 (.007)} & \bestcell{1.245 (.005)} & \bestcell{0.438 (.011)} & \bestcell{1.327 (.023)} & \bestcell{0.861 (.008)} & \bestcell{1.243 (.006)} \\
    \texttt{COLA-s}           & 2.279 (.013) & 1.125 (.006) & 2.442 (.008) & 1.300 (.006) & 0.585 (.020) & 1.581 (.028) & 0.994 (.011) & 1.352 (.006) \\
    \texttt{TR-COINS}         & 2.173 (.012) & 1.104 (.005) & 2.370 (.007) & 1.263 (.006) & 0.486 (.014) & 1.473 (.026) & 0.906 (.009) & 1.280 (.006) \\
    \texttt{VFCP}             & 2.209 (.011) & 1.105 (.006) & 2.416 (.006) & 1.275 (.006) & 0.721 (.012) & 1.603 (.023) & 1.058 (.010) & 1.394 (.006) \\
    \texttt{CSA}              & \bestcell{2.097 (.009)} & 1.132 (.005) & 2.456 (.009) & 1.276 (.005) & 0.496 (.009) & 1.397 (.021) & 0.954 (.008) & 1.271 (.006) \\
    \texttt{SACP++}           & \bestcell{2.086 (.009)} & 1.102 (.005) & 2.560 (.009) & \bestcell{1.243 (.005)} & 0.490 (.009) & 1.478 (.022) & 0.916 (.007) & 1.293 (.006) \\
    \midrule
    $s_1$                     & 2.244 (.011) & 1.290 (.006) & 2.832 (.009) & 1.537 (.007) & 0.943 (.020) & 1.936 (.028) & 1.258 (.013) & 1.667 (.008) \\
    $s_2$                     & 2.334 (.012) & \bestcell{1.090 (.005)} & 2.984 (.010) & 1.259 (.006) & 1.170 (.020) & 2.222 (.028) & 1.349 (.012) & 1.626 (.008) \\
    $s_3$                     & 2.200 (.011) & 1.293 (.006) & 2.748 (.012) & 1.465 (.006) & 0.722 (.013) & 1.607 (.024) & 1.033 (.009) & 1.385 (.007) \\
    $s_4$                     & 2.387 (.012) & 1.298 (.007) & 3.150 (.015) & 1.349 (.006) & 1.553 (.026) & 2.641 (.033) & 1.617 (.014) & 1.579 (.009) \\
    $s_5$                     & 2.186 (.010) & 1.417 (.004) & 2.402 (.005) & 1.664 (.006) & 0.862 (.011) & 1.733 (.020) & 1.320 (.009) & 1.414 (.004) \\
    \bottomrule
    \end{tabular}}
    \endgroup
\end{table*}

\subsubsection{Details of the UCI data analyses}\label{assec:detailed-uci-analysis}

This subsection provides additional implementation details for the UCI experiments in Section~\ref{sssec:uci-datasets-analyses}.
Table~\ref{tab:uci-dataset-information} summarizes the basic information about the eight datasets.
In what follows, we will introduce the preprocessing step, the construction of the nonconformity scores, and the implementations of the \texttt{CSA} and \texttt{SACP++} procedures.

\textbf{Preprocessing}. For each dataset and repetition, we sampled 1,000 observations without replacement and randomly divided them into 300 training, 200 tuning, 200 final calibration, and 300 test observations.
Numerical features were median-imputed and standardized, whereas categorical features were mode-imputed and one-hot encoded.
All preprocessing transformations were fitted using only the training observations.
The response was centered and scaled using its mean and sample standard deviation in the training split.

\begin{table*}[t]
    \setstretch{1}
    \centering
    \caption{\small Basic information for the eight datasets from the UCI Machine Learning Repository \citep{kelly2023uci}. Here, $n$ and $p$ denote the numbers of available observations and raw features, respectively.}
    \label{tab:uci-dataset-information}

    \begingroup
    \scriptsize
    \setlength{\tabcolsep}{3.5pt}
    \renewcommand{\arraystretch}{1.08}

    \resizebox{\textwidth}{!}{\begin{tabular}{llrrp{7.5cm}}
    \toprule
    UCI ID & Dataset & $n$ & $p$ & Response variable \\
    \midrule
    UCI-001 & Abalone                       & 4,177   & 8   & Number of shell rings, used as a proxy for age \\
    UCI-165 & Concrete Compressive Strength & 1,030   & 8   & Concrete compressive strength in MPa \\
    UCI-265 & Protein Tertiary Structure     & 45,730  & 9   & Root-mean-square deviation (RMSD) of the protein structure \\
    UCI-291 & Airfoil Self-Noise             & 1,503   & 5   & Scaled sound-pressure level in dB \\
    UCI-304 & BlogFeedback                   & 52,397  & 280 & Number of comments received in the following 24 hours \\
    UCI-332 & Online News Popularity         & 39,644  & 59  & Number of times an online news article was shared \\
    UCI-440 & SGEMM GPU Kernel Performance   & 241,600 & 14  & Mean GPU-kernel running time over four repeated runs, in ms \\
    UCI-464 & Superconductivity              & 21,263  & 81  & Critical temperature of the superconductor in K \\
    \bottomrule
    \end{tabular}}

    \endgroup
\end{table*}

\textbf{Construction of the five nonconformity scores.}
Let $\widehat\mu_{\mathrm{RF}}$ and
$\widehat\mu_{\mathrm{XGB}}$ denote the random-forest and XGBoost mean regressors, respectively;
let $\widehat\sigma_{\mathrm{RF}}$ and
$\widehat\sigma_{\mathrm{XGB}}$ denote the corresponding scale regressors; and let
$\widehat q_{\mathrm{L}}$ and $\widehat q_{\mathrm{U}}$ denote the lower and upper quantile-regression-forest estimates.
The five candidate scores are
\begin{equation*}
    \begin{aligned}
        s_1(x,y)
        & =
        \left|
            y-\widehat\mu_{\mathrm{RF}}(x)
        \right|,
        &
        s_2(x,y)
        =
        \left|
            y-\widehat\mu_{\mathrm{XGB}}(x)
        \right|,
        \\
        s_3(x,y)
        &=
        \frac{
            \left|y-\widehat\mu_{\mathrm{RF}}(x)\right|
        }{
            \widehat\sigma_{\mathrm{RF}}(x)
        },
        &s_4(x,y)
        =
        \frac{
            \left|y-\widehat\mu_{\mathrm{XGB}}(x)\right|
        }{
            \widehat\sigma_{\mathrm{XGB}}(x)
        },
        \\
        s_5(x,y)
        &=
        \max\left\{
            \widehat q_{\mathrm{L}}(x)-y,\,
            y-\widehat q_{\mathrm{U}}(x)
        \right\}.
    \end{aligned}
\end{equation*}
Here, $\widehat q_{\mathrm{L}}$ and $\widehat q_{\mathrm{U}}$ estimate the conditional
$\alpha/2$ and $1-\alpha/2$ quantiles, respectively.
Thus, $s_5$ is the conformalized quantile regression score of \citet{romano2019conformalized}.

To construct the targets for the scale regressors, we first obtain five-fold out-of-fold mean predictions and define
\begin{equation*}
    R_i^{\mathrm{oof}}
    =
    \frac{1}{2}
    \left\{
        \left|
            Y_i-\widehat\mu_{\mathrm{RF}}^{\mathrm{oof}}(X_i)
        \right|
        +
        \left|
            Y_i-\widehat\mu_{\mathrm{XGB}}^{\mathrm{oof}}(X_i)
        \right|
    \right\}.
\end{equation*}
Random-forest and XGBoost scale regressors are then fitted to
$\{(X_i,R_i^{\mathrm{oof}})\}$, with their predictions lower-bounded by $0.02$ on the standardized response scale.
Moreover, the dataset-specific hyperparameters are fixed before the Monte Carlo experiment.
For the sequential intersection methods, the five scores are ordered within each replication by increasing mean length of their individually calibrated intervals computed from the out-of-fold training repetitions.

\textbf{Implementations of \texttt{CSA} and \texttt{SACP++}.}
Both methods require nonnegative nonconformity scores while $s_5$ may produce a negative value.
Therefore, we first transform the fifth score as follows.
Using five-fold out-of-fold predictions on the training split, we compute the median $m_5$ and interquartile range $r_5$ of the fifth training score.
We then define
\begin{equation*}
    \widetilde{s}_5(x,y)
    =
    \log\left[
        1+\exp\left\{
            \frac{s_5(x,y)-m_5}{r_5}
        \right\}
    \right],
\end{equation*}
and write
$\widetilde{\bs}(x,y)
:=
(s_1(x,y),\ldots,s_4(x,y),\widetilde{s}_5(x,y))^\top$.
This softplus transformation ensures that the transformed fifth score is positive; its parameters $m_5$ and $r_5$ are computed exclusively from the training split.
The remaining intersection-based methods operate directly on the raw scores.

For \texttt{CSA} \citep{patel2025conformal}, we first independently generate
$M=100$ directions $\{\bu_m\}_{m=1}^M$ uniformly over the positive orthant of the unit sphere in $\bbR^5$.
For any $\beta\in[\alpha/M,\alpha]$, let $q_m(\beta)$ be the order
statistic with rank $\min\{n_A,\lceil(n_A+1)(1-\beta)\rceil\}$ among
$
    \left\{
        \bu_m^\top\widetilde{\bs}(X_i,Y_i):
        i\in\cI_A
    \right\},
$
where $\cI_A$ denotes the tuning split and $n_A=|\cI_A|=200$.
Then, a 30-step binary search selects the largest $\widehat{\beta}$ for which the resulting quantile envelope has empirical coverage at least $1-\alpha$ on the tuning split:
\begin{equation*}
    \frac{1}{n_A}
    \sum_{i\in\cI_A}
    \bbI\left\{
        \bu_m^\top\widetilde{\bs}(X_i,Y_i)
        \le q_m(\widehat{\beta}),
        \ \forall\,m\in[M]
    \right\}
    \ge 1-\alpha.
\end{equation*}
This envelope induces the scalar score
\begin{equation*}
    g_{\mathrm{CSA}}(x,y)
    =
    \max_{m\in[M]}
    \frac{
        \bu_m^\top\widetilde{\bs}(x,y)
    }{
        q_m(\widehat{\beta})
    }.
\end{equation*}
Finally, $g_{\mathrm{CSA}}$ is used as a nonconformity score on the independent  calibration split $\cI_C$, and the prediction set is
$ \hcC_{\mathrm{CSA}}(x) =
\left\{
    y\in\cY:
    g_{\mathrm{CSA}}(x,y)
    \le
    \widehat q_{\mathrm{CSA}}
\right\}$,
where $\widehat q_{\mathrm{CSA}}$ is the usual split-conformal upper quantile computed from
$\{g_{\mathrm{CSA}}(X_i,Y_i):i\in\cI_C\}$.

For \texttt{SACP++} \citep{alami2026symmetric}, consider a calibration dataset $\{(X_i, Y_i)\}_{i = 1}^n$ and an unlabeled feature vector $X_{n+1} \in \cX$.
For each $y \in \cY$, let $Z_i^y=Z_i$ for $i \in [n]$ and $Z_{n+1}^y=(X_{n+1},y)$.
For each score coordinate $t\in[5]$, define
\begin{equation*}
    E_i^t(y)
    =
    \frac{(n+1) \tilde{s}_t(Z_i^y)}
    {\sum_{j=1}^{n+1} \tilde{s}_t(Z_j^y)},
    \qquad
    i\in[n+1],\quad t\in[5],
\end{equation*}
where $\tilde{s}_t = s_t$ for $t \le 4$.
We then aggregate the resulting $e$-vectors using
\begin{equation*}
    G_{p,i}(y)
    =
    \bigg\{
        \sum_{t=1}^5
        \bigl(E_i^t(y)\bigr)^p
    \bigg\}^{1/p},
    \qquad
    p > 0,\quad i\in[n+1].
\end{equation*}
For each candidate $y$, let
$\widehat q_p(y) = G_{p,(\ell)}(y)$,
where $ \ell =\lceil(n+1)(1-\alpha)\rceil$ and $G_{p,(\ell)}(y)$ denotes the $\ell$th order statistic of
$\{G_{p,i}(y): i \in [n]\}$.
The candidate $y$ is included in the prediction set if and only if
$G_{p,n+1}(y) \le \widehat q_p(y)$.

In the split-tuned implementation of \texttt{SACP++},
we first select $\widehat p$ from the candidate set $\{0.5,1,2,4,8,16\}$ on the tuning split by minimizing the average prediction-set length.
Holding $\widehat p$ fixed, we then use the independent calibration split to construct the final prediction set.

\subsubsection{Details for the demonstration in Figure~\ref{fig:demo-aggr-multi-scores}}
\label{assec:demo-aggr-multi-scores}

In this subsection, we provide implementation details and additional results for the illustrative example in Figure~\ref{fig:demo-aggr-multi-scores}.

Specifically, we consider the heteroscedastic regression model
\begin{equation*}
    X \sim \operatorname{Unif}(0,2),
    \quad
    Y = \mu(X) + \sigma(X)\varepsilon,
    \quad
    \mu(x) = x,
    \quad
    \sigma(x)=\min\{x+1/2,2\},
    \quad
    \varepsilon \sim \mathcal{N}(0,1),
\end{equation*}
where $\varepsilon$ is independent of $X$.
Two nonconformity scores are considered:
$s_{\mathrm{res}}(x,y) = |y-\mu(x)|$ and $s_{\mathrm{sca}}(x,y) = |y-\mu(x)|/(x+1/2)$.
Consequently, neither score is uniformly more efficient over the entire covariate space:
The residual score tends to produce unnecessarily wide intervals for small values of $x$, whereas the scaled residual score becomes relatively inefficient near the right boundary of the covariate space. This complementary behavior motivates aggregating the two scores.
For comparison, the oracle interval is included, defined as the conditional-density level set
\begin{equation*}
    \cC^\ora(x)
    =
    \left\{
        y:
        f_{Y\mid X}(y\mid x)\geq\lambda_{\mathrm{ora}}
    \right\},
\end{equation*}
where $\lambda_{\mathrm{ora}}=0.07$ is chosen numerically to attain marginal coverage of approximately $0.9$.
This oracle interval is included only as an efficiency benchmark because its construction requires knowledge of the true conditional distribution.

For each repetition, we generate 2,000 independent observations, which are evenly split into a calibration set and a test set.
We set $\alpha=0.1$, so the total exclusion-count budget is
$c_{\mathrm{tot}} = \lfloor (n+1)\alpha \rfloor = 100$.
We use the prespecified allocation $(c_1,c_2)=(80,20)$ for \coinm.
All metrics are computed over 1,000 independent repetitions.
Table~\ref{tab:demo-aggr-multi-scores-additional} summarizes the results.
It shows that all three methods attain empirical coverage close to the nominal level of $0.9$, while \coinm produces the shortest average prediction interval.

\begin{table}[!ht]
    \setstretch{1}
    \centering
    \caption{\small
        Average interval lengths and empirical marginal coverage rates for the illustrative aggregation example.
    }
    \label{tab:demo-aggr-multi-scores-additional}
    \begin{tabular}{lcc}
        \toprule
        Method
        & Average length
        & Empirical coverage \\
        \midrule
        Residual
        & 5.074
        & 0.8998 \\
        Scaled residual
        & 4.812
        & 0.9004 \\
        \coinm
        & 4.662
        & 0.9004 \\
        \bottomrule
    \end{tabular}
\end{table}

\section{Nested-Quantile Results}
\label{asec:nested-quantile-results}

This section develops the population nested quantiles introduced in
Section~\ref{ssec:population-utility-loss} and compares them with thresholds
estimated from a finite sample. We distinguish the direct empirical
quantiles $\hbq$ from the thresholds $\hbtau$ actually used by \coinm,
which incorporate the finite-sample rank correction and integer exclusion
counts. Proofs of the results are collected in Section~\ref{asec:proofs-appendix}.
Specifically, let $\bS = (S_1, \cdots, S_T)$ be an $\bbR^T$-valued random vector with distribution $P_\bS$ on $(\bbR^T, \cB(\bbR^T))$, and let $(\bS_1,\cdots,\bS_n)$ be an i.i.d.\ sample of size $n$ from $P_\bS$.
Let $F_\bS$ denote the joint CDF of $\bS$, and let $F_{n,\bS}$ denote the empirical joint CDF based on $(\bS_1,\cdots,\bS_n)$.
For $t=1,\ldots,T$, let $\bS(\ole t) = (S_1, \cdots, S_t)$ denote the subvector consisting of the first $t$ components of $\bS$, and let $F_{{\scriptscriptstyle \bS(\ole t)}}$ and $F_{n,{\scriptscriptstyle \bS(\ole t)}}$ denote its CDF and the corresponding empirical CDF, respectively.

For an allocation proportion vector $\bvrho\in\Theta_{\bvrho}$ and a target
level $\alpha\in(0,1)$, the population recursion is, for each $t\in[T]$,
\begin{equation}\label{eq:def-nested-quant}
    q_t (\alpha \bvrho)
    \! := \!
    \begin{cases}
        \infty, & \varrho_t = 0,
        \\
        \inf \left\{
            v \in \overline{\bbR}: \pr_{P_{\bS}} \left(
                S_t \le v, S_{t'} \le q_{t'} (\alpha  \bvrho), \forall \, t' < t
            \right)
            \ge
            1 - \alpha \beta_t(\bvrho)
        \right\},
        & \varrho_t > 0,
    \end{cases}
\end{equation}
where $\beta_t(\bvrho) = \sum_{t'=1}^t\varrho_{t'}$ is the cumulative allocation proportion up to stage~$t$.
We denote $\bq (\alpha \bvrho) = \big(q_1(\alpha \bvrho), \cdots, q_T(\alpha \bvrho) \big)$.
Replacing the population distribution in~\eqref{eq:def-nested-quant} with
the empirical distribution, and using the preceding empirical thresholds
at each stage, defines the plug-in nested quantiles
$\hbq(\alpha\bvrho)=(\hq_1(\alpha\bvrho),\ldots,\hq_T(\alpha\bvrho))$.
These plug-in quantiles need not equal the \coinm thresholds in
Appendix~\ref{ssec:implementcoin}. Proposition~\ref{prop:conservativequantile}
accounts for the finite-sample correction and compares the actual
\coinm thresholds directly with the population quantiles.

The following theorem characterizes the nonasymptotic behavior of $\hbq (\alpha \bvrho)$ under the non-atomic assumption.

\begin{theorem}\label{the:errorofnestedquantile}
    Suppose that every coordinate marginal distribution of $P_{\bS}$
    is non-atomic. Let $(\bS_1,\ldots,\bS_n)$ be an i.i.d.\ sample from
    $P_{\bS}$. If $n > 15, 1 < T < n\log n /9$,
    then, with probability at least $1-n^{-1}$, the following inequalities
    hold simultaneously for every $t\in[T]$:
    \begin{equation*}\begin{aligned}
            \sup_{\alpha\in(0,1), \bvrho\in\Theta_{\bvrho}}
            \bigg|
            F_{\scriptscriptstyle \bS(\ole t)}
            \Big(
                q_1(\alpha\bvrho),\ldots,
                q_{t-1}(\alpha\bvrho),
                \widehat q_t(\alpha\bvrho)
            \Big)
            -
            F_{{\scriptscriptstyle \bS(\ole t)}}
            \Big(
                q_1(\alpha\bvrho),\ldots,
                q_t(\alpha\bvrho)
            \Big)
            \bigg|
            \le
            3t\sqrt{\frac{T\log n}{n}},
        \end{aligned}
    \end{equation*}
    and
    \begin{equation*}
        \sup_{\alpha\in(0,1), \bvrho\in\Theta_{\bvrho}}
        \pr_{P_{\bS}}
        \bigg[
        \Big(
            \bigcap_{j=1}^t
            \{S_j\le\widehat q_j(\alpha\bvrho)\}
        \Big)
        \triangle
        \Big(
            \bigcap_{j=1}^t
            \{S_j\le q_j(\alpha\bvrho)\}
        \Big)
        \bigg]
        \le
        3t\sqrt{\frac{T\log n}{n}}.
    \end{equation*}
\end{theorem}

We now compare the population nested quantiles with the empirical thresholds
used by \coinm. Fix $\alpha\in(0,1)$, and let
$\cD^\ca=\{(X_i,Y_i)\}_{i=1}^n$ be the calibration sample underlying
$(\bS_1,\ldots,\bS_n)$; that is, $S_{t,i}=s_t(X_i,Y_i)$. For
$\bvrho\in\Theta_{\bvrho}$, write the thresholds produced by the
calibration-only recursion in Appendix~\ref{ssec:implementcoin} as
\begin{equation}
    \label{eq:def-coin-empirical-threshold}
    \htau_t(\alpha\bvrho)
    \coloneqq
    \htau_t(\alpha\bvrho;\cD^\ca)
    \coloneqq
    \htau_t\bigl(\bc(\bvrho;n);\cD^\ca\bigr),
    \qquad t\in[T],
\end{equation}
where $\bc(\bvrho;n)$ is the exclusion-count allocation defined in
Section~\ref{sec:opt-via-erm}. Recall from
Section~\ref{ssec:population-utility-loss} that $q_t(\balp)$ is defined for any
nonnegative stagewise budget vector $\balp$ whose entries sum to less than
one. The following proposition uniformly brackets the empirical thresholds
in~\eqref{eq:def-coin-empirical-threshold} by two such population nested
quantiles.

\begin{proposition}
    \label{prop:conservativequantile}
    Suppose that the assumptions and sample-size conditions of Theorem~\ref{the:errorofnestedquantile} hold.
    Denote $a_n \coloneqq 5 \sqrt{T \log n / n}$.
    If $a_n T (T+1) / 2 < 1-\alpha$, then,
    with probability at least $1-n^{-1}$,
    simultaneously for every $\bvrho\in\Theta_{\bvrho}$ and $t\in[T]$,
    \begin{equation*}
        q_t\left(
            \balp_n^-(\bvrho)
        \right)
        \ge
        \htau_t\bigl(\bc(\bvrho;n);\cD^\ca\bigr)
        \ge
        q_t\left(
            \balp_n^+(\bvrho)
        \right),
    \end{equation*}
    where
    $\balp_n^-(\bvrho)=(\alpha_{n,t}^-(\bvrho))_{t=1}^T$ with
    $\alpha_{n,t}^-(\bvrho)=(\alpha\varrho_t-ta_n)_+$, and
    $\balp_n^+(\bvrho)=(\alpha_{n,t}^+(\bvrho))_{t=1}^T$ with
    $\alpha_{n,t}^+(\bvrho)=\alpha\varrho_t+ta_n$.
\end{proposition}

Finally, we compare the two prediction sequences produced by the population nested quantiles and the empirical thresholds used by \coinm.
Specifically, thresholds in~\eqref{eq:def-coin-empirical-threshold} induce
\begin{equation}
    \label{eq:coin-threshold-representation}
    \hcC_t^\coin
    \big(
        x;\cD^\ca,\bc(\bvrho;n)
    \big)
    =
    \left\{
        y\in\cY:
        s_{t'}(x,y)
        \le
        \htau_{t'}\bigl(\bc(\bvrho;n);\cD^\ca\bigr),
        \ \forall\,t'\le t
    \right\}.
\end{equation}
Meanwhile, the population nested quantiles induce the population-level prediction sequence as
\begin{equation*}
    \cC_t^\coin(x;\balp)
    :=
    \left\{
        y\in\cY:
        s_{t'}(x,y)
        \le
        q_{t'}(\balp),
        \ \forall\,t'\le t
    \right\},
\end{equation*}
for any nonnegative $\balp=(\alpha_1,\ldots,\alpha_T)$ satisfying
$\sum_{t=1}^T\alpha_t<1$.
Then, Proposition~\ref{prop:conservativequantile} immediately implies a utility comparison, which is formalized in the following corollary.

\begin{corollary}\label{cor:nested-quantile-utility-comparison}
    Suppose that the conditions of
    Proposition~\ref{prop:conservativequantile} hold and that $u$ is
    set-wise monotone in the sense of
    \eqref{eq:monotone-utility}. Then, with probability at least
    $1-n^{-1}$, simultaneously for every
    $\bvrho\in\Theta_{\bvrho}$ and $x\in\cX$,
    \begin{equation*}
        u\Big(
            \cC_{1:T}^\coin \big( x; \balp_n^-(\bvrho) \big)
        \Big)
        \le
        u\Big(
            \hcC_{1:T}^\coin \big( x; \cD^\ca, \bc(\bvrho;n) \big)
        \Big)
        \le
        u\Big(
            \cC_{1:T}^\coin \big( x; \balp_n^+(\bvrho) \big)
        \Big).
    \end{equation*}
\end{corollary}
\section{Uniform Convergence Results}\label{asec:uniform-convergence-delta}

This section considers two settings in which the function classes induced by
the threshold-indexed prediction sets and their utility functions are
Glivenko--Cantelli, implying that the required uniform deviation is
$o_p(1)$ under suitable growth conditions.
Their proofs are collected in
Sections~\ref{assec:proof-uniform-utility-class}
and~\ref{assec:proof-uniform-aggregated-size}.

The first result establishes uniform convergence for bounded utility functions in classification:

\begin{lemma}[Uniform convergence of bounded utility for classification]
    \label{lem:uniform-utility-classification}
    Suppose that $X_1,\ldots,X_n$ are drawn i.i.d.\ from $P_X$, that the score functions $\{s_t\}_{t = 1}^T$ are trained on data independent of $\{X_i\}_{i=1}^n$, and that $\cY=[M]$.
    Suppose also that
    \begin{equation*}
        0
        \le
        u\big(\cC_{1:T}\big)
        \le
        U_{\rm max} < \infty,
    \end{equation*}
    for every prediction sequence $\{\cC_t\}_{t=1}^T$.
    If $n>15$, then, with probability at least $1-n^{-1}$,
    \begin{equation}
        \label{eq:uniform-utility-classification}
        \begin{aligned}
        \sup_{\btau\in\overline{\bbR}^T}
        \left|
            \frac{1}{n}
            \sum_{i=1}^n
            u\left(
                \cC_{1:T}(X_i;\btau)
            \right)
            -
            \E_X
            \left[
                u\left(
                    \cC_{1:T}(X;\btau)
                \right)
            \right]
        \right|
        \le
        3U_{\rm max}
        \sqrt{
            \frac{T\log(2nM)}{n}
        }.
        \end{aligned}
    \end{equation}
\end{lemma}

Intuitively, for a fixed sample, although $\btau$ ranges over a continuous space, the induced prediction sets change only when a coordinate $\tau_t$ crosses one of the values in $\{s_t(X_i,y): i\in[n],\,y\in[M]\}$.
Consequently, the threshold vectors induce at most $(nM+1)^T$ distinct prediction-set configurations on the sample.
This finite reduction yields a uniform deviation bound for any bounded process-level utility.

The early-resolution utility is a direct instance:
\[
    u_{\rm res}\bigl(\cC_{1:T}\bigr)
    =
    \frac{
        T+1-\min\{t\in[T]:|\cC_t|\leq1\}
    }{T},
\]
with $\min\varnothing=T+1$. This utility takes values in $[0,1]$, so
Lemma~\ref{lem:uniform-utility-classification} applies with
$U_{\rm max}=1$.

The second lemma establishes the uniform law of large numbers underlying $\Delta_{n_1}$ for the terminal utility:

\begin{lemma}[Uniform convergence of aggregated-set size]
    \label{lem:uniform-aggregated-set-size}
    Suppose that $X_1,\cdots,X_n$ are generated i.i.d.\ from $P_X$ and $\{s_t\}_{t = 1}^T$ are trained on independent data.
    Suppose also that either $\cY=[0,M]$, with set size given by Lebesgue measure, or $\cY=[M]$, with set size given by cardinality.
    If $n>15$, then, conditional on the independently trained score functions, with probability at least $1-n^{-1}$,
    \begin{equation}
        \label{eq:uniform-aggregated-set-size}
        \sup_{\btau\in\overline{\bbR}^T}
        \left|
            \frac{1}{n}\sum_{i=1}^{n}
            \big|\cC_T(X_i;\btau)\big|
            -
            \E_X\big[\big|\cC_T(X;\btau)\big|\big]
        \right|
        \le
        5M\sqrt{\frac{T\log n}{n}}.
    \end{equation}
\end{lemma}

For the terminal utility used in ordered score aggregation,
\[
    u_{\rm term}\bigl(\cC_{1:T}(x;\btau)\bigr)
    =
    -\left|\bigcap_{t=1}^T\cC_t(x;\btau)\right|
    =
    -|\cC_T(x;\btau)|,
\]
where the second equality follows from nestedness. Hence
Lemma~\ref{lem:uniform-aggregated-set-size} gives the corresponding bound on
$\Delta_n$.

\section{Technical Proofs for Main-Text Results}
\label{asec:proofs-main}

The proofs are presented in the order in which the corresponding
results appear in the main text.

\subsection{Proof of Theorem~\ref{the:basicresult}}

\paragraph*{Proof idea.}
The proof has two steps. First, at $y=Y_{n+1}$,
Condition~\ref{p3:per-inv} preserves the uniform assignment of surviving
observation values to surviving indices conditional on
$\cF_t^{Y_{n+1}}$.
If $d$ of the $m$ survivors are excluded at the next stage, the surviving
test observation is therefore excluded with conditional probability $d/m$.
Second, a backward recursion uses the fixed total exclusion count in
Condition~\ref{p2:reje-count} to obtain an eventual exclusion
probability of $c_{\rm tot}/(n+1)$.
By Condition~\ref{p1:nondecreasing}, this eventual exclusion event is
equivalent to any-stage miscoverage, as in~\eqref{eq:equiv-two-events}.
We now formalize these steps.

\begin{proof}[Proof of Theorem~\ref{the:basicresult}]
    By~\eqref{eq:pred-set-from-seq-test} and monotone exclusion,
    \[
        \left\{
            \exists\,t\in[T]:
            Y_{n+1}\notin\hcC_t(X_{n+1})
        \right\}
        =
        \left\{
            \delta_T^{Y_{n+1}}(Z_{n+1})=1
        \right\}.
    \]
    Thus, it suffices to evaluate all quantities at $y=Y_{n+1}$.
    Here $\{\cF_t^{Y_{n+1}}\}_{t=0}^T$ is the filtration obtained by
    applying the recursion in Condition~\ref{p3:per-inv} to the actual
    augmented sample $\ocD(Y_{n+1})$.  The measurability clause in that
    condition ensures that the resulting exclusion sets are well defined.

    We first note that, conditional on $\cF_t^{Y_{n+1}}$, the observations that have not been excluded are exchangeable over the remaining indices.
    Formally, for every $i\in[n+1]$ and every measurable set
    $A\subseteq\cX\times\cY$,
    \begin{equation*}
        \pr\left\{
            Z_i^{Y_{n+1}}\in A
            \mid
            \cF_t^{Y_{n+1}},
            \, i \notin \cR_t^{Y_{n+1}}
        \right\}
        =
        \frac{
            \sum_{j\in[n+1]\setminus\cR_t^{Y_{n+1}}}
            \bbI\big\{Z_j^{Y_{n+1}}\in A\big\}
        }{
            (n+1)-|\cR_t^{Y_{n+1}}|
        }.
    \end{equation*}
    At $t=0$, this follows from Assumption~\ref{cond:exchangeability} after conditioning on the unordered augmented dataset and the score functions.
    Moreover, if the property holds at stage~$t-1$, the common value-based rule at stage~$t$ selects a subset from the multiset of the surviving observations using only $\cF_{t-1}^{Y_{n+1}}$ and not their index labels.
    Therefore, revealing $\cR_t^{Y_{n+1}}$ and the multiset of excluded observations leaves the assignment of the remaining observations to the remaining indices uniform.
    The result follows by induction.

    For $t=0,\ldots,T-1$, let
    $d_{t+1}\coloneqq
    |\cR_{t+1}^{Y_{n+1}}\setminus\cR_t^{Y_{n+1}}|$ and
    $c_{\rm tot}\coloneqq\lfloor(n+1)\alpha\rfloor$.
    The quantity $d_{t+1}$ is $\cF_t^{Y_{n+1}}$-measurable because the
    stage-$(t+1)$ exclusion decisions use only $\cF_t^{Y_{n+1}}$, which
    determines the multiset of surviving observations.
    Hence, conditional exchangeability gives
    \begin{equation*}
        \begin{aligned}
            & \,  \pr \Big\{
                n+1 \in \cR_{t+1}^{Y_{n+1}}
                \mid
                \cF_t^{Y_{n+1}}, n+1 \notin \cR_t^{Y_{n+1}}
            \Big\}
            \\
            = & \,
            \pr\Big\{
                n+1 \in \cR_{t+1}^{Y_{n+1}} \setminus \cR_t^{Y_{n+1}}
                \mid
                \cF_t^{Y_{n+1}}, n+1 \notin \cR_t^{Y_{n+1}}
            \Big\}
            =
            \frac{d_{t+1}}{(n+1)-|\cR_t^{Y_{n+1}}|}.
        \end{aligned}
    \end{equation*}

    We now show by backward induction that, on
    $\{n+1\notin\cR_t^{Y_{n+1}}\}$,
    \begin{equation*}
        \pr\left\{
            n+1 \in \cR_T^{Y_{n+1}}
            \mid
            \cF_t^{Y_{n+1}}
        \right\}
        =
        \frac{c_{\rm tot}-|\cR_t^{Y_{n+1}}|}{(n+1)-|\cR_t^{Y_{n+1}}|}.
    \end{equation*}
    At $t=T$, both sides are zero.
    Suppose the equality holds at stage $t+1$.
    Conditional on $\cF_t^{Y_{n+1}}$ and $n+1 \notin \cR_t^{Y_{n+1}} $, the test index
    either is newly excluded at stage $t+1$ or remains among the
    surviving indices.
    Using the preceding exclusion probability and the induction
    hypothesis gives
    \begin{equation*}
        \begin{aligned}
        & \, \pr\left\{
            n+1 \in \cR_T^{Y_{n+1}}
            \mid
            \cF_t^{Y_{n+1}}
        \right\}
        \\ = & \,
        \frac{d_{t+1}}{(n+1)-|\cR_t^{Y_{n+1}}|}
        + \frac{(n+1)-|\cR_t^{Y_{n+1}}|-d_{t+1}}
             {(n+1)-|\cR_t^{Y_{n+1}}|} \times
        \frac{c_{\rm tot} - |\cR_t^{Y_{n+1}}|-d_{t+1}}
             {(n+1)-|\cR_t^{Y_{n+1}}|-d_{t+1}}
        \\
        = & \,
        \frac{c_{\rm tot} - |\cR_t^{Y_{n+1}}|}
             {(n+1)-|\cR_t^{Y_{n+1}}|}.
        \end{aligned}
    \end{equation*}
    This proves the backward-induction claim.

    Since $\cR_0^{Y_{n+1}} = \varnothing$, it follows that
    \begin{equation*}
        \pr\{n+1\in\cR_T^{Y_{n+1}}\}
        =
        \frac{c_{\rm tot}}{n+1}
        =
        \frac{\lfloor(n+1)\alpha\rfloor}{n+1}.
    \end{equation*}
    Combining this identity with the first display proves the result.
\end{proof}

\subsection{Proof of Theorem~\ref{the:universality-principles}}

\begin{proof}[Proof of Theorem~\ref{the:universality-principles}]
    The classical universality principle is due to
    \citet{vovk2005algorithmic}. We adapt the leave-one-out proof
    presented in Theorem~9.7 of \citet{angelopoulos2024theoretical}
    to an entire nested prediction sequence. The leave-one-out
    construction below reproduces the original reports on every
    augmented sample. The exact-count argument is restricted to
    augmented samples of pairwise distinct observations, as specified
    in Theorem~\ref{the:universality-principles}.

    For a dataset $\cD$ and a point $z$, let $\cD_{\setminus z}$ be obtained by removing one copy of $z$ if it appears in $\cD$, and leave $\cD$ unchanged otherwise.
    Then, for every $y\in\cY$, $t\in[T]$, and
    $z=(x',y')\in\cX\times\cY$, define
    \begin{equation*}
        \delta_t^y(z)
        :=
        \bbI\big\{
            y'\notin\cC_t(x';\ocD(y)_{\setminus z})
        \big\},
        \qquad y\in\cY,\ t\in[T].
    \end{equation*}
    Since removing $Z_{n+1}^y=(X_{n+1},y)$ from $\ocD(y)$ leaves $\cD^\ca$ up to a permutation,
    we have $\delta_t^y(Z_{n+1}^y) = \bbI\{y\notin\cC_t(X_{n+1};\cD^\ca)\}$.
    Consequently, these exclusion functions induce the same prediction sequence $\{\cC_t\}_{t = 1}^T$:
    \begin{equation*}
        \hcC_t(X_{n+1})
        =
        \big\{
            y\in\cY:
            \delta_t^y(Z_{n+1}^y)=0
        \big\}
        =
        \cC_t(X_{n+1};\cD^\ca).
    \end{equation*}
    It remains to verify the three structural conditions in order.

    \emph{Condition~\ref{p1:nondecreasing}.}
    If $\delta_{t-1}^y(z)=1$, then
    $y'\notin\cC_{t-1}(x';\ocD(y)_{\setminus z})$.
    Nestedness implies
    $y'\notin\cC_t(x';\ocD(y)_{\setminus z})$, and hence
    $\delta_t^y(z)=1$.
    Therefore, $\delta_{t-1}^y(z)\leq\delta_t^y(z)$ for every $t>1$,
    $y\in\cY$, and $z\in\cX\times\cY$.

    \emph{Condition~\ref{p2:reje-count}.}
    Fix $y\in\cY$ such that $\ocD(y)$ consists of pairwise distinct
    observations, and regard the augmented multiset
    $\Lbag\ocD(y)\Rbag$ as fixed.
    Write $Z_i^y=(X_i^y,Y_i^y)$, and let $\pi$ be a uniformly random
    permutation of $[n+1]$.
    Thus, the random sequence
    $(Z_{\pi(1)}^y,\cdots,Z_{\pi(n+1)}^y)$ is exchangeable.
    Let $\{Z_{\pi(j)}^y\}_{j=1}^n$ serve as the calibration data and
    $Z_{\pi(n+1)}^y$ as the test point.
    A permutation preserves distinctness, so the assumed AMR equality
    applies to this randomly permuted sequence.
    Consequently, we have
    \begin{equation}\label{eq:apply-assumed-amr}
        \pr_\pi\Big\{
            \exists\,t\in[T]:
            Y_{\pi(n+1)}^y
            \notin
            \cC_t\big(
                X_{\pi(n+1)}^y;
                \{Z_{\pi(j)}^y\}_{j=1}^n
            \big)
        \Big\}
        =
        \frac{\lfloor(n+1)\alpha\rfloor}{n+1},
    \end{equation}
    where $\pr_\pi$ denotes probability over $\pi$.

    Now consider the event $\{\pi(n+1)=i\}$.
    Conditional on this event, the calibration data contain precisely all observations in $\ocD(y)$ except one copy of $Z_i^y$;
    in particular, we have
    $\Lbag\{Z_{\pi(j)}^y\}_{j=1}^n\Rbag
    =
    \Lbag\ocD(y)_{\setminus Z_i^y}\Rbag$.
    By symmetry of $\cC_T$, the ordering of these $n$ calibration observations is irrelevant.
    Hence, conditional on $\pi(n+1)=i$, the final-stage exclusion indicator satisfies
    \begin{equation*}
        \bbI\big\{
            Y_i^y
            \notin
            \cC_T(X_i^y;\ocD(y)_{\setminus Z_i^y})
        \big\}
        =
        \delta_T^y(Z_i^y).
    \end{equation*}
    Since $\pi(n+1)$ is uniform over $[n+1]$, averaging over the possible
    test indices yields
    \begin{equation*}
        \pr_\pi\Big\{
            Y_{\pi(n+1)}^y
            \notin
            \cC_T\big(
                X_{\pi(n+1)}^y;
                \{Z_{\pi(j)}^y\}_{j=1}^n
            \big)
        \Big\}
        =
        \sum_{i=1}^{n+1}
        \pr_\pi\{\pi(n+1)=i\}
        \delta_T^y(Z_i^y)
        =
        \frac{1}{n+1}
        \sum_{i=1}^{n+1}\delta_T^y(Z_i^y).
    \end{equation*}
    Nestedness makes the any-stage event
    in~\eqref{eq:apply-assumed-amr} identical to the final-stage event in
    the last display. Combining the two identities gives
    Condition~\ref{p2:reje-count}.

    \emph{Condition~\ref{p3:per-inv}.}
    Initialize
    $\cF_0^y=\sigma\bigl(\Lbag\ocD(y)\Rbag,s_1\bigr)$ and
    $\cR_0^y=\varnothing$,
    and, for $t\in[T]$, recursively set
    \begin{equation*}
        \begin{aligned}
            \cR_t^y
            :=
            \big\{
                i\in[n+1]:
                \delta_t^y(Z_i^y)=1
            \big\},
            \qquad
            \cF_t^y
            :=
            \cF_{t-1}^y
            \vee
            \sigma\Big(
                \cR_t^y,
                \Lbag Z_i^y:i\in\cR_t^y\Rbag,
                s_{t+1}
            \Big).
        \end{aligned}
    \end{equation*}
    As in Condition~\ref{p3:per-inv}, the final argument is omitted when
    $t=T$.
    For every $t$, $\cF_{t-1}^y$ contains $\cF_0^y$ and therefore
    determines the full augmented multiset $\Lbag\ocD(y)\Rbag$.
    Given $z$, the full augmented multiset determines
    $\Lbag\ocD(y)_{\setminus z}\Rbag$.
    Symmetry of $\cC_t$ therefore implies that
    $\delta_t^y(z)$ depends only on $z$ and information contained in
    $\cF_{t-1}^y$.
    The same value-based rule is applied for every $y\in\cY$, with
    candidate dependence entering only through $\cF_{t-1}^y$.
    This verifies Condition~\ref{p3:per-inv}.

\end{proof}

\subsection{Proof of Proposition~\ref{prop:prediction-set-dominance}}

\begin{proof}[Proof of Proposition~\ref{prop:prediction-set-dominance}]
    By the superadditivity of the floor function,
    \begin{equation*}
        c_T
        =
        \left\lfloor
            (n+1)\sum_{r=1}^T\alpha_r
        \right\rfloor
        -
        \sum_{r<T}\lfloor(n+1)\alpha_r\rfloor
        \ge
        \lfloor(n+1)\alpha_T\rfloor.
    \end{equation*}
    Hence, $c_r\ge\lfloor(n+1)\alpha_r\rfloor$ for every
    $r\in[T]$, with equality for $r<T$.

    Fix $t\in[T]$ and
    $y\notin\hcC_t^\bon(X_{n+1};\balp)$.
    Then, for some $r\le t$,
    \begin{equation*}
        \sum_{j=1}^{n+1}
        \bbI\left\{
            s_r(Z_j^y)\ge s_r(Z_{n+1}^y)
        \right\}
        \le
        \lfloor(n+1)\alpha_r\rfloor
        \le c_r.
    \end{equation*}
    If $\delta_{r-1}^{y,\coin}(Z_{n+1}^y)=1$, then the augmented
    test point has already been excluded by \coinm, and hence
    $y\notin\hcC_t^\coin(X_{n+1};\{c_r\}_{r=1}^T)$ by the
    monotone-exclusion condition. Otherwise, the augmented test point
    remains eligible at stage $r$. Restricting the comparison to
    observations that have not previously been excluded can only
    decrease the number of observations whose $r$th score is at least
    $s_r(Z_{n+1}^y)$. Therefore, the exclusion criterion in
    \eqref{eq:coin-selection-set} is satisfied, so the augmented test
    point is excluded by \coinm at stage $r$. In either case,
    $y\notin\hcC_t^\coin(X_{n+1};\{c_r\}_{r=1}^T)$.

    Thus,
    $y\notin\hcC_t^\bon(X_{n+1};\balp)$ implies
    $y\notin\hcC_t^\coin(X_{n+1};\{c_r\}_{r=1}^T)$.
    Hence $\hcC_t^\coin\subseteq\hcC_t^\bon$ for every $t\in[T]$,
    and the utility inequality follows from the set-wise monotonicity
    of $u$ in~\eqref{eq:monotone-utility}.
\end{proof}

\subsection{Proof of Proposition~\ref{the:eprocess-val}}\label{assec:proof-eprocess-val}

\begin{proof}[Proof of Proposition~\ref{the:eprocess-val}]
    Work under Assumption~\ref{cond:exchangeability} and evaluate all
    quantities at $y=Y_{n+1}$.

    First, Condition~\ref{p3:per-inv} implies adaptedness.
    Indeed, for every $y$, $\cF_t^y$ contains both $\cR_t^y$ and the multiset of excluded observations.
    It therefore determines whether $n+1\in\cR_t^y$, as well as $|\cR_t^y|$, so $E_t^y$ is $\cF_t^y$-measurable.
    The measurability requirement in Condition~\ref{p3:per-inv} gives the same conclusion for $\{E_t^{Y_{n+1}}\}_{t=0}^T$.
    Moreover, Conditions~\ref{p1:nondecreasing} and~\ref{p2:reje-count} imply that
    $|\cR_t^y|\le|\cR_T^y|=\lfloor(n+1)\alpha\rfloor<n+1$.
    Hence, $E_t^y$ is well defined, nonnegative, and bounded above by $\alpha^{-1}$.

    Second, Condition~\ref{p3:per-inv} preserves conditional exchangeability within the pool of surviving observations.
    At $t=0$, this follows from exchangeability after conditioning on the unordered augmented dataset.
    At each subsequent stage, the common value-based rule uses only $\cF_{t-1}^y$ and does not use the pairing between the surviving observations and their index labels.
    Consequently, revealing $\cR_t^y$ and the multiset of excluded observations preserves the uniform assignment of the remaining observations to the remaining indices.
    Let $\cN_t^y=[n+1]\setminus\cR_t^y$.
    This conditional uniformity can be expressed as, for any measurable set $A \subseteq \cX \times \cY$,
    \begin{equation*}
        \pr\left\{
            Z_{n+1}^{Y_{n+1}} \in A
            \mid
            \cF_t^{Y_{n+1}},
            \ n+1\in\cN_t^{Y_{n+1}}
        \right\}
        =
        \frac{1}{|\cN_t^{Y_{n+1}}|}
        \sum_{
            j \in \cN_t^{Y_{n+1}}
        }
        \bbI\bigl\{Z_j^{Y_{n+1}}\in A\bigr\}.
    \end{equation*}
    Since the stage-$(t+1)$ exclusion decisions use only $\cF_t^{Y_{n+1}}$, the number of observations newly excluded at stage~$t+1$ is determined by $\cF_t^{Y_{n+1}}$.
    It follows that, conditional on $\cF_t^{Y_{n+1}}$ and $n+1\in\cN_t^{Y_{n+1}}$, the probability that the test index is newly excluded at stage~$t+1$ is
    $|\cR_{t+1}^{Y_{n+1}}\setminus\cR_t^{Y_{n+1}}|/
    |\cN_t^{Y_{n+1}}|$.

    We next verify the martingale property.
    Fix $t<T$.
    On $\{n+1 \in \cR_t^{Y_{n+1}}\}$, the construction in~\eqref{eq:e-process} gives
    $E_{t+1}^{Y_{n+1}} = \alpha^{-1} = E_t^{Y_{n+1}}$.
    On $\{n+1 \in \cN_t^{Y_{n+1}}\}$, the preceding conditional exclusion probability and the definition of the process yield
    \begin{equation*}
        \begin{aligned}
            \E\left(
                E_{t+1}^{Y_{n+1}}
                \mid
                \cF_t^{Y_{n+1}}
            \right)
            &=
            \frac{
                |
                    \cR_{t+1}^{Y_{n+1}}
                    \setminus
                    \cR_t^{Y_{n+1}}
                |
            }{
                |\cN_t^{Y_{n+1}}|
            }
            \times \alpha^{-1}
            +
            \frac{
                n+1- |\cR_{t+1}^{Y_{n+1}} |
            }{
                n+1- |\cR_t^{Y_{n+1}} |
            }
            \times
            \alpha^{-1}
            \frac{
                \alpha(n+1)- |\cR_{t+1}^{Y_{n+1}} |
            }{
                n+1- |\cR_{t+1}^{Y_{n+1}} |
            }
            \\
            &=
            \alpha^{-1}
            \frac{
                \alpha(n+1)- |\cR_t^{Y_{n+1}} |
            }{
                n+1- |\cR_t^{Y_{n+1}} |
            }
            =
            E_t^{Y_{n+1}},
        \end{aligned}
    \end{equation*}
    where the final equality uses monotone exclusion and hence
    $|\cR_{t+1}^{Y_{n+1}}\setminus\cR_t^{Y_{n+1}}|
    =|\cR_{t+1}^{Y_{n+1}}|-|\cR_t^{Y_{n+1}}|$.

    Thus, $\{E_t^{Y_{n+1}}\}_{t=0}^T$ is a nonnegative
    $\{\cF_t^{Y_{n+1}}\}_{t=0}^T$-martingale with initial value one.
    Because both the horizon and the martingale are bounded, the optional
    stopping theorem gives
    $\E[E_\tau^{Y_{n+1}}]=1$ for every stopping time $\tau\leq T$.
    Therefore, $\{E_t^{Y_{n+1}}\}_{t=0}^T$ is an $e$-process under the exchangeability null.
\end{proof}

\subsection{Proof of Corollary~\ref{cor:validity-voptcoins}}

\begin{proof}[Proof of Corollary~\ref{cor:validity-voptcoins}]
    Conditional on the strategy-learning sample and the independently
    trained score functions, the learned allocation
    $\hbvrho^{\rm s}$ is fixed.  The final-calibration observations and the
    test observation remain exchangeable, so
    Corollary~\ref{the:errorcontrolcoin}, applied with calibration size
    $n_2$, gives the stated AMR equality.
\end{proof}

\subsection{Proof of Theorem~\ref{the:shifted-oracle-inequality}}

For a calibration dataset $\cD$, let
$\hbtau(\bc;\cD)=(\htau_1(\bc;\cD),\ldots,\htau_T(\bc;\cD))$
denote the empirical \coinm thresholds obtained from the recursion in
Appendix~\ref{ssec:implementcoin}. Define the empirical utility of a
threshold vector $\btau$ by
\[
    \hcU(\btau;\cD)
    \coloneqq
    \frac{1}{|\cD|}
    \sum_{(x,y)\in\cD}
    u\left(\cC_{1:T}(x;\btau)\right).
\]

\begin{proof}[Proof of Theorem~\ref{the:shifted-oracle-inequality}]
    Apply Proposition~\ref{prop:conservativequantile} with $n=n_1$ and
    $\cD^\ca=\cD_1^\ca$. With probability at least $1-n_1^{-1}$,
    simultaneously for every $\bvrho\in\Theta_{\bvrho}$ and $t\in[T]$,
    \[
        q_t\bigl((\alpha\bvrho-a_{n_1}\iota_T)_+\bigr)
        \geq
        \htau_t\bigl(\bc(\bvrho;n_1);\cD_1^\ca\bigr)
        \geq
        q_t\bigl(\alpha\bvrho+a_{n_1}\iota_T\bigr).
    \]
    On this event, the prediction sets generated by the lower shifted budget contain
    the empirical prediction sets, which in turn contain those generated
    by the upper shifted budget. Therefore, by the set-wise monotonicity
    of $u$, simultaneously for every $\bvrho\in\Theta_{\bvrho}$,
    \begin{equation*}
        \hcU\Big(
            \bq\big((\alpha\bvrho-a_{n_1}\iota_T)_+\big);
            \cD_1^\ca
        \Big)
        \le
        \hcU\Big(
            \hbtau\big(\bc(\bvrho;n_1);\cD_1^\ca\big);
            \cD_1^\ca
        \Big)
        \le
        \hcU\Big(
            \bq\big(\alpha\bvrho+a_{n_1}\iota_T\big);
            \cD_1^\ca
        \Big).
    \end{equation*}
    Consequently,
    \begin{equation*}
        \begin{aligned}
            &\cU\Big(
                \bq\big((\alpha\bvrho^*-a_{n_1}\iota_T)_+\big)
            \Big)
            -
            \cU\Big(
                \bq\big(\alpha\hbvrho^{\rm s}
                +a_{n_1}\iota_T\big)
            \Big)
            \\ \le &
            \hcU\Big(
                \bq\big((\alpha\bvrho^*
                -a_{n_1}\iota_T)_+\big);
                \cD_1^\ca
            \Big)
            -
            \hcU\Big(
                \bq\big(\alpha\hbvrho^{\rm s}
                +a_{n_1}\iota_T\big);
                \cD_1^\ca
            \Big)
            +2\Delta_{n_1}
            \\ \le &
            \hcU\Big(
                \hbtau\big(
                    \bc(\bvrho^*;n_1);
                    \cD_1^\ca
                \big);
                \cD_1^\ca
            \Big)
            -
            \hcU\Big(
                \hbtau\big(
                    \bc(\hbvrho^{\rm s};n_1);
                    \cD_1^\ca
                \big);
                \cD_1^\ca
            \Big)
            +2\Delta_{n_1}
            \\ \le &
            2\Delta_{n_1}.
        \end{aligned}
    \end{equation*}
    The first inequality follows from the definition of
    $\Delta_{n_1}$ applied to the two displayed threshold vectors,
    the second follows from the preceding utility sandwich, and the
    last follows from the empirical optimality of
    $\hbvrho^{\rm s}$ in~\eqref{eq:learn-allocation-d1}.
\end{proof}

\subsection{Proof of Proposition~\ref{prop:adaptive-feature-acquisition-validity}}

\begin{proof}[Proof of Proposition~\ref{prop:adaptive-feature-acquisition-validity}]
    The proof follows the counting argument of
    Theorem~\ref{the:basicresult}, with
    Condition~\ref{p2:reje-count} replaced by
    $|\cR_2^y|\le c_{\rm tot}$.
    The branch allocation satisfies
    $\sum_{j=1}^J c_{2,j}=c_2$, while
    $|\cS_{2,j}^{y,\mathrm{br}}|\le c_{2,j}$.
    Hence, for every $y\in\cY$,
    \begin{equation*}
        |\cR_2^y|
        \le
        c_1+\sum_{j=1}^J|\cS_{2,j}^{y,\mathrm{br}}|
        \le
        c_1+\sum_{j=1}^Jc_{2,j}
        =
        c_{\rm tot}.
    \end{equation*}
    The cumulative update ensures that exclusions are irreversible.
    The same selection rule and score functions are applied to every
    augmented observation, and the branch budgets depend only on branch
    sizes. Permuting the observations therefore permutes their exclusion
    indicators in the same way. At the true candidate $y=Y_{n+1}$,
    exchangeability implies that each observation has the same probability
    of exclusion. Consequently,
    \begin{align*}
        \pr\left\{
            \exists\,t\in[2]:
            Y_{n+1}\notin
            \hcC_t^{\mathrm{br}}(X_{n+1})
        \right\}
        &=
        \frac{\E|\cR_2^{Y_{n+1}}|}{n+1}
        \le
        \frac{c_{\rm tot}}{n+1}
        \le\alpha.
    \end{align*}
\end{proof}

\subsection{Proof of Corollary~\ref{cor:validity-rlcoinm}}

\begin{proof}[Proof of Corollary~\ref{cor:validity-rlcoinm}]
    Condition on the score functions and allocation, chosen independently
    of the augmented sample, and on $\tX_{n+1}=\tilde x$.
    The calibration pairs remain i.i.d.\ from
    $P_X\times P_{Y\mid X}$ and independent of the test pair. Since the
    auxiliary draw depends on the test pair only through $X_{n+1}$, the
    conditional response law remains $P_{Y\mid X}$. By Bayes' rule,
    \[
        \frac{\rmd P_{X_{n+1}\mid\tX_{n+1}=\tilde x}}
             {\rmd P_X}(x)
        =
        \frac{H(x,\tilde x)}
             {\int_{\cX}H(v,\tilde x)\,\rmd P_X(v)}
    \]
    for almost every realized $\tilde x$.
    Thus Assumption~\ref{cond:covariate-shift} holds conditionally with
    likelihood-ratio weights proportional to $H(x,\tilde x)$.
    Theorem~\ref{the:validity-covariate-shift} gives the AMR bound
    conditional on these quantities. Averaging over the independently
    chosen scores and allocation yields the stated bound conditional on
    $\tX_{n+1}$; averaging over $\tX_{n+1}$ gives the marginal bound.
\end{proof}

\section{Technical Proofs for Appendix Results}
\label{asec:proofs-appendix}

The proofs are presented in the order in which the corresponding results
appear in the preceding appendix sections; the auxiliary lemmas used below
are collected in the final subsection.

\subsection{Proof of Theorem~\ref{the:basicresult-weighted}}

\begin{proof}[Proof of Theorem~\ref{the:basicresult-weighted}]
    The argument parallels the proof of
    Theorem~\ref{the:basicresult}, with uniform ranks replaced by weighted
    ranks.  Under covariate shift, conditional on the unordered augmented
    sample, the probability that a particular observation occupies the test
    position is proportional to its weight.  The filtration restriction in
    Condition~\ref{p3w:per-inv-weight} preserves this property among the
    surviving observations.

    By the definition of the induced prediction sets and
    Condition~\ref{p1:nondecreasing},
    \begin{equation*}
        \left\{
            \exists\,t\in[T]:
            Y_{n+1}\notin\hcC_t(X_{n+1})
        \right\}
        =
        \left\{
            \exists\,t\in[T]:
            \delta_t^{Y_{n+1}}(Z_{n+1}^{Y_{n+1}})=1
        \right\}
        =
        \{n+1\in\cR_T^{Y_{n+1}}\}.
    \end{equation*}
    Thus, it suffices to evaluate all quantities at $y=Y_{n+1}$.
    The measurability requirement in
    Condition~\ref{p3w:per-inv-weight} permits the filtration recursion
    to be evaluated at this random label.

    We first establish that, conditional on
    $\cF_t^{Y_{n+1}}$, the test observation is weighted-exchangeable
    over the observations that have not been excluded. Specifically, for every
    measurable $A\subseteq\cX\times\cY$,
    \begin{equation*}
        \pr\left\{
            Z_{n+1}^{Y_{n+1}}\in A
            \mid
            \cF_t^{Y_{n+1}},
            \ n+1\notin\cR_t^{Y_{n+1}}
        \right\}
        =
        \frac{
            \sum_{j\notin\cR_t^{Y_{n+1}}}
            w(X_j)\bbI\{Z_j^{Y_{n+1}}\in A\}
        }{
            \sum_{j\notin\cR_t^{Y_{n+1}}}w(X_j)
        }.
    \end{equation*}
    At $t=0$, the identity follows because the calibration and test laws
    share the conditional law $P_{Y\mid X}$ and their joint likelihood
    ratio is proportional to $w(X)$.
    Conditioning further on the independently trained score functions
    preserves this identity.
    If the property holds at stage~$t-1$, the common
    $\cF_{t-1}^{Y_{n+1}}$-measurable rule selects observations without
    using the pairing between their values and the remaining index labels.
    Revealing the new exclusion set and the multiset of excluded observations
    preserves weighted conditional exchangeability among the remaining
    observations.  Induction establishes the displayed identity for every
    $t\in\{0,\ldots,T\}$.

    Let $W=\sum_{i=1}^{n+1}w(X_i)$ and
    $b_t=\sum_{i\in\cR_t^{Y_{n+1}}}w(X_i)$.
    Condition~\ref{p2w:reje-count-weight} gives
    $b_T\le c_{\rm tot}^w=\alpha W$.
    On the event that the test observation survives stage $t$, the
    denominator $W-b_t$ is positive almost surely.  Moreover,
    $b_{t+1}-b_t$ is $\cF_t^{Y_{n+1}}$-measurable because the common
    exclusion rule determines the multiset of observations newly excluded
    at stage~$t+1$ without revealing the survivor index--value pairing.
    Hence, conditional weighted exchangeability gives
    \begin{equation*}
        \pr\left\{
            n+1\in\cR_{t+1}^{Y_{n+1}}
            \mid
            \cF_t^{Y_{n+1}},
            \ n+1\notin\cR_t^{Y_{n+1}}
        \right\}
        =
        \frac{b_{t+1}-b_t}{W-b_t}.
    \end{equation*}

    We now show by backward induction that, on
    $\{n+1\notin\cR_t^{Y_{n+1}}\}$,
    \begin{equation*}
        \pr\left\{
            n+1\in\cR_T^{Y_{n+1}}
            \mid
            \cF_t^{Y_{n+1}}
        \right\}
        \le
        \frac{c_{\rm tot}^w-b_t}{W-b_t}.
    \end{equation*}
    At $t=T$, the left-hand side is zero, while the right-hand side is
    nonnegative. Suppose the inequality holds at stage~$t+1$.
    Conditional on $\cF_t^{Y_{n+1}}$ and
    $n+1\notin\cR_t^{Y_{n+1}}$, the test index either is newly excluded
    at stage~$t+1$ or survives. The preceding exclusion
    probability and the induction hypothesis give
    \begin{equation*}
        \begin{aligned}
        &\pr\left\{
            n+1\in\cR_T^{Y_{n+1}}
            \mid
            \cF_t^{Y_{n+1}}
        \right\}
        \\
        &\quad\le
        \frac{b_{t+1}-b_t}{W-b_t}
        +
        \frac{W-b_{t+1}}{W-b_t}
        \frac{c_{\rm tot}^w-b_{t+1}}{W-b_{t+1}}
        =
        \frac{c_{\rm tot}^w-b_t}{W-b_t}.
        \end{aligned}
    \end{equation*}
    This proves the backward-induction claim.

    Since $\cR_0^{Y_{n+1}}=\varnothing$, we have
    \begin{equation*}
        \pr\left\{
            n+1\in\cR_T^{Y_{n+1}}
            \mid
            \cF_0^{Y_{n+1}}
        \right\}
        \le
        \frac{c_{\rm tot}^w}{W}
        =
        \alpha.
    \end{equation*}
    Taking expectations and using the first display proves the result.
\end{proof}

\subsection{Proof of Theorem~\ref{the:validity-covariate-shift}}

\begin{proof}[Proof of Theorem~\ref{the:validity-covariate-shift}]
    It suffices to verify the three weighted structural conditions and then
    apply Theorem~\ref{the:basicresult-weighted}.

    (\textbf{Monotone exclusion})
    By~\eqref{eq:coin-update-weight}, once an observation is excluded,
    its decision value remains equal to one at every later stage.
    Hence, Condition~\ref{p1:nondecreasing} holds.

    (\textbf{Weighted exclusion-count constraint})
    At stage~$t$, the newly excluded index set
    $\cS_t^{y,\wcoin}$ contains only indices of observations that have not been excluded
    through stage~$t-1$.
    By~\eqref{eq:coin-selection-set-weight}, the observations indexed by it form an initial segment
    of the currently surviving observations ranked in decreasing order
    of $s_t$, whose total weight is no greater than
    $c_t^w(\bvrho;c_{\rm tot}^w)$.
    Since each stagewise excluded index set contains only indices of observations
    that survived all previous stages, these sets are mutually disjoint, and
    \begin{equation*}
        \begin{aligned}
        \sum_{i=1}^{n+1}
        w(X_i)\delta_T^{y,\wcoin}(Z_i^y)
        &=
        \sum_{t=1}^T
        \sum_{i\in\cS_t^{y,\wcoin}}w(X_i)
        \\
        &\le
        \sum_{t=1}^T c_t^w(\bvrho;c_{\rm tot}^w)
        =
        c_{\rm tot}^w
        =
        \alpha\sum_{i=1}^{n+1}w(X_i).
        \end{aligned}
    \end{equation*}
    Thus, Condition~\ref{p2w:reje-count-weight} holds.

    (\textbf{Filtration restriction})
    For every $t$, $\cF_{t-1}^y$ contains the full augmented multiset,
    the set of indices excluded through stage~$t-1$, the corresponding
    multiset of excluded observations, the score functions $s_1,\ldots,s_t$, and the
    weight function.
    It therefore determines both the set of surviving indices and the
    multiset of surviving observations.
    Moreover, $c_{\rm tot}^w$ is determined by the augmented multiset
    and $w$, and the stagewise budgets are determined by
    $c_{\rm tot}^w$ because $\bvrho$ is prespecified.
    Given this information,
    \eqref{eq:coin-selection-set-weight} applies the same weighted
    ranking rule to every surviving occurrence.
    Consequently, induction from
    $\delta_0^{y,\wcoin}\equiv0$ shows that
    $\delta_t^{y,\wcoin}$ is $\cF_{t-1}^y$-measurable and is generated
    by the same value-based rule for every $y\in\cY$.
    Hence, Condition~\ref{p3w:per-inv-weight} holds.

    The conclusion now follows from
    Theorem~\ref{the:basicresult-weighted}.
\end{proof}

\subsection{Proof of Proposition~\ref{prop:robust-covariate-shift}}

\begin{proof}[Proof of Proposition~\ref{prop:robust-covariate-shift}]
    Let $\mathcal E$ denote the any-stage miscoverage event in the
    proposition.  Write $\widetilde P_g$ for the marginal law of
    $\widetilde X$ obtained by drawing $X\sim P_X\circ g$ and then
    $\widetilde X\mid X\sim H(X,\cdot)$.  Conditional on
    $\widetilde X=\tilde x$, define
    \[
        K_{\tilde x}
        \coloneqq P_X\circ H(\cdot,\tilde x),
        \qquad
        K_{\tilde x}^g
        \coloneqq P_X\circ g\circ H(\cdot,\tilde x).
    \]
    Under the distribution in the proposition, the conditional covariate
    law of the test observation is $K_{\tilde x}^g$.  Introduce a reference
    law that leaves the calibration sample and the marginal law
    $\widetilde P_g$ unchanged but replaces this conditional law by
    $K_{\tilde x}$.  Under the reference law, the weights
    $H(\cdot,\tilde x)$ are proportional to the likelihood ratio of the
    test-covariate law with respect to $P_X$.  Hence
    Theorem~\ref{the:validity-covariate-shift}, applied conditionally on
    $\widetilde X=\tilde x$, gives
    \[
        \pr_{\mathrm{ref}}(\mathcal E\mid\widetilde X=\tilde x)
        \leq \alpha.
    \]
    The calibration laws and the conditional laws of $Y$ given $X$ agree
    under the two distributions.  The total-variation comparison for
    conditional laws therefore yields
    \begin{equation}
        \label{eq:robustness-tv-reduction}
        \pr_g(\mathcal E)
        \leq
        \alpha+
        \E_{\widetilde X\sim\widetilde P_g}
        \left[
            d_{\rm TV}
            \bigl(K_{\widetilde X},K_{\widetilde X}^g\bigr)
        \right].
    \end{equation}

    It remains to bound the expectation in
    \eqref{eq:robustness-tv-reduction}.  For $\widetilde P_g$-almost every
    $\tilde x$, let
    $m(\tilde x)=\E_{K_{\tilde x}}[g(X)]>0$.  The Radon--Nikodym derivative
    of $K_{\tilde x}^g$ with respect to $K_{\tilde x}$ is
    $g(x)/m(\tilde x)$.  Consequently, if
    $X,X'\stackrel{\mathrm{i.i.d.}}{\sim}K_{\tilde x}$, Jensen's
    inequality gives
    \[
        d_{\rm TV}(K_{\tilde x},K_{\tilde x}^g)
        =
        \frac{\E|g(X)-m(\tilde x)|}{2m(\tilde x)}
        \leq
        \frac{\E|g(X)-g(X')|}{2m(\tilde x)}.
    \]
    Let $P_{\widetilde X}$ be the marginal law obtained from
    $X\sim P_X$ and $\widetilde X\mid X\sim H(X,\cdot)$.  The density
    defining $\widetilde P_g$ cancels the denominator $m(\tilde x)$ in
    the preceding display, and hence
    \begin{equation}
        \label{eq:robustness-tv-pair}
        \E_{\widetilde P_g}
        d_{\rm TV}(K_{\widetilde X},K_{\widetilde X}^g)
        \leq
        \frac{1}{2\E_{P_X}[g(X)]}
        \E_{\widetilde X\sim P_{\widetilde X}}
        \E_{X,X'\stackrel{\mathrm{i.i.d.}}{\sim}K_{\widetilde X}}
        |g(X)-g(X')|.
    \end{equation}

    Fix a measurable $A\subseteq\cX$ and $\varepsilon>0$.  On the event
    that $X,X'\in A$ and both lie within distance $\varepsilon$ of
    $\widetilde X$, the triangle inequality gives
    $\|X-X'\|\leq2\varepsilon$ and therefore
    \[
        \frac12|g(X)-g(X')|
        \leq \varepsilon L_{g,2\varepsilon}(A).
    \]
    On the complementary event, nonnegativity of $g$ gives
    $|g(X)-g(X')|\leq\|g\|_\infty$.  A union bound, together with the
    fact that each pair $(X,\widetilde X)$ has the joint law obtained by
    drawing $X\sim P_X$ and then
    $\widetilde X\mid X\sim H(X,\cdot)$, bounds the right-hand side of
    \eqref{eq:robustness-tv-pair} by
    \[
        \frac{
            \varepsilon L_{g,2\varepsilon}(A)
            +\|g\|_\infty
             \pr_{P_{X,\widetilde X}}
             \{\|X-\widetilde X\|>\varepsilon\}
            +\|g\|_\infty P_X(A^c)
        }{\E_{P_X}[g(X)]}.
    \]
    Taking the infimum over $A$ and $\varepsilon$ and substituting this
    bound into \eqref{eq:robustness-tv-reduction} proves the claim.  This
    is the total-variation argument underlying Theorem~3 of
    \citet{hore2025conformal}, applied here to the any-stage
    miscoverage event.
\end{proof}

\subsection{Proof of Proposition~\ref{prop:equiv-coin}}
\label{assec:proof-equivalent-coin}

\begin{proof}[Proof of Proposition~\ref{prop:equiv-coin}]
    Consider a calibration realization for which the scores
    $\{s_t(Z_i):i\in[n]\}$ are distinct for every $t\in[T]$.
    This event has probability one under Assumption~\ref{cond:exchangeability}.
    For $x\in\cX$, define the calibration-score boundary
    \[
        \cT(x;\cD^\ca)
        \coloneqq
        \bigcup_{t=1}^T\bigcup_{i=1}^n
        \bigl\{y\in\cY:s_t(x,y)=s_t(Z_i)\bigr\}.
    \]
    The allocation in Section~\ref{ssec:implementcoin} satisfies
    $\sum_{t=1}^T c_t\leq n$. By induction in the calibration recursion,
    $|\cA_{t-1}|=n-\sum_{r<t}c_r\geq c_t$, so every positive-order
    statistic used in~\eqref{eq:calibration-only-coin-thresholds} exists.

    Fix $x\in\cX$ and $y\notin\cT(x;\cD^\ca)$, and apply the candidate-wise
    construction to the augmented sample $(Z_1,\ldots,Z_n,(x,y))$.
    We prove by induction on $t$ that (i) the candidate-wise and threshold
    constructions make the same inclusion decision for $y$ through stage~$t$,
    and (ii), whenever the candidate survives, the calibration indices
    remaining in the augmented exclusion process are exactly $\cA_t$ from
    \eqref{eq:calibration-only-coin-thresholds}. Both claims hold at stage~$0$,
    when all calibration observations and the candidate are retained.

    Suppose they hold through stage~$t-1$. If the candidate has already
    been excluded, nestedness keeps it out of both prediction sequences at
    stage~$t$. Otherwise, the surviving calibration indices in the
    candidate-wise process are $\cA_{t-1}$. When $c_t=0$, neither
    construction excludes the candidate or any calibration observation:
    $\htau_t(\bc)=+\infty$, and all scores are real-valued, so
    $\cA_t=\cA_{t-1}$. This also covers the case $\cA_{t-1}=\varnothing$.

    When $c_t\geq1$, the calibration scores are distinct, and the
    candidate score differs from every calibration score because
    $y\notin\cT(x;\cD^\ca)$. The ranking rule
    in~\eqref{eq:coin-selection-set} therefore retains the candidate if
    and only if
    \[
        s_t(x,y)
        <
        \bbC_{c_t}\bigl(\{s_t(Z_i):i\in\cA_{t-1}\}\bigr)
        =
        \htau_t(\bc).
    \]
    Since equality cannot occur, this is equivalent to membership in
    $\tcC_t(x)$. Moreover, whenever the candidate survives, the $c_t$
    observations excluded at stage~$t$ are precisely the calibration
    observations in $\cA_{t-1}$ whose scores are at least
    $\htau_t(\bc)$. The remaining calibration indices are therefore
    $\cA_t$, completing the induction.

    Thus, for every $x\in\cX$ and $t\in[T]$,
    \[
        \hcC_t^\coin(x)\mathbin{\triangle}\hcC_t^{\rm thr}(x)
        \subseteq \cT(x;\cD^\ca).
    \]
    In particular, any disagreement at the true outcome implies
    $Y_{n+1}\in\cT(X_{n+1};\cD^\ca)$, which requires
    $s_t(Z_{n+1})=s_t(Z_i)$ for some $t\in[T]$ and $i\in[n]$.
    This event has probability zero by the no-ties clause of
    Assumption~\ref{cond:exchangeability}, proving the stated result.
\end{proof}

\subsection{Proof of Theorem~\ref{the:oracle-recovery}}

\begin{proof}[Proof of Theorem~\ref{the:oracle-recovery}]
    Continuity of $-f_{Y\mid X}(Y\mid X)$ gives
    $\pr\{Y\in\cC^\ora(X)\}=1-\alpha$, and hence
    $\pr_{XY}(\cO^*)=\alpha$.  For $t\in[T]$, define
    \[
        \varrho_t^*
        \coloneqq
        \frac{1}{\alpha}
        \left[
            \pr_{XY}\left\{
                \bigcup_{j=1}^t\cO_j(\tau_j^*)
            \right\}
            -
            \pr_{XY}\left\{
                \bigcup_{j=1}^{t-1}\cO_j(\tau_j^*)
            \right\}
        \right],
    \]
    where the second union is empty when $t=1$.  These increments are
    nonnegative.  The theorem's premise and
    $\pr_{XY}(\cO^*)=\alpha$ imply
    $\sum_{t=1}^T\varrho_t^*=1$, so
    $\bvrho^*\coloneqq(\varrho_1^*,\ldots,\varrho_T^*)$ belongs to
    $\Theta_{\bvrho}$.  Moreover,
    \begin{equation}
        \label{eq:oracle-cumulative-mass}
        1-\alpha\beta_t(\bvrho^*)
        =
        \pr_{XY}\left\{
            \left(
                \bigcup_{j=1}^t\cO_j(\tau_j^*)
            \right)^c
        \right\}.
    \end{equation}

    Set $\cA_0\coloneqq\cX\times\cY$ and, for $t\in[T]$, let
    \[
        \cA_t
        \coloneqq
        \bigcap_{j=1}^t
        \cO_j\bigl(q_j(\alpha\bvrho^*)\bigr)^c
    \]
    be the population acceptance region induced by \coinm.  We prove by
    induction that
    \begin{equation}
        \label{eq:oracle-induction}
        \pr_{XY}\left[
            \cA_t\mathbin{\triangle}
            \left(
                \bigcup_{j=1}^t\cO_j(\tau_j^*)
            \right)^c
        \right]=0,
        \qquad t\in[T]\cup\{0\}.
    \end{equation}
    The assertion is immediate for $t=0$.  Suppose it holds at
    stage $t-1$.

    If $\varrho_t^*=0$, the two nested target exclusion regions through
    stages $t-1$ and $t$ have the same probability and therefore agree
    up to a null set.  By convention,
    $q_t(\alpha\bvrho^*)=+\infty$, so
    $\cA_t=\cA_{t-1}$ and the induction claim follows.

    Suppose instead that $\varrho_t^*>0$.  By the induction hypothesis,
    the stage-$t$ nested quantile is
    \[
        q_t(\alpha\bvrho^*)
        =
        \inf\left\{
            v\in\overline{\bbR}:
            \pr_{XY}\bigl(
                \cA_{t-1}\cap\cO_t(v)^c
            \bigr)
            \geq
            1-\alpha\beta_t(\bvrho^*)
        \right\}.
    \]
    The no-ties assumption makes $s_t(X,Y)$ non-atomic.  Hence
    $v\mapsto\pr_{XY}\{\cO_t(v)\setminus\cO^*\}$ is continuous, so the
    infimum defining $\tau_t^*$ is attained and
    $\cO_t(\tau_t^*)\subseteq\cO^*$ up to a null set.  Together with
    \eqref{eq:oracle-cumulative-mass} and the induction hypothesis, this
    gives
    \[
        \pr_{XY}\bigl(
            \cA_{t-1}\cap\cO_t(\tau_t^*)^c
        \bigr)
        =
        1-\alpha\beta_t(\bvrho^*).
    \]
    Hence $q_t(\alpha\bvrho^*)\leq\tau_t^*$.  Non-atomicity also makes the map
    $v\mapsto\pr_{XY}\{\cA_{t-1}\cap\cO_t(v)^c\}$ continuous.
    It follows that
    \[
        \pr_{XY}(\cA_t)=1-\alpha\beta_t(\bvrho^*).
    \]
    Because $q_t(\alpha\bvrho^*)\leq\tau_t^*$, monotonicity of
    $\cO_t(v)$ in $v$ also gives
    \[
        \cA_t
        \subseteq
        \left(
            \bigcup_{j=1}^t\cO_j(\tau_j^*)
        \right)^c
        \quad\text{up to a null set}.
    \]
    This inclusion and equality of probabilities prove
    \eqref{eq:oracle-induction} at stage $t$.

    Taking $t=T$ in \eqref{eq:oracle-induction} and using the theorem's
    premise yields
    $\pr_{XY}\{\cA_T\triangle(\cO^*)^c\}=0$.  Finally,
    $(\cO^*)^c=\{(x,y):y\in\cC^\ora(x)\}$ and
    $\cA_T=\{(x,y):y\in\cC_T^\coin(x;\alpha\bvrho^*)\}$.  Since the
    population \coinm sequence is nested,
    $\cC_T^\coin=\bigcap_{t=1}^T\cC_t^\coin$, completing the proof.
\end{proof}

\subsection{Proof of Theorem~\ref{the:errorofnestedquantile}}

\begin{proof}[Proof of Theorem~\ref{the:errorofnestedquantile}]
    The multivariate DKW inequality
    \citep{devroye1977uniform} implies that
    \begin{equation*}
        \pr \bigg\{
            \sup_{(v_1,\ldots,v_t)\in\overline{\bbR}^t}
            \left|
                F_{\bS(\ole t)}(v_1,\ldots,v_t)
                -
                F_{n,\bS(\ole t)}(v_1,\ldots,v_t)
            \right|
            \ge u
        \bigg\}
        \le
        2(2n)^t
        e^{-2n(u-t/n)^2},
    \end{equation*}
    for any $u>t/n$, $t>1$, and $n>1$.
    For $t=1$, the same, possibly looser, bound follows from the
    ordinary univariate DKW inequality. Therefore,
    \begin{equation*}
        \begin{aligned}
            &
            \pr \bigg\{
                \max_{t\in[T]}
                \sup_{(v_1,\ldots,v_t)\in\bbR^t}
                \left|
                    F_{\bS(\ole t)}(v_1,\ldots,v_t)
                    -
                    F_{n,\bS(\ole t)}(v_1,\ldots,v_t)
                \right|
                \ge u
            \bigg\}
            \\
            &\le
            \sum_{t=1}^T
            2(2n)^t e^{-2n(u-t/n)^2}
            \le
            4(2n)^T e^{-2n(u-T/n)^2}.
        \end{aligned}
    \end{equation*}
    Set $\epsilon_n=\sqrt{T\log n/n}+T/n$ and define
    \[
        \cE_n
        \coloneqq
        \bigg\{
            \max_{t\in[T]}
            \sup_{(v_1,\ldots,v_t)\in\overline{\bbR}^t}
            \left|
                F_{\bS(\ole t)}(v_1,\ldots,v_t)
                -
                F_{n,\bS(\ole t)}(v_1,\ldots,v_t)
            \right|
            \le\epsilon_n
        \bigg\}.
    \]
    At this choice of $\epsilon_n$, the preceding upper bound equals
    $4\,2^Tn^{-T}$, which is at most $n^{-1}$ when $n>15$ and
    $T\geq2$.  Therefore $\pr(\cE_n)\geq1-n^{-1}$.

    Marginal non-atomicity implies that there are no ties among the
    observations in each coordinate almost surely.  Work on the
    intersection of this probability-one event and $\cE_n$.  Fix arbitrary
    $\alpha\in(0,1)$ and $\bvrho\in\Theta_{\bvrho}$, and suppress the
    argument $\alpha\bvrho$ below.

    On $\cE_n$,
    \begin{equation}\label{eq:decomposeevents1}
        \begin{aligned}
            &
            \left|
                F_{\bS(\ole t)}
                (q_1,\ldots,q_{t-1},q_t)
                -
                F_{\bS(\ole t)}
                (q_1,\ldots,q_{t-1},\hq_t)
            \right|
            \\
            &\le
            \left|
                F_{\bS(\ole t)}
                (q_1,\ldots,q_{t-1},q_t)
                -
                F_{n,\bS(\ole t)}
                (\hq_1,\ldots,\hq_{t-1},\hq_t)
            \right|
            \\
            &\quad+
            \left|
                F_{n,\bS(\ole t)}
                (\hq_1,\ldots,\hq_{t-1},\hq_t)
                -
                F_{\bS(\ole t)}
                (\hq_1,\ldots,\hq_{t-1},\hq_t)
            \right|
            \\
            &\quad+
            \left|
                F_{\bS(\ole t)}
                (\hq_1,\ldots,\hq_{t-1},\hq_t)
                -
                F_{\bS(\ole t)}
                (q_1,\ldots,q_{t-1},\hq_t)
            \right|
            \\
            &\le
            \frac1n+\epsilon_n
            +
            \left|
                F_{\bS(\ole t)}
                (\hq_1,\ldots,\hq_{t-1},\hq_t)
                -
                F_{\bS(\ole t)}
                (q_1,\ldots,q_{t-1},\hq_t)
            \right|.
        \end{aligned}
    \end{equation}
    The first inequality follows from the triangle inequality.
    The second inequality follows from $\cE_n$ and
    \begin{equation*}
        F_{\bS(\ole t)}(q_1,\ldots,q_t)
        =
        1-\alpha\beta_t(\bvrho),
        \qquad
        \left|
            F_{n,\bS(\ole t)}(\hq_1,\ldots,\hq_t)
            -
            \{1-\alpha\beta_t(\bvrho)\}
        \right|
        \le\frac1n.
    \end{equation*}
    These relations follow from the population and empirical nested
    quantile definitions, respectively, together with marginal
    non-atomicity and the absence of empirical ties.

    Let
    $A_t=\{S_t\le q_t\}$ and
    $\hA_t=\{S_t\le\hq_t\}$ be two subsets of the sample space
    $\Omega$. Denote
    $Q_t=\cap_{t'\le t}A_{t'}$ and
    $\hQ_t=\cap_{t'\le t}\hA_{t'}$, with
    $Q_0=\hQ_0=\Omega$.
    The left-hand side of \eqref{eq:decomposeevents1} is therefore
    equivalent to
    \begin{equation*}
        \left|
            \pr_{\bS}(A_t\cap Q_{t-1})
            -
            \pr_{\bS}(\hA_t\cap Q_{t-1})
        \right|.
    \end{equation*}
    It follows that
    \begin{equation}\label{eq:decomposeevents2}
        \begin{aligned}
            &
            \left|
                \pr_{\bS}(A_t\cap Q_{t-1})
                -
                \pr_{\bS}(\hA_t\cap Q_{t-1})
            \right|
            \\
            &
            \overset{(i)}{\le}
            \frac1n+\epsilon_n
            +
            \left|
                \pr_{\bS}(\hA_t\cap\hQ_{t-1})
                -
                \pr_{\bS}(\hA_t\cap Q_{t-1})
            \right|
            \\
            &\overset{(ii)}{\le}
            \frac1n+\epsilon_n
            +
            \left|
                \pr_{\bS}(\hQ_{t-1})
                -
                \pr_{\bS}(Q_{t-1})
            \right|
            +
            \pr_{\bS}
            \left(
                \hA_t^c
                \cap
                (Q_{t-1}\triangle\hQ_{t-1})
            \right)
            \\
            &\overset{(iii)}{\le}
            2\left(\frac1n+\epsilon_n\right)
            +
            \pr_{\bS}
            \left(
                \hA_t^c
                \cap
                (Q_{t-1}\triangle\hQ_{t-1})
            \right).
        \end{aligned}
    \end{equation}
    Here, step~$(i)$ is a reexpression of the result in~\eqref{eq:decomposeevents1} and step~$(ii)$ follows from Lemma~\ref{lem:boundintersectionmeasure}.
    For step~$(iii)$, the nested quantile identities at stage $t-1$ and the event $\cE_n$ give
    \begin{equation*}
        \left|
            \pr_{\bS}(\hQ_{t-1})
            -
            \pr_{\bS}(Q_{t-1})
        \right|
        \le
        \epsilon_n+\frac1n.
    \end{equation*}

    We now focus on the second term in the last line of~\eqref{eq:decomposeevents2}.
    Specifically,
    \begin{equation}\label{eq:inductionbase}
        \begin{aligned}
            \pr_{\bS}(Q_t\triangle\hQ_t)
            \overset{(i)}{\le}
            &
            \pr_{\bS}
            \left(
                (A_t\triangle\hA_t)\cap Q_{t-1}
            \right)
            +
            \pr_{\bS}
            \left(
                \hA_t
                \cap
                (Q_{t-1}\triangle\hQ_{t-1})
            \right)
            \\
            =
            &
            \left|
                \pr_{\bS}(A_t\cap Q_{t-1})
                -
                \pr_{\bS}(\hA_t\cap Q_{t-1})
            \right|
            +
            \pr_{\bS}
            \left(
                \hA_t
                \cap
                (Q_{t-1}\triangle\hQ_{t-1})
            \right)
            \\
            \overset{(ii)}{\le}
            &
            2\left(\frac1n+\epsilon_n\right)
            +
            \pr_{\bS}
            \left(
                \hA_t^c
                \cap
                (Q_{t-1}\triangle\hQ_{t-1})
            \right)
            +
            \pr_{\bS}
            \left(
                \hA_t
                \cap
                (Q_{t-1}\triangle\hQ_{t-1})
            \right)
            \\
            =
            &
            2\left(\frac1n+\epsilon_n\right)
            +
            \pr_{\bS}
            (Q_{t-1}\triangle\hQ_{t-1}).
        \end{aligned}
    \end{equation}
    Here, step~$(i)$ follows from
    Lemma~\ref{lem:boundsymmetricdifference}, and step~$(ii)$ follows
    from \eqref{eq:decomposeevents2}. The equality in the second line
    uses the fact that $A_t$ and $\hA_t$ are threshold events for the
    same scalar random variable and are therefore nested.

    Since \eqref{eq:inductionbase} holds for every $t\in[T]$ on
    $\cE_n$ and $Q_0=\hQ_0$, induction gives
    \begin{equation}\label{eq:inductionresult}
        \pr_{\bS}(Q_t\triangle\hQ_t)
        \le
        2t\left(\frac1n+\epsilon_n\right).
    \end{equation}
    Combining it with~\eqref{eq:decomposeevents2} yields
    \begin{equation*}
        \begin{aligned}
            &
            \left|
                \pr_{\bS}(A_t\cap Q_{t-1})
                -
                \pr_{\bS}(\hA_t\cap Q_{t-1})
            \right|
            \\ \le &
            2\left(\frac1n+\epsilon_n\right)
            +
            \pr_{\bS}
            \left(
                \hA_t^c
                \cap
                (Q_{t-1}\triangle\hQ_{t-1})
            \right)
            \\ \le &
            2\left(\frac1n+\epsilon_n\right)
            +
            \pr_{\bS}(Q_{t-1}\triangle\hQ_{t-1})
            \le
            2t\left(\frac1n+\epsilon_n\right).
        \end{aligned}
    \end{equation*}
    This inequality holds on $\cE_n$.

    Finally, under $T\ge2$ and $T<n\log n/9$,
    \begin{equation*}
        \frac1n+\epsilon_n
        =
        \sqrt{\frac{T\log n}{n}}
        +
        \frac{T+1}{n}
        \le
        \frac32
        \sqrt{\frac{T\log n}{n}}.
    \end{equation*}
    The two assertions now follow from
    \eqref{eq:inductionresult}, the preceding inequality, and the
    definitions of $A_t,\hA_t,Q_t$, and $\hQ_t$.  Finally, $\cE_n$ does
    not depend on $\alpha$ or $\bvrho$, and every step above is a
    deterministic implication on this common event.  The conclusions
    therefore hold simultaneously over $\alpha\in(0,1)$ and
    $\bvrho\in\Theta_{\bvrho}$.
\end{proof}

\subsection{Proof of Proposition~\ref{prop:conservativequantile}}

\begin{proof}[Proof of Proposition~\ref{prop:conservativequantile}]
    First, both perturbed budget vectors are admissible:
    \begin{equation*}
        \sum_{t=1}^T\alpha_{n,t}^-(\bvrho)
        \le\alpha<1,
        \qquad
        \sum_{t=1}^T\alpha_{n,t}^+(\bvrho)
        =
        \alpha+\frac{a_nT(T+1)}2
        <1.
    \end{equation*}
    We first control the population probabilities of the lower orthants
    determined by the empirical \coinm thresholds and then compare those
    thresholds with the two perturbed population nested quantiles.

    Recall the empirical-CDF event used in the proof of
    Theorem~\ref{the:errorofnestedquantile}. With probability at least
    $1-n^{-1}$,
    \begin{equation}
        \label{eq:uniform-cdf-bound-for-coin-threshold}
        \max_{t\in[T]}
        \sup_{(v_1,\ldots,v_t)\in\overline{\bbR}^{\,t}}
        \left|
            F_{n,{\scriptscriptstyle\bS(\ole t)}}
            (v_1,\ldots,v_t)
            -
            F_{{\scriptscriptstyle\bS(\ole t)}}
            (v_1,\ldots,v_t)
        \right|
        \le
        \sqrt{\frac{T\log n}{n}}
        +
        \frac{T}{n}.
    \end{equation}
    Marginal non-atomicity also implies that there are no ties among
    the calibration scores in any coordinate almost surely. We work
    on the intersection of these two events.

    For any $\bvrho\in\Theta_{\bvrho}$ and $t\in[T]$, the definition
    of the exclusion-count allocation gives
    \begin{equation*}
        \sum_{t'=1}^t c_{t'}(\bvrho;n)
        =
        \left\lfloor
            (n+1)\alpha\beta_t(\bvrho)
        \right\rfloor.
    \end{equation*}
    Thus, after stage $t$, exactly
    $\lfloor(n+1)\alpha\beta_t(\bvrho)\rfloor$ calibration observations have been excluded.
    Then, every surviving observation satisfies the first $t$ threshold inequalities in
    \[
        F_{n,{\scriptscriptstyle\bS(\ole t)}}
        \big(\htau_1(\alpha\bvrho),\ldots,\htau_t(\alpha\bvrho)\big).
    \]
    Moreover, by the definition, at most one cutoff observation from each of the first $t$ stages can additionally satisfy all of them.
    Consequently,
    \begin{equation*}
        n-
        \left\lfloor
            (n+1)\alpha\beta_t(\bvrho)
        \right\rfloor
        \le
        nF_{n,{\scriptscriptstyle\bS(\ole t)}}
        \left(
            \htau_1(\alpha\bvrho),
            \ldots,
            \htau_t(\alpha\bvrho)
        \right)
        \le
        n-
        \left\lfloor
            (n+1)\alpha\beta_t(\bvrho)
        \right\rfloor+t.
    \end{equation*}
    Meanwhile, since $\alpha\beta_t(\bvrho)<1$, we have
    \begin{equation*}
        \left|
            \alpha\beta_t(\bvrho)
            -
            \frac{
                \lfloor(n+1)\alpha\beta_t(\bvrho)\rfloor
            }{n}
        \right|
        \le
        \frac{1}{n},
    \end{equation*}
    which gives
    \begin{equation}
        \label{eq:coin-empirical-rank-error}
        \left|
            F_{n,{\scriptscriptstyle\bS(\ole t)}}
            \left(
                \htau_1(\alpha\bvrho),
                \ldots,
                \htau_t(\alpha\bvrho)
            \right)
            -
            \left\{
                1-\alpha\beta_t(\bvrho)
            \right\}
        \right|
        \le
        \frac{t+1}{n}.
    \end{equation}
    Let $r_n\coloneqq\sqrt{T\log n/n}$, so $a_n=5r_n$.  Combining
    \eqref{eq:uniform-cdf-bound-for-coin-threshold} and
    \eqref{eq:coin-empirical-rank-error} bounds the relevant absolute error
    by $r_n+(T+t+1)/n$.  The condition
    $T<n\log n/9$ implies $T/n<r_n/3$, while $n>15$ and $T>1$ imply
    $1/n<r_n$.  Hence this error is smaller than
    $8r_n/3<11a_n/12$, which gives
    \begin{equation}
        \label{eq:coin-threshold-population-mass}
        \left|
            F_{{\scriptscriptstyle\bS(\ole t)}}
            \left(
                \htau_1(\alpha\bvrho),
                \ldots,
                \htau_t(\alpha\bvrho)
            \right)
            -
            \left\{
                1-\alpha\beta_t(\bvrho)
            \right\}
        \right|
        <
        \frac{11a_n}{12},
    \end{equation}
    under the condition $1 < T<n\log n/9$.
    Because the event in~\eqref{eq:uniform-cdf-bound-for-coin-threshold} controls all lower
    orthants, this inequality holds simultaneously for every
    $\bvrho\in\Theta_{\bvrho}$ and $t\in[T]$.

    We next compare the empirical thresholds with the two perturbed
    population nested quantiles.
    Fix $\bvrho\in\Theta_{\bvrho}$ and suppress its occurrence in
    $\balp_n^\pm(\bvrho)$ and
    $\alpha_{n,t}^\pm(\bvrho)$. We first consider $t=1$. If
    $\alpha_{n,1}^-=0$, then $q_1(\balp_n^-)=+\infty$, and the left
    inequality is immediate. Otherwise,
    $\alpha_{n,1}^-=\alpha\varrho_1-a_n$, and
    \eqref{eq:coin-threshold-population-mass} gives
    \begin{equation*}
        F_{{\scriptscriptstyle\bS(\ole1)}}
        \big(
            \htau_1(\alpha\bvrho)
        \big)
        <
        1-\alpha\varrho_1+\frac{11a_n}{12}
        <
        1-\alpha\varrho_1+a_n
        =
        1-\alpha_{n,1}^-.
    \end{equation*}
    Hence, by the definition of the generalized inverse,
    $
        q_1(\balp_n^-)
        \ge
        \htau_1(\alpha\bvrho).
    $
    Similarly,
    \begin{equation*}
        F_{{\scriptscriptstyle\bS(\ole1)}}
        \big(
            \htau_1(\alpha\bvrho)
        \big)
        >
        1-\alpha\varrho_1-\frac{11a_n}{12}
        >
        1-\alpha\varrho_1-a_n
        =
        1-\alpha_{n,1}^+.
    \end{equation*}
    Therefore,
    $
        \htau_1(\alpha\bvrho)
        \ge
        q_1(\balp_n^+).
    $

    We now proceed by induction. Let $t\ge2$ and suppose that
    \begin{equation*}
        q_{t'}(\balp_n^-)
        \ge
        \htau_{t'}(\alpha\bvrho)
        \ge
        q_{t'}(\balp_n^+),
        \qquad t'<t.
    \end{equation*}
    If $\alpha_{n,t}^-=0$, then
    $q_t(\balp_n^-)=+\infty$, so the left inequality is immediate.
    Suppose that $\alpha_{n,t}^->0$, in which case
    $\alpha_{n,t}^-=\alpha\varrho_t-ta_n$. The induction hypothesis, together with
    \eqref{eq:coin-threshold-population-mass}, implies
    \begin{equation*}
        \begin{aligned}
            &
            F_{{\scriptscriptstyle\bS(\ole t)}}
            \left(
                q_1(\balp_n^-),
                \ldots,
                q_{t-1}(\balp_n^-),
                \htau_t(\alpha\bvrho)
            \right)
            \\
            \overset{(i)}{\le} &
            F_{{\scriptscriptstyle\bS(\ole t)}}
            \left(
                \htau_1(\alpha\bvrho),
                \ldots,
                \htau_t(\alpha\bvrho)
            \right)
            \\
            &\quad+
            F_{{\scriptscriptstyle\bS(\ole t-1)}}
            \left(
                q_1(\balp_n^-),
                \ldots,
                q_{t-1}(\balp_n^-)
            \right)
            -
            F_{{\scriptscriptstyle\bS(\ole t-1)}}
            \left(
                \htau_1(\alpha\bvrho),
                \ldots,
                \htau_{t-1}(\alpha\bvrho)
            \right)
            \\ \le &
            \Big(1 - \alpha \beta_t(\bvrho) + \frac{11 a_n}{12}\Big)
            +
            \Big(1 - \sum_{t' < t} \alpha^{-}_{n,t'}\Big)
            -
            \Big(1 - \alpha \beta_{t-1}(\bvrho) - \frac{11 a_n}{12}\Big)
            \\ = &
            1-\sum_{t'\le t}\alpha_{n,t'}^-
            -ta_n+\frac{11a_n}{6}
            <
            1-\sum_{t'\le t}\alpha_{n,t'}^-,
        \end{aligned}
    \end{equation*}
    where step~$(i)$ uses the inclusion between the two preceding
    lower-orthant events, and the last inequality follows from $t\ge2$.
    Thus the defining
    probability level of $q_t(\balp_n^-)$ has not been reached at
    $\htau_t(\alpha\bvrho)$, and therefore
    \begin{equation*}
        q_t(\balp_n^-)
        \ge
        \htau_t(\alpha\bvrho).
    \end{equation*}

    For the lower bracket, the other half of the induction hypothesis gives
    \begin{equation*}
        \begin{aligned}
            &
            F_{{\scriptscriptstyle\bS(\ole t)}}
            \left(
                q_1(\balp_n^+),
                \ldots,
                q_{t-1}(\balp_n^+),
                \htau_t(\alpha\bvrho)
            \right)
            \\ \ge &
            F_{{\scriptscriptstyle\bS(\ole t)}}
            \left(
                \htau_1(\alpha\bvrho),
                \ldots,
                \htau_t(\alpha\bvrho)
            \right)
            \\
            &\quad-
            F_{{\scriptscriptstyle\bS(\ole t-1)}}
            \left(
                \htau_1(\alpha\bvrho),
                \ldots,
                \htau_{t-1}(\alpha\bvrho)
            \right)
            +
            F_{{\scriptscriptstyle\bS(\ole t-1)}}
            \left(
                q_1(\balp_n^+),
                \ldots,
                q_{t-1}(\balp_n^+)
            \right)
            \\ > &
            1-\sum_{t'\le t}\alpha_{n,t'}^+
            +ta_n-\frac{11a_n}{6}
            >
            1-\sum_{t'\le t}\alpha_{n,t'}^+.
        \end{aligned}
    \end{equation*}
    Hence $\htau_t(\alpha\bvrho)$ belongs to the upper-level set
    defining $q_t(\balp_n^+)$, which gives
    \begin{equation*}
        \htau_t(\alpha\bvrho)
        \ge
        q_t(\balp_n^+).
    \end{equation*}

    This completes the induction. Since $\bvrho$ was arbitrary and
    the empirical-CDF event is common to all lower orthants, the
    conclusion holds simultaneously for every
    $\bvrho\in\Theta_{\bvrho}$ and $t\in[T]$.
\end{proof}

\subsection{Proof of Corollary~\ref{cor:nested-quantile-utility-comparison}}

\begin{proof}
    On the event in
    Proposition~\ref{prop:conservativequantile}, for every
    $\bvrho\in\Theta_{\bvrho}$ and $t\in[T]$,
    \begin{equation*}
        q_t\big(\balp_n^-(\bvrho)\big)
        \ge
        \htau_t\bigl(\bc(\bvrho;n);\cD^\ca\bigr)
        \ge
        q_t\big(\balp_n^+(\bvrho)\big).
    \end{equation*}
    Hence, by \eqref{eq:coin-threshold-representation}, for every
    $x\in\cX$ and $t\in[T]$,
    \begin{equation*}
        \cC_t^\coin\big(x;\balp_n^+(\bvrho)\big)
        \subseteq
        \hcC_t^\coin
        \big(x;\cD^\ca,\bc(\bvrho;n)\big)
        \subseteq
        \cC_t^\coin\big(x;\balp_n^-(\bvrho)\big).
    \end{equation*}
    The claimed inequalities now follow from
    \eqref{eq:monotone-utility}. Since the event above is uniform in
    $\bvrho$ and the set inclusions hold pointwise for every $x$, the
    conclusion holds simultaneously over
    $\bvrho\in\Theta_{\bvrho}$ and $x\in\cX$.
\end{proof}

\subsection{Proof of Lemma~\ref{lem:uniform-utility-classification}}\label{assec:proof-uniform-utility-class}

\begin{proof}[Proof of Lemma~\ref{lem:uniform-utility-classification}]
    Condition throughout on the independently trained score functions.
    Because $\cY=[M]$ is finite, each finite coordinate of $\btau$ can be
    approximated from above by a decreasing rational sequence.  For every
    $x$, all induced prediction sets then eventually agree with those at
    $\btau$; bounded convergence handles the population expectation.
    Thus, the supremum may equivalently be taken over the countable set
    $(\mathbb Q\cup\{-\infty,+\infty\})^T$, which also establishes its
    measurability.
    For $\btau\in\overline{\bbR}^T$, let
    \begin{equation*}
        g(x;\btau)
        =
        u\left(
            \cC_{1:T}(x;\btau)
        \right)
        -
        \frac{U_{\rm max}}{2}.
    \end{equation*}
    Centering the utility does not change its empirical-process deviation, and
    $
        |g(x;\btau)|
        \le
        U_{\rm max}/2
    $
    for every $x\in\cX$ and $\btau\in\overline{\bbR}^T$.

    Conditional on $\{X_1,\ldots,X_n\}$, for each $t\in[T]$,
    varying $\tau_t$ produces at most $nM+1$ distinct values of the array
    $\left(
        \bbI\{s_t(X_i,y)\le\tau_t\}
    \right)_{i\in[n],\,y\in[M]}$.
    Once these arrays are fixed for every $t\in[T]$, all the running
    intersections
    \begin{equation*}
        \cC_t(X_i;\btau)
        =
        \bigcap_{t'\le t}
        \big\{
            y:
            s_{t'}(X_i,y)\le\tau_{t'}
        \big\},
        \qquad
        i\in[n],\quad t\in[T],
    \end{equation*}
    are fixed. Consequently, as $\btau$ varies over
    $\overline{\bbR}^T$, the vectors
    \begin{equation*}
        \Big(
            g(X_1;\btau),
            \ldots,
            g(X_n;\btau)
        \Big)
    \end{equation*}
    take at most $(nM+1)^T$ distinct values.

    It follows that
    \begin{equation}
        \label{eq:expected-uniform-utility-classification}
        \begin{aligned}
        &\E\bigg[
            \sup_{\btau\in\overline{\bbR}^T}
            \bigg|
                \frac{1}{n}
                \sum_{i=1}^n
                u\left(
                    \cC_{1:T}(X_i;\btau)
                \right)
                -
                \E_X
                \left[
                    u\left(
                        \cC_{1:T}(X;\btau)
                    \right)
                \right]
            \bigg|
        \bigg]
        \\
        ={}&
        \E\bigg[
            \sup_{\btau\in\overline{\bbR}^T}
            \bigg|
                \frac{1}{n}
                \sum_{i=1}^n
                g(X_i;\btau)
                -
                \E_X\big[g(X;\btau)\big]
            \bigg|
        \bigg]
        \\
        \overset{(i)}{\le}{}&
        2
        \E\bigg[
            \sup_{\btau\in\overline{\bbR}^T}
            \bigg|
                \frac{1}{n}
                \sum_{i=1}^n
                \varsigma_i g(X_i;\btau)
            \bigg|
        \bigg]
        \\
        \overset{(ii)}{\le}{}&
        U_{\rm max}
        \sqrt{
            \frac{
                2\big\{
                    T\log(nM+1)+\log 2
                \big\}
            }{n}
        },
        \end{aligned}
    \end{equation}
    where $\varsigma_1,\ldots,\varsigma_n$ are independent Rademacher variables.
    Step~$(i)$ uses the symmetrization inequality (Proposition~4.11, \citealp{wainwright2019high}).
    Step~$(ii)$ follows from the standard finite-class bound for Rademacher averages; see the argument underlying Lemma~4.14 in \citet{wainwright2019high}.
    Indeed, adjoining the negatives of the sample traces produces at most $2(nM+1)^T$ vectors, each having Euclidean norm at most $U_{\rm max}\sqrt{n}/2$.

    The supremum in~\eqref{eq:expected-uniform-utility-classification} is a function of
    $X_1,\ldots,X_n$.
    Replacing one observation changes it by at most $U_{\rm max}/n$.
    The bounded-differences inequality (Theorem~6.2, \citealp{boucheron2013concentration}) therefore implies that, with probability at least $1-\varepsilon$,
    \begin{align*}
        &\sup_{\btau\in\overline{\bbR}^T}
        \left|
            \frac{1}{n}
            \sum_{i=1}^n
            u\left(
                \cC_{1:T}(X_i;\btau)
            \right)
            -
            \E_X
            \left[
                u\left(
                    \cC_{1:T}(X;\btau)
                \right)
            \right]
        \right|
        \\
        &\qquad\le
        U_{\rm max}
        \sqrt{
            \frac{
                2\big\{
                    T\log(nM+1)+\log 2
                \big\}
            }{n}
        }
        +
        U_{\rm max}
        \sqrt{
            \frac{\log(1/\varepsilon)}{2n}
        }.
    \end{align*}
    Set $\varepsilon=n^{-1}$.  The inequalities
    $nM+1\leq2nM$, $\log2\leq T\log(2nM)$, and
    $\log n\leq T\log(2nM)$ show that the preceding bound is at most
    the right-hand side of
    \eqref{eq:uniform-utility-classification}.
\end{proof}

\subsection{Proof of Lemma~\ref{lem:uniform-aggregated-set-size}}\label{assec:proof-uniform-aggregated-size}

\begin{proof}[Proof of Lemma~\ref{lem:uniform-aggregated-set-size}]
    Condition throughout on the independently trained score functions.
    Fix $y\in\cY$.
    Let
    $g(x;y,\btau)=\prod_{t=1}^T\bbI\{s_t(x,y)\leq\tau_t\}$.
    Conditional on $\{X_1,\ldots,X_n\}$, varying $\tau_t$ produces at
    most $n+1$ distinct vectors
    \begin{equation*}
        \bigl(
            \bbI\{s_t(X_1,y)\leq\tau_t\},
            \ldots,
            \bbI\{s_t(X_n,y)\leq\tau_t\}
        \bigr).
    \end{equation*}
    Hence, as $\btau$ varies over $\overline{\bbR}^T$, the vectors
    \begin{equation*}
        \bigl(g(X_1;y,\btau),\ldots,g(X_n;y,\btau)\bigr)
    \end{equation*}
    take at most $(n+1)^T$ distinct values.
    Therefore,
    \begin{equation}\label{eq:fixed-y-lower-orthant-bound}
        \begin{aligned}
            &\E\bigg[
                \sup_{\btau\in\overline{\bbR}^T}
                \bigg|
                    \frac{1}{n}\sum_{i=1}^{n}
                    \prod_{t=1}^T
                    \bbI\{s_t(X_i,y)\le\tau_t\}
                    -
                    \E_X\bigg[
                        \prod_{t=1}^T
                        \bbI\{s_t(X,y)\le\tau_t\}
                    \bigg]
                \bigg|
            \bigg]
            \\ = &
            \E\bigg[
                \sup_{\btau\in\overline{\bbR}^T}
                \bigg|
                    \frac{1}{n}\sum_{i=1}^{n}
                    g(X_i;y, \btau)
                    -
                    \E_X\Big[
                        g(X;y, \btau)
                    \Big]
                \bigg|
            \bigg]
            \\ \overset{(i)}{\le} &
            2 \,
            \E\bigg[
                \sup_{\btau\in\overline{\bbR}^T}
                \bigg|
                    \frac{1}{n}\sum_{i=1}^{n}
                    \varsigma_i g(X_i;y, \btau)
                \bigg|
            \bigg]
            \overset{(ii)}{\le}
            2\sqrt{
                \frac{
                    2\{T\log(n+1)+\log 2\}
                }{n}
            },
        \end{aligned}
    \end{equation}
    where $\varsigma_1,\ldots,\varsigma_n$ are independent Rademacher
    variables.  Step~$(i)$ is the symmetrization inequality
    (Proposition~4.11, \citealp{wainwright2019high}), and step~$(ii)$ is
    the finite-class bound for Rademacher averages (see the argument
    underlying Lemma~4.14 in \citealp{wainwright2019high}).

    We now integrate the fixed-$y$ bound.
    In both regression and classification, set size can be written as
    \begin{equation}
        \label{eq:aggregated-set-size-integral}
        \big|\cC_T(x;\btau)\big|
        =
        \int_{\cY}
        \prod_{t=1}^T
        \bbI\big\{s_t(x,y)\le\tau_t\big\}
        \,\mathrm{d} \mu(y)
        =
        \int_{\cY}
        g(x;y, \btau)
        \,\mathrm{d} \mu(y)
        ,
    \end{equation}
    where $\mu$ is Lebesgue measure in regression and counting measure in
    classification; in either case, $\mu(\cY)=M$.

    For completeness, the suprema below are measurable.  Indeed, under
    the standing joint-measurability assumptions, each finite coordinate
    of $\btau$ may be approximated from above by a decreasing sequence of
    rational numbers.  The corresponding indicators converge pointwise,
    and dominated convergence applies to both the $\mu$-integral and the
    expectation over $X$.  Thus, taking the supremum over
    $\overline{\bbR}^T$ is equivalent to taking it over the countable set
    $(\mathbb Q\cup\{-\infty,+\infty\})^T$.
    It follows that
    \begin{equation*}
        \begin{aligned}
            &\E\bigg[
                \sup_{\btau\in\overline{\bbR}^T}
                \bigg|
                    \frac{1}{n}\sum_{i=1}^{n}
                    \big|\cC_T(X_i;\btau)\big|
                    -
                    \E_X\big[\big|\cC_T(X;\btau)\big|\big]
                \bigg|
            \bigg]
            \\ = &
            \E\bigg[
                \sup_{\btau\in\overline{\bbR}^T}
                \bigg|
                    \frac{1}{n}\sum_{i=1}^{n}
                    \int_{\cY}
                    g(X_i;y, \btau)
                    \,\mathrm{d} \mu(y)
                    -
                    \E_X\Big[\int_{\cY}
                    g(X;y, \btau)
                    \,\mathrm{d} \mu(y)\Big]
                \bigg|
            \bigg]
            \\ = &
            \E\bigg[
                \sup_{\btau\in\overline{\bbR}^T}
                \bigg|
                    \int_{\cY}
                    \left\{
                        \frac{1}{n}\sum_{i=1}^{n}
                        g(X_i;y,\btau)
                        -
                        \E_X\big[g(X;y,\btau)\big]
                    \right\}
                    \,\mathrm{d}\mu(y)
                \bigg|
            \bigg].
        \end{aligned}
    \end{equation*}
    The first equality uses
    \eqref{eq:aggregated-set-size-integral}.  The second follows from
    linearity and Fubini's theorem, since the integrands are bounded and
    $\mu(\cY)<\infty$.  The triangle inequality and
    $\sup_{\btau}\int h_{\btau}\,\rmd\mu
      \leq\int\sup_{\btau}h_{\btau}\,\rmd\mu$
    then give the first inequality below; Tonelli's theorem interchanges
    the remaining expectation and integration:
    \begin{equation}
        \label{eq:expected-aggregated-set-size-bound}
        \begin{aligned}
            &\E\bigg[
                \sup_{\btau\in\overline{\bbR}^T}
                \bigg|
                    \frac{1}{n}\sum_{i=1}^{n}
                    \big|\cC_T(X_i;\btau)\big|
                    -
                    \E_X\big[\big|\cC_T(X;\btau)\big|\big]
                \bigg|
            \bigg]
            \\ \le &
            \int_{\cY}
            \E\left[
                \sup_{\btau\in\overline{\bbR}^T}
                \left|
                    \frac{1}{n}\sum_{i=1}^{n}
                    g(X_i;y, \btau)
                    -
                    \E_X\Big[
                        g(X;y, \btau)
                    \Big]
                \right|
            \right]
            \,\mathrm{d} \mu (y)
            \\ \le &
            2M\sqrt{
                \frac{
                    2\{T\log(n+1)+\log 2\}
                }{n}
            }.
        \end{aligned}
    \end{equation}
    The second inequality applies \eqref{eq:fixed-y-lower-orthant-bound} at each $y$ and uses that the measure of $\cY$ is $M$.

    The supremum in \eqref{eq:expected-aggregated-set-size-bound} is a function of $X_1,\ldots,X_{n}$.
    Replacing one observation changes it by at most $M/n$.
    The bounded-differences inequality (Theorem~6.2, \citealp{boucheron2013concentration}) thus implies that, with
    probability at least $1-\varepsilon$,
    \begin{equation*}
        \sup_{\btau\in\overline{\bbR}^T}
        \left|
            \frac{1}{n}\sum_{i=1}^{n}
            \big|\cC_T(X_i;\btau)\big|
            -
            \E_X\big[\big|\cC_T(X;\btau)\big|\big]
        \right|
        \le
        2M\sqrt{
            \frac{
                2\{T\log(n+1)+\log 2\}
            }{n}
        }
        +
        M\sqrt{\frac{\log(1/\varepsilon)}{2n}}.
    \end{equation*}
    Set $\varepsilon=n^{-1}$.  Since $n>15$ and $T\ge1$,
    \begin{equation*}
        \log(n+1)
        \le
        \frac{5}{4}\log n,
        \qquad
        \log 2
        \le
        \frac{1}{4}T\log n.
    \end{equation*}
    The preceding deviation bound is therefore at most the right-hand side
    of \eqref{eq:uniform-aggregated-set-size}.

\end{proof}

\subsection{Auxiliary lemmas}

\begin{lemma}[Bound for symmetric difference of finite intersections]
    \label{lem:boundsymmetricdifference}
    Let $\{A_t\}_{t=1}^K$ and $\{B_t\}_{t=1}^K$ be sequences of
    subsets of $\Omega$.  Define
    $I_k\coloneqq\bigcap_{t=1}^k A_t$ and
    $J_k\coloneqq\bigcap_{t=1}^k B_t$, with
    $I_0=J_0=\Omega$.  Then, for every $k\in[K]$,
    \begin{equation*}
        I_k \triangle J_k
        \subseteq
        \big( (A_k \triangle B_k)\cap I_{k-1}\big)
        \cup
        \big( B_k \cap (I_{k-1}\triangle J_{k-1})\big).
    \end{equation*}
\end{lemma}

\begin{proof}
    The elementary set inclusion
    \[
        (A\cap B)\triangle(C\cap D)
        \subseteq
        \bigl((B\triangle D)\cap A\bigr)
        \cup
        \bigl(D\cap(A\triangle C)\bigr)
    \]
    follows by considering whether a point in the symmetric difference
    belongs to $D$.  Apply it with
    $(A,B,C,D)=(I_{k-1},A_k,J_{k-1},B_k)$.
\end{proof}

\begin{lemma}[A comparison bound for intersections]
    \label{lem:boundintersectionmeasure}
    Let $(\Omega,\mathcal{F},\mu)$ be a finite measure space, i.e., $\mu(\Omega)<\infty$, and let $A,B,C\in\mathcal{F}$ be measurable sets.
    Then
    \begin{equation*}
        \bigl|\mu(A\cap B)-\mu(A\cap C)\bigr|
        \le
        \bigl|\mu(B)-\mu(C)\bigr|
        +
        \mu \big(A^{c} \cap (B\triangle C)\big).
    \end{equation*}
\end{lemma}

\begin{proof}
    Set
    \[
        a\coloneqq\mu(A\cap B)-\mu(A\cap C),
        \qquad
        b\coloneqq\mu(A^c\cap B)-\mu(A^c\cap C).
    \]
    Additivity over the partition $\{A,A^c\}$ gives
    $\mu(B)-\mu(C)=a+b$, and therefore
    \[
        |a|
        \leq |\mu(B)-\mu(C)|+|b|.
    \]
    Since
    \[
        |b|
        \leq
        \mu\bigl(A^c\cap(B\setminus C)\bigr)
        +
        \mu\bigl(A^c\cap(C\setminus B)\bigr)
        =
        \mu\bigl(A^c\cap(B\triangle C)\bigr),
    \]
    the claimed inequality follows.
\end{proof}

\end{document}